\documentclass[pra, onecolumn,superscriptaddress,nofootinbib]{revtex4}
\usepackage[a4paper, left=1in, right=1in, top=1in, bottom=1in]{geometry}

\usepackage{cmap} 
\usepackage[utf8]{inputenc}
\usepackage[english]{babel}
\usepackage[T1]{fontenc}
\usepackage{appendix}
\usepackage{braket}
\usepackage{url, amsfonts, tikz, physics, listings, xcolor, graphicx, float}
\usepackage[shortlabels]{enumitem}
\usepackage{dsfont}
\usepackage{amsmath, amssymb, amstext, amscd, amsthm, makeidx, graphicx, url, mathrsfs, mathtools, longdivision, polynom, bbm, complexity}
\usepackage{algorithm2e}
\RestyleAlgo{ruled}
\usepackage{threeparttable}
\usepackage{tablefootnote}

\definecolor{blueviolet}{rgb}{0.2, 0.2, 0.6}
\definecolor{webgreen}{rgb}{0,.5,0}
\definecolor{webbrown}{rgb}{.6,0,0}
\usepackage[pdftex,
  bookmarks=false,
  colorlinks=true, 
  urlcolor=webbrown,
  linkcolor=blueviolet, 
  citecolor=webgreen,
  pdfstartpage=1,
  pdfstartview={FitH},  
  bookmarksopen=false
  ]{hyperref}

\makeatletter

\newcommand\RedeclareMathOperator{%
  \@ifstar{\def\rmo@s{m}\rmo@redeclare}{\def\rmo@s{o}\rmo@redeclare}%
}
\newcommand\rmo@redeclare[2]{%
  \begingroup \escapechar\m@ne\xdef\@gtempa{{\string#1}}\endgroup
  \expandafter\@ifundefined\@gtempa
     {\@latex@error{\noexpand#1undefined}\@ehc}%
     \relax
  \expandafter\rmo@declmathop\rmo@s{#1}{#2}}
\newcommand\rmo@declmathop[3]{%
  \DeclareRobustCommand{#2}{\qopname\newmcodes@#1{#3}}%
}
\@onlypreamble\RedeclareMathOperator
\makeatother

\allowdisplaybreaks

\RedeclareMathOperator*{\E}{{\mathbb{E}}}

\numberwithin{equation}{section}

\newtheorem{theorem}{Theorem}
\newtheorem{prop}{Proposition}
\newtheorem{lemma}{Lemma}
\newtheorem{corollary}{Corollary}
\newtheorem{definition}{Definition}

\newtheorem{remark}{Remark}

\newcommand*\diff{\mathop{}\!\mathrm{d}}

\newcommand{\Gen}{\mathsf{Gen}}
\newcommand{\Enc}{\mathsf{Enc}}
\newcommand{\Dec}{\mathsf{Dec}}

\newcommand{\sk}{\mathsf{sk}}

\begin{document}

\title{Predicting adaptively chosen observables in quantum systems}

\author{Jerry Huang}
\email{jyhuang@alumni.caltech.edu}
\affiliation{Institute for Quantum Information and Matter, Caltech, Pasadena, CA, USA}

\author{Laura Lewis}
\affiliation{Institute for Quantum Information and Matter, Caltech, Pasadena, CA, USA}
\affiliation{Google Quantum AI, Venice, CA, USA}
\affiliation{University of Cambridge, Cambridge, UK}

\author{Hsin-Yuan Huang}
\affiliation{Institute for Quantum Information and Matter, Caltech, Pasadena, CA, USA}
\affiliation{Google Quantum AI, Venice, CA, USA}
\affiliation{Massachusetts Institute of Technology, Cambridge, MA, USA}

\author{John Preskill}
\affiliation{Institute for Quantum Information and Matter, Caltech, Pasadena, CA, USA}
\affiliation{AWS Center for Quantum Computing, Pasadena, CA, USA}

\date{\today}

\begin{abstract}
Recent advances have demonstrated that $\mathcal{O}(\log M)$ measurements suffice to predict $M$ properties of arbitrarily large quantum many-body systems. However, these remarkable findings assume that the properties to be predicted are chosen independently of the data. This assumption can be violated in practice, where scientists adaptively select properties after looking at previous predictions. This work investigates the adaptive setting for three classes of observables: local, Pauli, and bounded-Frobenius-norm observables.
We prove that $\Omega(\sqrt{M})$ samples of an arbitrarily large unknown quantum state are necessary to predict expectation values of $M$ adaptively chosen local and Pauli observables. We also present computationally-efficient algorithms that achieve this information-theoretic lower bound. In contrast, for bounded-Frobenius-norm observables, we devise an algorithm requiring only $\mathcal{O}(\log M)$ samples, independent of system size.
Our results highlight the potential pitfalls of adaptivity in analyzing data from quantum experiments and provide new algorithmic tools to safeguard against erroneous predictions in quantum experiments.
\end{abstract}

\maketitle

{\renewcommand\addcontentsline[3]{} \section{Introduction}}

Estimating properties of quantum systems using data collected from physical experiments is fundamental to the advancement of quantum information science. However, the required resources often scale unfavorably with system size $n$. For instance, full tomography of an $n$-qubit state requires a number of samples that grows exponentially in $n$ \cite{banaszek2013focus,blume2010optimal,gross2010quantum,hradil1997quantum,haah2017sample,o2016efficient}.
Recent results have addressed this scaling problem by demonstrating that many properties can be predicted from very few samples, even for very large $n$ \cite{aaronson2018shadow,huang2020predicting,paini2019approximate,cotler2019quantum,elben2022randomized}.
These sample-efficient protocols are known as shadow tomography.
In particular, the classical shadow formalism \cite{huang2020predicting} constructs a succinct classical description of an unknown state $\rho$ that accurately captures expectation values $\{o_i\}$ of a set of observables $\{O_i\}$,
\begin{equation}
    o_i(\rho)=\trace(O_i\rho), \quad 1\leq i \leq M,
\end{equation}
using only $\mathcal{O}(\log M)$ samples\footnote{The $\mathcal{O}(\log M)$ sample complexity applies when the observables have a shadow norm independent of $n$, which is true for many classes of physically-relevant observables.}. This achieves an exponential improvement in sample complexity compared to direct measurement of each observable, even when $M = \mathrm{poly}(n)$.

However, this efficiency relies on a critical assumption: the observables $\{O_i\}$ are chosen \textit{non-adaptively}. This means that the next observable $O_{i+1}$ cannot be selected based on the prediction outcomes for observables $O_1, O_2, \ldots, O_i$. While classical shadows can be acquired using randomized measurements \cite{elben2022randomized} with no knowledge of the observables, the prediction performance guarantee derived in \cite{huang2020predicting} assumes that the set of observables $\{ O_i \}$ is chosen independently of the measurement outcomes and the predicted values of other observables. In contrast, scientific research is often conducted adaptively: hypotheses are formulated from experimental data, tested using the same data, and then adaptively modified and retested. This process extends to quantum experiments, where one may learn certain properties of a quantum state and, inspired by the results, decide on additional properties to predict using the same data. However, there is currently no rigorous guarantee maintaining the $\mathrm{poly\,log}(M)$ sample complexity achieved by the classical shadow formalism when observables are chosen \emph{adaptively}.

This work addresses this gap by exploring the reusability of quantum data. We focus on answering the following central question:
\vspace{1em}
\begin{center}
    \emph{Can we still predict $M$ properties of an arbitrarily large quantum system\\from $\mathrm{poly\,log}(M)$ samples when the properties are chosen adaptively?}
\end{center}
Specifically, we seek to maintain the exponential reduction in sample complexity compared to direct observable measurement (as achieved by classical shadows in the non-adaptive setting) even when the number of properties $M = \mathrm{poly}(n)$. This necessitates avoiding any polynomial dependence of the sample complexity on the system size $n$. While existing protocols in shadow tomography \cite{aaronson2018shadow,aaronson2019gentle,badescu2020improved} achieve a polynomial reduction in sample complexity compared to direct measurement, with a sample complexity of $\mathcal{O}(n \log^2 n)$ for $M = \mathrm{poly}(n)$, they do not fully address our central question, as we aim for a superpolynomially reduced sample complexity of $\mathrm{poly\,log}(M)$.

Our work relates closely to adaptive data analysis in classical statistics \cite{DBLP:journals/corr/DworkFHPRR14,dwork2015generalization, bassily2015algorithmic,pmlr-v51-russo16,feldman2017generalization,feldman2018calibrating,jung2019new,ganesh2020privately,dagan2021boundednoise,ghazi2020avoiding,blanc2023subsampling}, where similar challenges arise in maintaining statistical validity under adaptive hypothesis testing \cite{ioannidis2005contradicted,ioannidis2005most,prinz2011believe,begley2012raise,gelman2016statistical}. In the classical setting, the best known sample complexity for predicting $M$ adaptively chosen queries is $\mathcal{O}(\sqrt{M})$ \cite{dagan2021boundednoise,ghazi2020avoiding,bassily2015algorithmic} when dimension dependence is not allowed, or $\mathcal{O}(\sqrt{d}\log M)$ samples where $d$ is the dataset dimension \cite{5670948,bassily2015algorithmic,steinke2015interactive}.
Despite significant progress in classical statistics, it is not immediately clear whether these results apply to physically relevant classes of observables in the quantum setting, and whether the classical shadow formalism is compromised by adaptively choosing the observables.

By building on techniques in shadow tomography and adaptive data analysis, we present both positive results and impossibility theorems for three physically-relevant classes of observables: local, Pauli, and bounded-Frobenius-norm observables. Estimation of these observables are key subroutines in algorithms for quantum chemistry \cite{peruzzo2014variational,crawford2020efficient,huggins2019efficient,izmaylov2019unitary,kandala2017hardware}, quantum field theory \cite{kokail2019self}, and linear algebra \cite{huang2019near}, as well as protocols for estimating fidelities \cite{flammia2011direct,da2011practical, huang2020predicting} and verifying entanglement \cite{guhne2009entanglement, elben2020mixed, elben2022randomized}. Our main findings are:
\begin{enumerate}
    \item For local and Pauli observables, we prove a lower bound of $\Omega(\sqrt{M})$ samples for any algorithm that accurately predicts adaptively chosen properties. We complement this with computationally-efficient algorithms achieving matching upper bounds.
    \item For bounded-Frobenius-norm observables, we devise a (computationally inefficient) algorithm that accurately predicts expectation values with a sample complexity of $\mathcal{O}(\log M)$, independent of system size. This improves upon existing shadow tomography protocols by exploiting the low-dimensional nature of these observables.
\end{enumerate}
The $\Omega(\sqrt{M})$ lower bound demonstrates that it is impossible to achieve the desired $\mathrm{poly} \log(M)$ sample complexity even for predicting local and Pauli observables. We note that this surprising result requires the system size to be exponential in $M$ for local observables and polynomial in $M$ for Pauli observables. Hence, this impossibility result only applies to systems containing many more qubits than the number of observables~$M$. For bounded-Frobenius-norm observables, the $\mathcal{O}(\log M)$ sample complexity improves upon the best-known scaling of $\mathcal{O}(n\log^2 M)$ \cite{badescu2020improved}.

Our work makes several key contributions. First, we demonstrate that a non-trivial extension of lower bounds in classical adaptive data analysis can yield lower bounds for adaptively-chosen local and Pauli observables in quantum systems. Second, by identifying and leveraging the connection between existing classical adaptive data analysis algorithms and classical shadows, we obtain upper bounds for local and Pauli observables in the quantum setting. Finally, we introduce a new (computationally inefficient) algorithm for predicting expectation values of bounded-Frobenius-norm observables that achieves logarithmic sample complexity in the adaptive setting.
In the process, we develop new algorithms for quantum threshold search~\cite{badescu2020improved} and shadow tomography which may be of independent interest.
In the following sections, we present our results in detail, including lower bounds, matching upper bounds, and our novel algorithm for bounded-Frobenius-norm observables. We conclude with a discussion of the broader implications of our work and potential directions for future research.

{\renewcommand\addcontentsline[3]{} \section{The Adaptive Model} \label{mainsec: preliminaries}} 

\begin{figure}[t]
    \centering
    \includegraphics[scale=0.6]{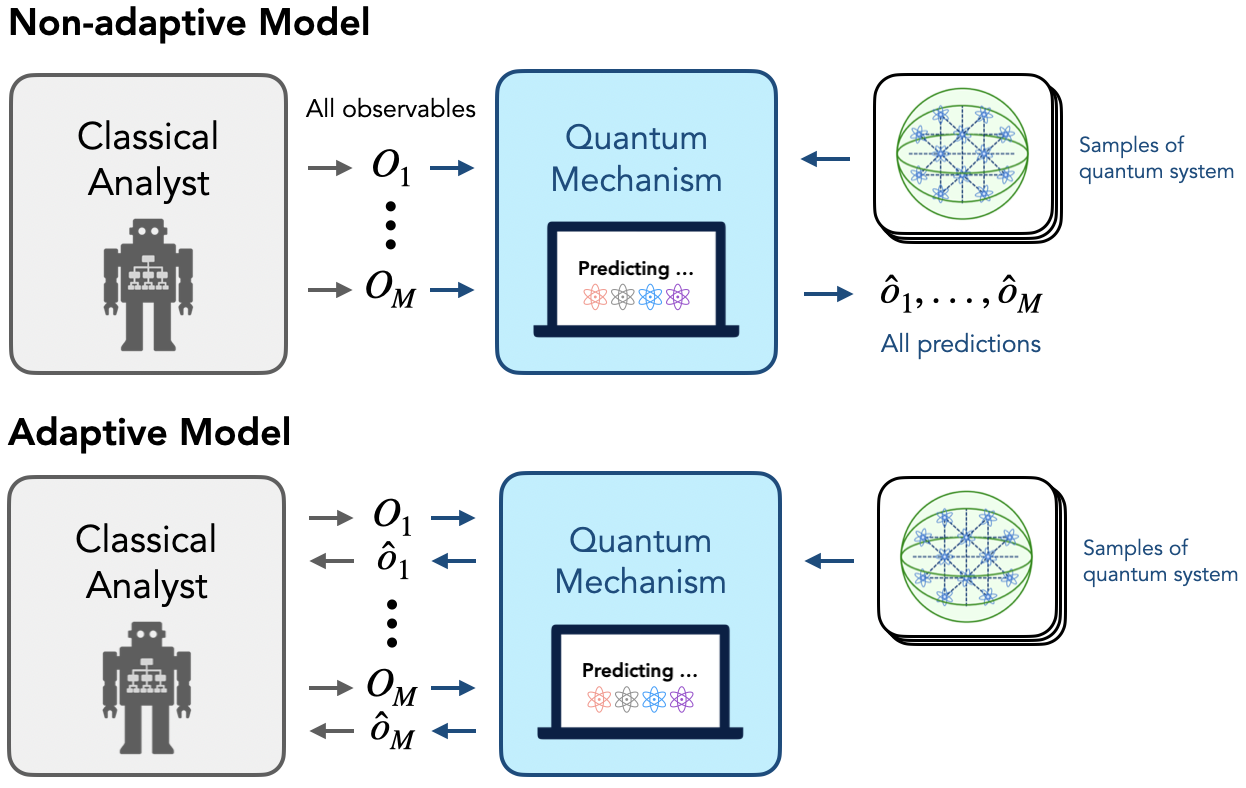}
    \caption{\textbf{Non-adaptive vs. adaptive model.} The analyst queries observables $O_1,...,O_M$ to the mechanism. The mechanism is given access to samples of a quantum state $\rho$ and can conduct quantum processing and measurements in order to output predictions $\hat{o_i}$ for the expectation values $\tr(O_i\rho)$. In the adaptive model (bottom), the analyst can condition on all prior observed predictions $\hat{o}_1,...,\hat{o}_{i-1}$ in order to select the next observable $O_i$. In contrast, a non-adaptive analyst (top) cannot observe prior predictions before querying observables.}
    \label{fig:adaptive-model}
\end{figure}
We introduce the model describing the adaptive interaction, with the full formal description provided in Appendix \ref{app: interaction model}. We conceptualize the predictive task as a game between an \emph{analyst} (e.g., a human scientist) and a \emph{mechanism} (e.g., a prediction algorithm), where only the mechanism has access to samples of the unknown quantum state $\rho$. This framework aligns with standard models in classical adaptive data analysis~\cite{DBLP:journals/corr/DworkFHPRR14}.
We consider the analyst to be classical, restricted to classical memory and processing, while the mechanism has access to quantum memory and has quantum processing capabilities.
The analyst-mechanism interaction unfolds over $M$ rounds, where in each round $i \in [M] \triangleq \{1,\ldots, M\}$: (1) the analyst selects an observable $O_i$ to query, and (2) the mechanism performs quantum processing on the samples and responds with an estimate $\hat{o}_i\approx\trace(O_i\rho)$. Crucially, the analyst can observe all prior responses $\hat{o}_1,\ldots,\hat{o}_{i-1}$ and utilize this information to determine the next observable $O_i$ to query.
This adaptive decision-making process characterizes the \emph{adaptive} nature of the interaction.
A visual representation of this interaction model is presented in Figure~\ref{fig:adaptive-model}.

{\renewcommand\addcontentsline[3]{} \subsection{Adaptive attack on classical shadows} \label{mainsec:attack}}

To illustrate the vulnerabilities introduced by adaptivity, we construct a simple adaptive attack executed by a classical analyst. This attack demonstrates how adversarial behavior can rapidly compromise the original classical shadows protocol~\cite{huang2020predicting}. For a review of the classical shadow formalism, we refer readers to Appendix~\ref{sec:non-adapt}.
Our proposed attack is inspired by the linear classification attack from \cite{DBLP:journals/corr/DworkFHPRR14}, which in turn draws from Freedman's Paradox~\cite{freedman1983note} --- a well-known example of standard statistical procedures yielding highly misleading results. Freedman's Paradox considers a regression problem where, given access to a large group of variables predicting a response variable, we (1) select a smaller group of variables correlated with the response variable and (2) fit a linear regression model on these variables. Surprisingly, standard statistical tests (e.g., an F-test) report the model as a good fit, even when no correlation exists between the input variables and the response variable.

A similar failure can occur in the quantum setting. Consider $N$ samples of a diagonal density matrix $\rho$ of system size $n=M+2^M$, where the first $M$ qubits represent variables and the last $2^M$ qubits correspond to all possible models that can be fit, each conditioned on a unique subset of the $M$ variables. The qubits are correlated such that the spins of the last $2^M$ qubits are determined by a majority vote of the spins of a subset of the first $M$ qubits.

We examine an analyst's adaptive attack querying only single-qubit observables from the set $\{Z_i:i \in [n]\}$, where $Z_i$ is the Pauli $Z$ observable acting solely on the $i$-th qubit. While the classical shadows protocol can accurately estimate expectation values of exponentially many observables in the number of samples $N$ for non-adaptive queries~\cite{huang2020predicting}, our adaptive attack causes the protocol to fail with high probability when the number of queries is only $M=\mathcal{O}(N)$, even for single-qubit observables.

The attack mimics Freedman's paradox in a two-step process:
\vspace{1em}
\begin{enumerate}
    \item Query $Z_i$ for each of the first $M$ qubits.
    \item Select a subset of qubits whose estimated expectation values are suitably biased towards $1$, and query $Z_j$ for the qubit correlated with this subset.
\end{enumerate}
This process yields a similar outcome to Freedman's paradox: the classical shadows estimate for $\trace(Z_j\rho)$ will likely be greater than $0.99$, even if the true value is $0$. Notably, only the last query is adaptive and is conditioned on the results of the first $M$ non-adaptive queries. We empirically demonstrate this phenomenon in Figure~\ref{fig:adaptive-attack}. Because the density matrix is diagonal and we only query $Z_i$ observables, we can simulate this protocol classically. The full description and proof of the attack are provided in Appendix~\ref{sec:attack}, with details of the numerical experiment in Appendix~\ref{sec:numerics}.

\begin{figure}[t]
    \centering
    \includegraphics[scale=0.8]{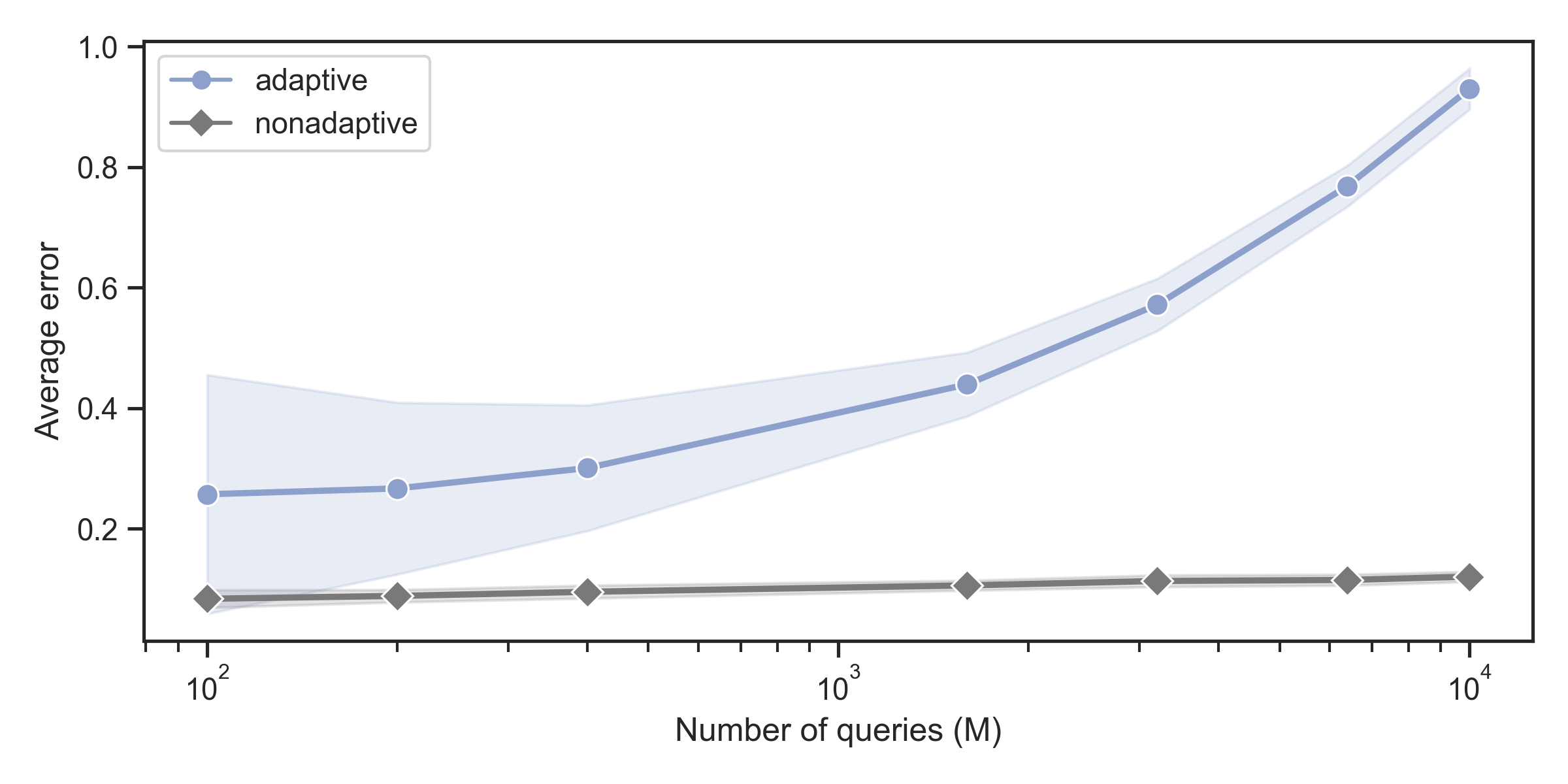}
    \caption{\textbf{The dangers of adaptivity.} Each point indicates the average error (averaged over $100$ independent runs) of the classical shadows protocol after being asked $M$ queries from the analyst for a fixed sample size of $N = 10000$. The shaded regions show the standard deviation over the independent runs. We see that after many adaptively chosen queries, the average error of classical shadows is much higher than for nonadaptive queries.}
    \label{fig:adaptive-attack}
\end{figure}

This example underscores the potentially detrimental effect of adaptivity on the accuracy of predictive algorithms: for single-qubit observable queries, a single adaptive query suffices to induce false predictions. Our goal is to prevent such failures and derive rigorous guarantees against false discovery in quantum experiments for specific classes of observables.

{\renewcommand\addcontentsline[3]{} \section{Main Results}}

Our primary focus is on the task of predicting properties $\tr(O_i\rho)$, for $i = 1,\dots, M$, where $\rho$ is an unknown $n$-qubit quantum state and $O_i$ is an observable. We aim to design a mechanism that can accurately predict these properties to within $\epsilon$ error using minimal samples/copies of $\rho$, even when the properties are chosen \textit{adaptively}. In the following subsections, we address this challenge for three distinct classes of observables: local, Pauli, and bounded-Frobenius-norm observables.

\vspace{2em}
{\renewcommand\addcontentsline[3]{} \subsection{Local observables}}

We define local observables as those acting on a constant number $k = \mathcal{O}(1)$ of qubits. Recall that the adaptive attack described in Section~\ref{mainsec:attack} can cause the classical shadows protocol to fail by querying only $M = \mathcal{O}(N)$ local observables, where $N$ is the number of samples of $\rho$. Our main result in this section demonstrates that $\mathcal{O}(N^2)$ local observables suffice to cause \textit{any} quantum mechanism to fail, given a sufficiently large system size $n$.
In other words, we prove an information-theoretic sample complexity lower bound of $N=\Omega(\sqrt{M})$ for predicting $M$ adaptively chosen local observables when the system size $n$ is exponential in $M$. We also show that this lower bound can be saturated by a computationally-efficient quantum mechanism. The formal statement is given as follows.

\vspace{1em}
\begin{theorem}[Local observables]
    \label{thm:local}
    For predicting expectation values $\trace(O_i\rho)$ of $M$ adaptively chosen local observables $O_1,\dots, O_M$, any quantum mechanism must use $N=\Omega(\min\{\sqrt{M},\log n\})$ samples. Moreover, there exists a computationally-efficient quantum mechanism using $N={\mathcal{O}}(\min\{\sqrt{M},\log n\})$ samples, running in time $\mathrm{poly}(N, n)$ per observable queried.
\end{theorem}

To prove the lower bound, we construct a density matrix $\rho$ and an analyst strategy for adaptively querying $M$ local observables, inspired by lower bounds in classical adaptive data analysis~\cite{hardt2014preventing, dwork2015preserving,steinke2015interactive} based on interactive fingerprinting codes~\cite{fiat2001dynamic}. The proof ideas are presented in Section~\ref{mainsec: proof ideas}, with detailed proofs in Appendix~\ref{sec:lb}.

The matching upper bound is achieved by applying classical adaptive data analysis algorithms~\cite{bassily2015algorithmic,feldman2017generalization} to classical shadow data~\cite{huang2020predicting}. We provide two computationally-efficient algorithms in Appendix~\ref{sec:classical shadow ub} and Appendix~\ref{app: acs}, both running in time $\mathrm{poly}(N,n)$ per observable. The first algorithm, specific to local observables, achieves a sample complexity of $\mathcal{O}(\sqrt{M})$. The second, more general algorithm applies to any class of observables and measurement primitive, achieving an upper bound of $\mathcal{O}(\sqrt{M}\log M \cdot \|O\|^2_\mathrm{shadow})$, where $\|\cdot \|_\mathrm{shadow}$ is the shadow norm dependent on the observables and measurement primitive.
For local observables, $\|O\|^2_\mathrm{shadow}$ is constant, so we achieve nearly the same $\tilde{\mathcal{O}}(\sqrt{M})$ upper bound. We expect this second algorithm to be of independent interest, as it applies to any observable.

{\renewcommand\addcontentsline[3]{} \subsection{Pauli observables}}

For Pauli observables $P \in \{I,X,Y,Z\}^{\otimes n}$, we demonstrate that any quantum mechanism requires at least $\Omega(\sqrt{M})$ samples, assuming a sufficiently large system size. Unlike local observables, the $\Omega(\sqrt{M})$ lower bound holds when the system size is polynomial in $M$.

\vspace{1em}
\begin{theorem}[Pauli observables]
    \label{thm:pauli}
    For predicting expectation values $\trace(P_i\rho)$ of $M$ adaptively chosen Pauli observables $P_1,\dots, P_M \in \{I, X, Y, Z\}^{\otimes n}$, any quantum mechanism must use $N=\Omega(\min\{\sqrt{M},n^{1/3}\})$ samples. Moreover, there exists a computationally-efficient quantum mechanism using $N=\mathcal{O}(\min\{\sqrt{M},n\})$ samples, running in time $\mathrm{poly}(N, n)$ per observable queried.
\end{theorem}

Similarly to the local observables case, the lower bound is obtained by constructing an adversarial choice of density matrix and analyst strategy, leveraging interactive fingerprinting codes~\cite{fiat2001dynamic}.
However, the particular construction differs from the case of predicting local observables in order to embed the interaction within a polynomial number of qubits. Proof ideas are presented in Section~\ref{mainsec: proof ideas}, with detailed proofs in Appendix~\ref{app: pauli}.

The $\mathcal{O}(\sqrt{M})$ upper bound is achieved by applying a computationally-efficient classical adaptive data analysis algorithm~\cite{bassily2015algorithmic} to the quantum machine learning (ML) algorithm from \cite{huang2021information}. The $\mathcal{O}(n)$ upper bound uses the non-adaptive quantum ML algorithm~\cite{huang2021information} to predict all $4^n$ Pauli observables.
As our main focus is on the $M$ scaling, we leave open the problem of closing the gap between $\mathcal{O}(n)$ and $\Omega(n^{1/3})$.

{\renewcommand\addcontentsline[3]{} \subsection{Bounded-Frobenius-norm observables}}

We define bounded-Frobenius-norm observables as those observables $O$ satisfying $\trace(O^2) \leq B$, where $B = \mathcal{O}(1)$.
In this case, we demonstrate that $M$ adaptively chosen observables can be predicted with only $\mathcal{O}(\log(M))$ samples.

\vspace{1em}
\begin{theorem}[Bounded-Frobenius-norm observables] \label{thm: bfo1}
    For any sequence of adaptively chosen observables $O_1,...,O_M$ satisfying $\trace(O_i^2) = \mathcal{O}(1)$ for all $i$, there exists a (computationally-inefficient) mechanism using $N=\Tilde{\mathcal{O}}(\log M / \epsilon^4)$ samples that can predict $\trace(O_1\rho),...,\trace(O_M\rho)$ to $\epsilon$ additive error.
\end{theorem}

Our algorithm builds upon existing shadow tomography protocols~\cite{aaronson2018shadow,badescu2020improved}, incorporating two key components: a threshold search subroutine\footnote{Threshold search, introduced as ``secret acceptor'' in~\cite{aaronson2016complexity}, aims to identify whether $\trace(O_j\rho) > \theta_j - \epsilon$ for some $j$ or $\tr(O_i\rho) \leq \theta_i$ for all $i$, given $\rho$, observables $O_1,\dots, O_M$, and thresholds $\theta_1,\dots, \theta_M$.} and a prediction subroutine. These components operate synergistically: the prediction subroutine provides estimates for the expectation values, while the threshold search subroutine validates these predictions against true values $\trace(O\rho)$. When predictions deviate significantly from true values, the prediction subroutine refines its model. We implement these components with novel algorithms (detailed in Appendix~\ref{app: acs}) that are applicable to any observables and provide robustness against adaptivity.

Our key insight is that predicting expectations of bounded-Frobenius-norm observables is equivalent to predicting an adaptively chosen observable after ``projecting'' onto an unknown low-dimensional subspace determined by the unknown state $\rho$. By learning to predict observables in this subspace, we achieve a system-size-independent $\mathcal{O}(\log M)$ sample complexity, improving upon the best-known adaptive shadow tomography protocol, which has a sample complexity of $\tilde{\mathcal{O}}(n\log^2 M/\epsilon^4)$~\cite{badescu2020improved}. Proof ideas are presented in Section \ref{mainsec: proof ideas}, with full proofs in Appendix \ref{app: bfo}.
For low-rank observables, we can further improve the sample complexity.

\begin{theorem}[Low-rank observables] \label{thm: bfo2}
    For any sequence of adaptively chosen observables $O_1,...,O_M$ with rank at most $R$, there exists a (computationally inefficient) mechanism using $N=\Tilde{\mathcal{O}}(R^2\log M / \epsilon^3)$ samples that can predict $\trace(O_1\rho),...,\trace(O_M\rho)$ to $\epsilon$ additive error.
\end{theorem}

\noindent The proof for this theorem is provided in Appendix~\ref{app: bfo}.

{\renewcommand\addcontentsline[3]{} \section{Proof Ideas} \label{mainsec: proof ideas}}

In this section, we describe the key ideas behind the proofs of our results.
The lower bound proofs for local and Pauli observables are provided in Appendices \ref{app: local} and \ref{app: pauli}. Both proofs leverage interactive fingerprinting codes~\cite{fiat2001dynamic} (see Appendix~\ref{sec:fingerprint} for an overview) to establish their guarantees, a technique also employed by \cite{steinke2015interactive} in classical adaptive data analysis. To prove upper bounds for bounded-Frobenius-norm observables, we first introduce a new threshold search algorithm and a shadow tomography algorithm based on classical shadows in Appendix~\ref{app: acs}. Subsequently, in Appendix~\ref{app: bfo}, we present our main algorithm, which integrates these two algorithms as the major subroutines, and prove that it achieves the desired $\log M$ upper bound.
We give an overview in the following sections.

\vspace{2em}
{\renewcommand\addcontentsline[3]{}
\subsection{Lower bounds for local and Pauli observables}}

We establish the lower bounds in Theorems~\ref{thm:local} and~\ref{thm:pauli} by reducing them to the guarantees of interactive fingerprinting codes (IFPCs)~\cite{fiat2001dynamic}. While the relationship between classical adaptive data analysis~\cite{hardt2014preventing, dwork2015preserving,steinke2015interactive} and fingerprinting codes is well-established, its extension to quantum adaptive data analysis remains largely unexplored, particularly when the queried observables are constrained to specific classes. Our primary technical contribution in this section is demonstrating that the requisite IFPC conditions can be satisfied within the quantum interaction model, even when the analyst is limited to querying only local and Pauli observables.

It is worth noting that our techniques can be adapted to derive new classical lower bounds for analogous problems in classical adaptive data analysis (see Corollaries~\ref{cor:local_queries} and~\ref{cor:parity_queries} in Appendices \ref{sec:lb} and \ref{app:pauli-lower}, respectively). This underscores the broader applicability of our approach beyond the quantum realm.

Interactive fingerprinting codes serve as a crucial component in establishing lower bounds for classical adaptive data analysis. Our proofs for Theorems~\ref{thm:local} and~\ref{thm:pauli} involve embedding IFPCs into a quantum interaction game with local and Pauli observables. To elucidate the underlying principles of IFPCs, we provide a concise overview (for a more comprehensive treatment, see Appendix~\ref{sec:fingerprint}).

Consider a scenario with $d$ users, some of whom form a colluding group creating illegal content. The primary objective of an IFPC is to embed watermarks into the content with sufficient sophistication to enable tracing of any illegally created material back to its source. IFPCs function within an interactive game involving an adversary $\mathcal{P}$ (representing the colluding group) and a fingerprinting code $\mathcal{F}$. In this setup, $\mathcal{F}$ distributes watermarks to all users, while $\mathcal{P}$ aims to generate illegal content without getting caught.
The game unfolds over multiple rounds. In each round, $\mathcal{F}$ examines the illegal content and identifies potential members of the colluding group. A fingerprinting code is deemed \emph{collusion resilient} if it can accurately identify colluding users while avoiding false accusations. The IFPC guarantee asserts the existence of a collusion resilient interactive fingerprinting code capable of adaptively tracing a colluding group of size $N$ within $\mathcal{O}(N^2)$ rounds.

Let us revisit the interaction model between an analyst and a mechanism, where the mechanism exclusively has access to $N$ copies of the unknown quantum state $\rho$, and the analyst queries observables for their expectation values. In this scenario, the mechanism strives to respond to the analyst's queries without divulging excessive information about the samples, while the analyst attempts to adaptively extract this information to force the mechanism into generating incorrect answers.
The parallel between this model and the IFPC setting can be drawn as follows:
\vspace{0.6em}
\begin{itemize}
    \item The analyst assumes the role of the fingerprinting code $\mathcal{F}$, selecting queries strategically.
    \item The mechanism corresponds to the adversary $\mathcal{P}$ in the IFPC game.
    \item The colluding users are encoded into the unknown quantum state $\rho$.
    \item The queried observables correspond to the watermarks in the fingerprinting code.
\end{itemize}
Provided that the conditions of the IFPC game are maintained (i.e., the adversary $\mathcal{P}$, and by extension the mechanism, only accesses information from non-accused members of the colluding group), the collusion resilient guarantees of the IFPC remain valid. Consequently, the mechanism can be compelled to provide inaccurate answers after $\mathcal{O}(N^2)$ queries. This directly aligns with our objective in proving the sample complexity lower bound. Therefore, our task reduces to constructing the density matrix $\rho$ and a strategy for the analyst in a manner that satisfies the conditions of the IFPC game.

The IFPC game requires the following two conditions to hold, which we refer to as the IFPC conditions.
\begin{enumerate}
    \item The fingerprinting code $\mathcal{F}$ (i.e., the analyst) must have the capability to query any watermarking assignment from the set $\{-1,1\}^d$, where each user is assigned either $-1$ or $1$.
    \item The adversary $\mathcal{P}$ (i.e., the mechanism) can only access the watermarks of non-accused users in the colluding group.
\end{enumerate}
To satisfy these conditions in our quantum setting, we construct $\rho$ as a mixture of $d=2000N$ computational basis states, each representing a user. The analyst's objective is to recover the computational basis states in the sample. This is achieved by measuring specific qubits in the $Z$ basis, which deterministically assigns each user a value of $-1$ or $1$, effectively serving as the watermark.
To fulfill the first IFPC condition, we consider the following. For the local observables case:
\begin{itemize}
    \item The analyst queries observables from the set $\{Z_i:i\in [n]\}$, where $Z_i$ is the Pauli $Z$ observable acting solely on the $i$th qubit.
    \item The number of qubits scales as $2^d$, with each qubit encoding one possible watermarking assignment from $\{-1,1\}^d$.
\end{itemize}
For the Pauli observables case:
\begin{itemize}
    \item The analyst queries observables from the set $Z_b \in \{I, Z\}^{\otimes n}$, where $b \in \{0,1\}^n$ indicates which qubits are acted on by the Pauli-$Z$ operator.
    \item The number of qubits scales as $\mathcal{O}(d)$, which significantly reduced the system size needed to encode all watermarking assignments.
\end{itemize}
To satisfy the second IFPC condition, we employ distinct strategies for local and Pauli observables.
For local observables, we implement randomization in the selection of $\rho$. This randomization conceals information from the mechanism, ensuring it only accesses data from non-accused members of the colluding group. For Pauli observables, we incorporate an encryption scheme. This encryption similarly restricts the mechanism's access to information from non-accused colluding group members.

For both cases, we allocate new qubits to encode watermarking assignments for each of the $M$ rounds. This crucial step ensures that information remains hidden across different rounds of the interaction. Consequently, we introduce an additional factor of $M$ in the number of qubits required for both the local and Pauli observables lower bounds.
These strategies collectively guarantee that the mechanism's access to information aligns with the IFPC game's constraints, thereby preserving the integrity of our quantum implementation of the IFPC framework.

With these conditions established, an analyst using the fingerprinting code $\mathcal{F}$ in the IFPC game inherits its guarantees, forcing any mechanism to fail rapidly. The interaction proceeds as follows:
\begin{enumerate}
    \item In round $j \in [M]$ of the interaction, the analyst queries observables $Z_x$\footnote{For local observables, $x \in [n]$, while for Pauli observables, $x \in \{0,1\}^n$.}, where $x$ corresponds to the watermarks sent by $\mathcal{F}$ to the adversary $\mathcal{P}$ in round $j$ of the IFPC game.
    \item Through adaptive querying of $Z_x$ observables, the analyst accumulates information, enabling the recovery of computational basis states in the sample provided to the mechanism.
    \item Upon gathering sufficient information about the sample, the analyst forces the mechanism to produce inaccurate responses. This strategy bears similarity to the adaptive attack in Section~\ref{mainsec:attack}.
\end{enumerate}
The IFPC guarantees ensure that the analyst can force any mechanism to respond inaccurately within $\mathcal{O}(N^2)$ queries. This directly yields our sample complexity lower bound of $N=\Omega(\sqrt{M})$, establishing a fundamental limit on the efficiency of quantum mechanisms.
The IFPC conditions also yields the requisite system size for our lower bounds. For local observables, the required number of qubits is $n = \Omega(M2^{\sqrt{M}})$, which scales exponentially in the number of observables queried $M$.
Meanwhile, for Pauli observables, the system size is much smaller, i.e., $n = \Omega(M^{3/2})$.

{\renewcommand\addcontentsline[3]{}
\subsection{Adaptive algorithms for classical shadows}}

We establish rigorous guarantees for new adaptivity-resistant algorithms using the classical shadow formalism~\cite{huang2020predicting} for shadow tomography and threshold search, which we black-box apply as subroutines to our main algorithm for predicting bounded-Frobenius-norm observables.
While quantum threshold search~\cite{badescu2020improved} and Aaronson's shadow tomography~\cite{aaronson2018shadow} may also be substituted to yield system-size-independent upper bounds, we achieve improved scaling by integrating classical shadows~\cite{huang2020predicting} with classical adaptive data analysis tools (see Remark~\ref{rmk: substitute}).
Moreover, these algorithms achieve sample complexity bounds that scale better with certain parameters, so we believe them to be of independent interest.

By combining classical shadows and the private multiplicative weights algorithm~\cite{5670948}, we derive a shadow tomography algorithm requiring
\begin{equation} \label{eq: st}
    N_{ST} = \Tilde{\mathcal{O}}\left(\frac{\sqrt{m}\cdot \log M}{\epsilon^3} \cdot \max_{i\in [M]} \|O_i\|^3_{\mathrm{shadow}}\right)
\end{equation}
samples. Here, $m$ is the number of bits needed for the classical shadow description and $\|O_i\|_\mathrm{shadow}$ is the shadow norm of observable $O_i$~\cite{huang2020predicting}. For $m=\mathcal{O}(n^2)$ (applicable to local and bounded-Frobenius-norm observables), this improves upon the best-known scaling~\cite{badescu2020improved} of $\tilde{\mathcal{O}}(n\log^2 M / \epsilon^4)$ by a factor of $\log M / \epsilon$, at the cost of the shadow norm term. Notably, the shadow norm is constant for bounded-Frobenius-norm observables.

We develop a new threshold search algorithm by incorporating the sparse vector algorithm~\cite{10.1145/1536414.1536467} from classical adaptive data analysis and introducing a preprocessing step for classical shadows. We also employ a new variant of classical shadows derived from uniform POVM measurements~\cite{kueng2019classical}, which offers enhanced concentration properties, leading to improved $\epsilon$ scaling.

The main new feature of our threshold search algorithm is its ``strong composition'' property, enabling sub-linear sample complexity growth when adaptively combining multiple instances of threshold search. This contrasts with previous quantum threshold search algorithms, such as~\cite{badescu2020improved}, which compose linearly. In those algorithms, each instance runs until an above-threshold event occurs, necessitating fresh copies of $\rho$ for each new instance. Consequently, for a budget of $\ell$ above-threshold events (termed "mistakes" in~\cite{badescu2020improved}), the sample complexity includes a linear factor of $\ell$.

In contrast, classical adaptive data analysis algorithms, particularly the sparse vector algorithm, leverage differential privacy to achieve a \textit{quadratic} improvement in composing multiple instances, introducing only a factor of $\sqrt{\ell}$. By utilizing classical shadows, our threshold search algorithm inherits this property, requiring
\begin{equation} \label{eq: ts}
    N_{TS} = \Tilde{\mathcal{O}}\left(\frac{\sqrt{\ell}\log M}{\epsilon^2} \cdot \max_i \sqrt{\tr(O_i^2)}\right)
\end{equation}
samples. This improves over the best-known quantum threshold search algorithm~\cite{badescu2020improved}, which has a sample complexity of $\mathcal{O}(\ell\log^2 M/ \epsilon^2)$. Our approach reduces the sample complexity by a factor of $\sqrt{\ell}\log M$, albeit at the cost of introducing a bound on the Frobenius-norm.

\vspace{2em}
{\renewcommand\addcontentsline[3]{}
\subsection{Upper bound for bounded-Frobenius-norm observables}}

We leverage our new shadow tomography and threshold search algorithms to design an improved algorithm for predicting bounded-Frobenius-norm observables. To build intuition, we first consider the case of single-rank observables before extending to bounded-Frobenius-norm observables. We adopt the terminology of ``student'' and ``teacher'' from \cite{badescu2020improved} to describe the two main components of the mechanism. For any query $\ketbra{\psi}$, the student predicts $\hat{o}$ for $\bra{\psi}\rho\ket{\psi}$, while the teacher verifies this prediction. If $\hat{o}$ is $\epsilon$-close to $\bra{\psi}\rho\ket{\psi}$, the teacher outputs ``pass''; otherwise, it declares a ``mistake'' and provides the correct value.

In our algorithm, the student and teacher use the shadow tomography and threshold search algorithms, respectively.
This is different from Aaronson's shadow tomography protocol~\cite{aaronson2018shadow}, which uses an online learning algorithm\footnote{Online learning refers to a setting where one is tasked to answer a series of (potentially adversarial) requests in a sequential fashion while minimizing error/loss.
This is similar to our setting, where we must predict expectation values for a sequence of observables.}~\cite{aaronson2018online} to implement the student. 
In Aaronson's protocol, the student learns a full $2^n\times 2^n$ representation of the $n$-qubit quantum state $\rho$, which introduces the dimension dependence in its sample complexity.
Our approach replaces the student with an algorithm that learns a low-dimensional representation of $\rho$ via iterative dimension expansion.
Here, the student only needs to run a low-dimensional instance of shadow tomography to learn this low-dimensional representation, avoiding the dimension dependence.
This is sufficient for predicting expectation values of bounded-Frobenius-norm (and low rank) observables and enables learning with few samples.

We make two key observations:
\begin{enumerate}
    \item Given any quantum state $\rho$, there exists a low-dimensional subspace $S$ such that for all pure state $\ket{\psi}$ and $\ket{\psi_S}$ (the projection of $\ket{\psi}$ to $S$),
    \begin{equation}
        |\bra{\psi}\rho\ket{\psi}-\bra{\psi_S}\rho\ket{\psi_S}|\leq \epsilon.
    \end{equation}
    \item Given $S$, there is a low-dimensional quantum state $\hat{\rho}$ such that 
    \begin{equation}
        \bra{\psi_S}\hat{\rho}\ket{\psi_S}=\bra{\psi_S}\rho\ket{\psi_S}.
    \end{equation}
    Note that $\hat{\rho}$ can indeed be constructed to satisfy the required properties of a quantum state (i.e., Hermitian, positive semi-definite, and unit trace).
\end{enumerate}

In light of these two key observations, our algorithm tasks the student with \textit{simultaneously} learning (1) the subspace $S$ and (2) an approximation of the low-dimensional quantum state $\hat{\rho}$ (which is used to predict $\bra{\psi_S}\hat{\rho}\ket{\psi_S}$).
At a high level, the student runs a low-dimensional instance of shadow tomography (using our new adaptive shadow tomography algorithm from the previous section) in order to learn $\hat{\rho}$ and produce predictions. At the same time, the student maintains a running guess of $S$. Every time the teacher declares "mistake" on the student's prediction, the student will apply a correction on $S$. Notably, the teacher only plays an \textit{indirect} role in the learning of $\hat{\rho}$ (in the sense that a "mistake" declaration may not necessarily cause the representation of $\hat{\rho}$ to update)\footnote{To be precise, we are running a (low-dimensional) shadow tomography algorithm as a subroutine within our overall shadow tomography algorithm for bounded Frobenius-norm observables. The low-dimensional subroutine itself contains a teacher and student iteratively interacting to learn $\hat{\rho}$. So one can say that there are two teachers within our algorithm!}. This is another key difference with Aaronson's protocol, where the teacher's main role is to "teach" the student the quantum state representation, whereas in our algorithm, the student is taught the low-dimensional subspace representation.

We outline the steps in the algorithm below:
\begin{enumerate}
    \item For each received query $\ket{\psi}$, the student projects $\ket{\psi}$ into $S$ and uses the new shadow tomography protocol to predict $\hat{o}=\bra{\psi_S}\hat{\rho}\ket{\psi_S}$.
    \item The student asks the teacher (the new threshold search algorithm) to check if $|\hat{o}-\bra{\psi}\rho\ket{\psi}|<\epsilon$.
    \item If the teacher declares a mistake, the teacher will produce the correct expectation value $o'$ using classical shadows.
    \item In the case that the teacher declares a mistake, the student will apply a correction by adding the query $\ket{\psi}$ to $S$, where the new subspace is now defined as $\mathrm{span}(S \cup \{\ket{\psi}\})$.
\end{enumerate}
The key result to showing that the sample complexity is dimension-independent is the following lemma.

\begin{lemma} [Mistake bound]
    A student algorithm that correctly predicts $\hat{o}=\bra{\psi_S}\rho\ket{\psi_S}$ will make at most $\mathcal{O}(1/\epsilon^2)$ mistakes.
\end{lemma}

Using this lemma, one can show that since the number of mistakes is not very large, the threshold search algorithm will not need to make too many corrections (i.e., the budget parameter $\ell$ is not too large) and can thus be implemented sample-efficiently. 
Together with the sample complexity bounds in Equations~\eqref{eq: st} and \eqref{eq: ts} this yields that
\begin{equation}
    N = \Tilde{\mathcal{O}}\left(\frac{\log M}{\epsilon^3}\right)
\end{equation}
samples are sufficient to predict single rank observables.

To extend to bounded-Frobenius-norm and low-rank observables, we only need to make the following modifications: (1) define a processing step that ``projects'' the queried observable $O$ to $S$, and (2) when a mistake is made, add every eigenstate of $O$ with a sufficiently large eigenvalue (in magnitude) to $S$. In the case of observables with Frobenius norm $\Tr(O^2) \leq B$, the number of mistakes is $\mathcal{O}(B^2/\epsilon^4)$, yielding the corresponding increase in the sample complexity stated in Theorem~\ref{thm: bfo1}. For low-rank observables, the number of mistakes is only $\mathcal{O}(R^2/\epsilon^2)$, in which case we maintain the same $\epsilon$ dependence as for single rank observables.

{\renewcommand\addcontentsline[3]{} \section{Outlook}}

The classical shadow formalism \cite{huang2020predicting} provides an experimentally feasible alternative to full quantum state tomography that allows one to predict many properties of an unknown quantum state from very few samples of the state.
However, its rigorous guarantees require that these properties are chosen non-adaptively, whereas adaptivity is a natural component of scientific research.
This work makes progress in understanding adaptivity in analyzing data from quantum experiments.
We show that obtaining similar sample complexity guarantees as the classical shadow formalism in the adaptive setting is difficult even in highly restricted cases.
For local and Pauli observables, we prove that it is not possible in general to prevent an adaptive attack using $\mathcal{O}(\log M)$ samples, where $M$ is the number of queried observables. Moreover, this holds even when the system size is only polynomial in $M$.
However, in the case of bounded-Frobenius-norm observables, we show that it becomes possible to predict many properties with few samples, maintaining the same scaling as in the classical shadow formalism with respect to $M$.
We remark that other works such as \cite{badescu2020improved, aaronson2019gentle} have previously suggested a connection between adaptive data analysis and shadow tomography.
Our work provides extensive evidence of this connection, demonstrating that results from adaptive data analysis can be used to prove new sample complexity and impossibility theorems for shadow tomography, and vice versa.

Despite the progress made in this work, many questions remain open.
For local and Pauli observables we showed that if the system size is large enough, then no algorithm can prevent an adaptive attack.
What are some other families of observables, where adaptive queries can lead to erroneous predictions?
While we were able to design a sample-efficient algorithm for predicting adaptively chosen observables with bounded-Frobenius-norm, our algorithm is computationally inefficient.
The main bottleneck in our algorithm is a projection onto a low-dimensional subspace.
Are there computationally efficient algorithms for predicting bounded-Frobenius-norm observables with a sample complexity of $\mathrm{poly log}(M)$?
Or is there a computational lower bound that cannot be surpassed?
We remark that there are well known computational hardness conjectures in classical adaptive data analysis, where any algorithm with a $N = \mathrm{poly\,log}(M)$ sample complexity is expected to be computationally inefficient.
However, bounded-Frobenius-norm observables are sufficiently structured such that assuming these computational hardness conjectures does not yet rule out the presence of computationally efficient quantum algorithms.
We also remark that if one allows system size dependence $N = \mathrm{poly}(n, \log M)$, then the algorithms in \cite{brandao2019sdp,van_Apeldoorn_2020} achieve a computationally efficient algorithm for low-rank observables.

\vspace{2em}
{\renewcommand\addcontentsline[3]{} \section*{Acknowledgments}}
The authors thank Yu Tong for valuable and inspiring discussions.
The authors also thank Ruohan Shen, Haimeng Zhao, Mehdi Soleimanifar, Tai-Hsuan Yang, Yiyi Cai, Charles Cao, Nadine Meister, and Chris Pattison for insightful comments and feedback.
JH is supported by a Caltech Summer Undergraduate Fellowship and a Saul and Joan Cogen Memorial SURF fund.
LL is supported by a Mellon Mays Undergraduate Fellowship and a Marshall Scholarship.
HH was supported by a Google PhD fellowship and a MediaTek Research Young Scholarship.
HH acknowledges the visiting associate position at the Massachusetts Institute of Technology.
JP acknowledges support from the U.S. Department of Energy Office of Science, Office of Advanced Scientific Computing Research (DE-NA0003525, DE-SC0020290), the U.S. Department of Energy Office of Science, National Quantum Information Science Research Centers, Quantum Systems Accelerator, and the National Science Foundation (PHY-1733907).
This work was done (in part) while a subset of the authors visited the Simons Institute for the Theory of Computing.
The Institute for Quantum Information and Matter is an NSF Physics Frontiers Center.


\bibliographystyle{unsrt}
{\renewcommand\addcontentsline[3]{} \bibliography{refs}}

\newpage
\appendix
\appendixpage

\tableofcontents

\section{Preliminaries}
\label{app: interaction model}

In Appendix~\ref{sec:int-mod}, we formally define the interaction models used in classical adaptive data analysis (Appendix~\ref{sec:stat-query}) and our setting of quantum adaptive data analysis (Appendix~\ref{sec:quantum-interact}).

In Appendix~\ref{sec:non-adapt}, we review several previous non-adaptive methods for predicting expectation values with favorable sample complexity.
Specifically, in Appendix~\ref{sec:classical-shadow}, we review the classical shadow formalism~\cite{huang2020predicting}.
Classical shadows are present in many of our algorithms to achieve sample complexity upper bounds, so it is helpful to understand this formalism.
We also review an algorithm~\cite{huang2021information} for predicting expectation values of non-adaptively chosen Pauli observables in Appendix~\ref{sec:nonadapt-pauli}, which we adapt in one of our sample complexity upper bounds.
In Appendix~\ref{sec:uniform-povm}, we also describe a different version of classical shadows in which the shadow is created via independent measurements of the uniform POVM~\cite{kueng2019classical} rather than the random measurement scheme in~\cite{huang2020predicting}.

In addition, we review tools that will be useful throughout our proofs.
In Appendix~\ref{sec:fingerprint}, we describe fingerprinting codes~\cite{705568}, which are utilized heavily in our proofs of the lower bounds.
Another tool used in our lower bounds is private-key cryptography, which we review in Appendix~\ref{app:crypto}.

Throughout the appendices, we use the notation $[m]$ for some $m \in \mathbb{N}$ to denote $[m] \triangleq \{1,\dots, m\}$.
We also assume that the spectral norm of the observables we want to predict is bounded by one: $\norm{O}_\infty \leq 1$.

\subsection{Interaction Models}
\label{sec:int-mod}
\subsubsection{Statistical Query Model}
\label{sec:stat-query}
We first introduce a classical model common in the adaptive data analysis literature, which we will refer to as the adaptive statistical query model~\cite{dwork2015preserving}.
In the next section, we will establish an analogue in the quantum setting.
Suppose there is an unknown distribution $\mathcal{D}$ over a data universe $\mathcal{X}$ with dimension $n$. The goal is to design an algorithm to answer \textit{statistical queries}~\cite{kearns1998efficient,feldman2017statistical}, which are specified by a function $q:\mathcal{X}\rightarrow [0,1]$. The true value of the query $q$ is denoted as \begin{equation}
    q(\mathcal{D})\triangleq \E_{x \sim \mathcal{D}}[q(x)].
\end{equation}
We would like to estimate the true expectation value $q(\mathcal{D})$ as accurately as possible, but we typically do not have direct access to $\mathcal{D}$ and are instead granted access to a sample $\mathcal{S}=(x_1,...,x_N) \sim \mathcal{D}^N$, i.e., each $x_i \in \mathcal{S}$ is sampled i.i.d. from $\mathcal{D}$. The most straightforward way to estimate $q(\mathcal{D})$ is to compute the empirical mean $q(\mathcal{S}) = \frac{1}{N}\sum_{x_i \in \mathcal{S}} q(x_i)$.
Note that statistical queries must satisfy
\begin{equation}
    \E_{S \sim \mathcal{D}^N} [q(\mathcal{S})] = q(\mathcal{D}).
\end{equation}
However, an algorithm simply outputting $q(\mathcal{S})$ may fail after only a few queries if the queries are selected adaptively (e.g. see the feature selection attack in \cite{dwork2015generalization} and Lecture 3 of \cite{roth2017adaptive}). In general, we will need to design a \textit{mechanism} that outputs a response based on $\mathcal{S}$ to protect against adaptivity.

Formally, the interaction can be described as a game between an analyst $\mathcal{A}$ and mechanism $\mathcal{M}$. In general, the analyst and mechanism do not have access to the exact true distribution $\mathcal{D}$, and only the mechanism has access to the sample $\mathcal{S}$. $\mathcal{A}$ and $\mathcal{M}$ interact over $k$ rounds, where in each round $j$, (1) $\mathcal{A}$ asks query $q^j$ and (2) $\mathcal{M}$ uses $\mathcal{S}$ to generate a response $a^j\in [0,1]$. Here, the analyst's access to the sample $\mathcal{S}$ in the form of answers to statistical queries is the same as that of statistical query learning~\cite{kearns1998efficient}.
However, in this case, $\mathcal{A}$ and $\mathcal{M}$ can depend on prior history. Specifically, the analyst can choose queries \textit{adaptively} based on the transcript of prior interactions $\mathcal{T}^j=(q^1,a^1,...,q^{j-1},a^{j-1})$.
This is in contrast to the typical setting in which the queries $q^j$ are specified beforehand.
Again, the goal is to design a mechanism that can answer these adaptive queries correctly, where we characterize its accuracy as follows.

\begin{definition}
    A mechanism $\mathcal{M}$ is \emph{$(\epsilon,\delta)$-accurate} for $k$ queries $q^1,...,q^k$ if for every analyst $\mathcal{A}$ and distribution $\mathcal{D}$,\begin{equation}
        \Pr[\max_{j\in [k]}|q^j(\mathcal{D})-a^j|\leq \epsilon]\geq 1-\delta
    \end{equation}
    where the probability is over the randomness of $\mathcal{M}$, $\mathcal{A}$, and the dataset $S\sim \mathcal{D}^N$. Here, $a^j$ denotes the mechanism $\mathcal{M}$'s response to the query $q^j$.
\end{definition}

\subsubsection{Quantum Interaction Model}
\label{sec:quantum-interact}
With this motivation from classical adaptive data analysis, we now discuss the interaction model in the quantum setting.
Consider a game between a classical analyst $\mathcal{A}$ and a quantum mechanism $\mathcal{M}$. The interaction can be described as follows. Suppose $\rho$ is an $n$-qubit unknown quantum state and that $\mathcal{M}$ has access to $N$ copies $\rho^{\otimes N}$. First, $\mathcal{M}$ applies an initial unitary $U_0$ to $\rho^{\otimes N}$ and a set of ancilla qubits:
\begin{equation}
    \mathcal{W}^0(\rho^{\otimes N})=U_0(\rho^{\otimes N} \otimes \ketbra{0^k})U_0^\dag,
\end{equation}
where $\mathcal{W}^j(\rho^{\otimes N})$ denotes the mechanism's quantum state at round $j$ given access to the samples $\rho^{\otimes N}$ stored in quantum memory.
This executes any preprocessing on the input data.
After this initial preprocessing, the analyst and the mechanism will alternate turns. The analyst adaptively generates a classical query $q^j$ at round $j$ (e.g., a classical description of an observable $O$).
Sometimes, we will also denote the query as $o^j$ to align with the notation of~\cite{huang2020predicting}.
Moreover, we assume throughout the work that the observable $O$ corresponding to this query has $\norm{O}_\infty \leq 1$, an assumption that is also present in~\cite{huang2020predicting}.
This query can be viewed as an additional input in the mechanism's quantum state.
Then, the mechanism responds to the query by performing a unitary on its current state, measuring some of the qubits, and returning the classical output $a^j$ to the analyst. 
This interaction is described at a high level in Algorithm~\ref{alg:interact}.
More formally, let $\mathcal{T}^j$ be the transcript of prior interactions up to round $j$, i.e. $\mathcal{T}^j=(q^1,a^1,...,q^{j-1},a^{j-1})$. For each round $j=1,...,M$,
\begin{enumerate}
    \item $\mathcal{A}$ generates query $q^j$ based on $\mathcal{T}^j$ and gives it to $\mathcal{M}$.
    \item $\mathcal{M}$ takes $q^j$ as input and applies a unitary $U_j$ to its current quantum state $\mathcal{W}^{j-1}$ before performing a partial measurement $\Pi_j=\ketbra{a^j} \otimes I$. Here, the mechanism only measures some subset of the qubits so that the measurement operator is identity on the remaining unmeasured qubits.
    The classical outcome $a^j$ from this measurement is observed and given to $\mathcal{A}$ with probability 
    \begin{equation}
        \tr(\Pi_j(U_j(\mathcal{W}^{j-1}\otimes \ketbra{q^j})U_j^\dag)).
    \end{equation}
    The mechanism's post-measurement state is
    \begin{equation}
        \mathcal{W}^j(\rho^{\otimes N})\triangleq\frac{\Pi_jU_j(\mathcal{W}^{j-1}\otimes \ketbra{q^j})U_j^\dag\Pi_j^\dag}{\tr(\Pi_j(U_j(\mathcal{W}^{j-1}\otimes \ketbra{q^j})U_j^\dag))}.
    \end{equation}
\end{enumerate}
We now discuss what it means for a mechanism to be robust against any adaptive analyst (or for an analyst to force the mechanism to fail). We first note that for any quantum mechanism $\mathcal{M}$ and classical analyst $\mathcal{A}$, the interaction can be encoded in the following state:
\begin{equation}
\sigma\triangleq\sum_{\substack{q^1,...,q^M\\a^1,...,a^M}}p_\mathcal{A}(q_1|\mathcal{T}^1)...p_\mathcal{A}(q_M|\mathcal{T}^M)\ketbra{\mathcal{T}^{M+1}},
\end{equation}
where $p_{\mathcal{A}}(q_j|T^j)$ denotes the probability that the analyst $\mathcal{A}$ selects query $q^j$ conditioned on the past history of interactions. We can define an observable $O_\rho$ dependent on $\rho$ that encodes a mechanism's error in a given interaction, i.e. \begin{equation}
    \trace(O_\rho\ketbra{\mathcal{T}^{M+1}})=\max_{j\in [M]} |q_j(\rho)-a^j|
\end{equation} is the mechanism's maximum error under transcript $\mathcal{T}^{M+1}$. So the expected error of the mechanism is $\trace(O_\rho \sigma)$. Denote
\begin{equation}
    P_{\mathcal{T}^{M+1}}^\mathcal{A}\triangleq\left(\prod_{j=1}^M p_\mathcal{A}(q_j|\mathcal{T}^j)\right)
\end{equation}
as the probability of transcript $\mathcal{T}^{M+1}$ occurring. We say that a mechanism can always answer $M$ adaptive queries to $\epsilon$ error with probability $1-\delta$ if for any analyst $\mathcal{A}$,
\begin{equation}
    \sum_{\mathcal{T}: \trace(O_\rho \ketbra{\mathcal{T}})\leq\epsilon} P_{\mathcal{T}}^\mathcal{A} \geq 1-\delta.
\end{equation}
\begin{algorithm}
\caption{Quantum Interaction Model}\label{alg:interact}
Let $\rho$ be an $n$-qubit density matrix and $M$ be the number of properties to be tested.\\
Mechanism obtains samples $\rho^{\otimes N}$.\\
\For{$j=1$ to $M(N)$}{
Let $\mathcal{T}^j$ be the transcript of prior interactions $\mathcal{T}^j=(q^1,a^1,...,q^{j-1},a^{j-1})$.\\
Analyst $\mathcal{A}(\mathcal{T}^j$) generates query $q^j$ corresponding to a property to be tested such that $q^j(\rho)\in [-1,1]$.\\
Mechanism $\mathcal{M}(\mathcal{T}^j,q^j,\rho^{\otimes N})$ generates response $a^j\in [-1,1]$.
}
\end{algorithm}

\subsection{Review of Previous Non-Adaptive Methods}
\label{sec:non-adapt}

In this section, we provide a brief overview of previous work~\cite{huang2020predicting,huang2021information,kueng2019classical} that can be used to predict expectation values of non-adaptively chosen observables.
In particular, we review the classical shadow formalism (which can predict expectation values of, e.g., local observables and observables with bounded Frobenius norm) in Appendix~\ref{sec:classical-shadow}, an algorithm for predicting expectation values of Pauli observables in Appendix~\ref{sec:nonadapt-pauli}, and a version of the classical shadow formalism using measurements of the uniform POVM in Appendix~\ref{sec:uniform-povm}.
As we discuss later, all of these methods implicitly assume that the observables are chosen non-adaptively.
We review them here because we build on them in our new algorithms.

\subsubsection{Classical Shadow Formalism}
\label{sec:classical-shadow}

In this section, we briefly review the classical shadow formalism with an emphasis on aspects that are important in our proofs.
For a more detailed presentation, we direct the reader to~\cite{huang2020predicting}.
We note that there exist several other techniques for constructing efficient classical representations of quantum systems to predict properties~\cite{chen2020robust,cotler2020quantum,evans2019scalable,huang2021efficient,koh2020classical,paini2019approximate}, but we focus on the classical shadow formalism in this work.

Recall that the goal of the classical shadow formalism is to construct a succinct classical representation of an unknown $n$-qubit quantum state $\rho$ that suffices to predict properties $\tr(O_i \rho)$ for observables $O_1,\dots, O_M$.
One can construct this classical representation, called the \emph{classical shadow}, by applying a unitary transformation $\rho \mapsto U \rho U^\dagger$, measuring all qubits in the computational basis to obtain an outcome $\hat{\ket{b}}$ for $\hat{b} \in \{0,1\}^n$, and repeating this process several times.
Here, $U$ is a unitary selected at random from a fixed ensemble of unitaries $\mathcal{U}$.
Different choices of $\mathcal{U}$ lead to different guarantees, as we discuss later.
After measuring in the computational basis, we can apply the inverse of $U$ to the outcome state to obtain $U^\dagger \hat{\ket{b}}\hat{\bra{b}}U$.
In expectation (over the random choice of unitary and the measurement), this quantity contains can be viewed as a quantum channel applied to $\rho$:
\begin{equation}
    \mathcal{M}(\rho) = \mathbb{E}\left[U^\dagger \hat{\ket{b}}\hat{\bra{b}}U\right].
\end{equation}
Thus, applying the inverse of this quantum channel to our classical data $U^\dagger \hat{\ket{b}}\hat{\bra{b}}U$ gives us an approximation of $\rho$:
\begin{equation}
    \hat{\rho} = \mathcal{M}^{-1}\left(U^\dagger \hat{\ket{b}}\hat{\bra{b}}U\right).
\end{equation}
The classical shadow representation is then a collection of such $\hat{\rho}$ after repeating this measurement procedure several times:
\begin{equation}
    \mathsf{S}(\rho; T) \triangleq \{\hat{\rho}_1,\dots, \hat{\rho}_T\}.
\end{equation}
We can then predict properties $\tr(O_i \rho)$ using the classical shadow via median-of-means~\cite{nemirovskij1983problem,jerrum1986random}.
Specifically, one can split the classical shadow into $K$ equally sized parts, take the empirical mean of each of these parts to obtain $\hat{\rho}_{(k)}$, compute the property $\tr(O_i \hat{\rho}_{(k)})$ for each of these empirical means, and then take median $\mathrm{median}(\tr(O_i \hat{\rho}_{(1)}),\dots, \tr(O_i \hat{\rho}_{(K)}))$ of the results.
For this simple procedure, the classical shadow formalism achieves the following rigorous guarantee for predicting expectation values.

\begin{theorem}[Theorem 1 in~\cite{huang2020predicting}]
    \label{thm:classical-shadow}
    Let $O_1,\dots,O_M$ be Hermitian $2^n \times 2^n$ matrices, and let $\epsilon, \delta \in [0,1]$. Then,
    \begin{equation}
        N = \mathcal{O}\left(\frac{\log(M/\delta)}{\epsilon^2}\max_{1\leq i\leq M} \norm{O_i - \frac{\tr(O_i)}{2^n}\mathbb{I}}^2_{\mathrm{shadow}}\right)
    \end{equation}
    copies of an unknown quantum state $\rho$ suffice to predict $\hat{o}_i$ such that
    \begin{equation}
        |\hat{o}_i - \tr(O_i \rho)| \leq \epsilon
    \end{equation}
    for all $1 \leq i \leq M$, with probability at least $1-\delta$.
\end{theorem}
Here, $\norm{\cdot}_{\mathrm{shadow}}$ denotes the shadow norm, which depends on the ensemble of unitary transformations used to create the classical shadow. Explicitly, the shadow norm is given by
\begin{equation}
    \norm{O}_{\mathrm{shadow}} = \max_{\sigma : \text{state}}\left(\operatornamewithlimits{\mathbb{E}}_{U \sim \mathcal{U}} \sum_{b \in \{0,1\}^n} \expval{U\sigma U^\dagger}{b} \expval{U\mathcal{M}^{-1}(O)U^\dagger}{b}^2\right)^{1/2}.
\end{equation}

There are two examples of unitary ensembles emphasized in~\cite{huang2020predicting} that we also focus on here.
These are the cases of tensor products of random single-qubit Pauli gates and random $n$-qubit Clifford circuits.
Random Pauli measurements allow us to predict properties for local observables while random Clifford measurements allow us to predict properties for observables with bounded Frobenius norm $\tr(O^2)$.
In what follows, we provide sample complexity upper and lower bounds for predicting properties with respect to local observables and observables with bounded Frobenius norm.

Moreover, in the local observable case, there is a particularly simple form for the inverse quantum channel $\mathcal{M}^{-1}$, which we make use of in the later sections.
In particular, when predicting local observables, the random unitary $U = U_1 \otimes \cdots \otimes U_n$ can be written as a tensor product of random single-qubit Pauli gates.
Then, the inverse quantum channel is given by
\begin{equation}
    \label{eq:local-channel}
    \hat{\rho} = \bigotimes_{j=1}^n \left(3U_j^\dagger \hat{\ket{b_j}}\hat{\bra{b_j}}U_j-\mathbb{I}\right),
\end{equation}
where $\hat{b}_1,\dots,\hat{b}_n \in \{0,1\}$.
Note that for local observables, the classical shadow of a quantum state can be stored efficiently in classical memory because
\begin{equation}
    U_j^\dagger \hat{\ket{b_j}}\hat{\bra{b_j}}U_j \in \{\ket{0}, \ket{1},\ket{+},\ket{-},\ket{+i}, \ket{-i}\}.
\end{equation}
This is clear because each of the $U_j$ are Pauli gates.
Thus, we can represent each classical snapshot $\hat{\rho}_i$ using $\mathcal{O}(n)$ bits ($\mathcal{O}(1)$ bits to store which of the $6$ Pauli eigenstates $U_j^\dagger \hat{\ket{b_j}}\hat{\bra{b_j}}U_j$ is for each of the $n$ qubits).
The entire classical shadow $\mathsf{S}(\rho; T) = \{\hat{\rho}_1,\dots, \hat{\rho}_T\}$ can then be stored in $\mathcal{O}(nT)$ bits.

\subsubsection{Non-adaptively chosen Pauli Observables}
\label{sec:nonadapt-pauli}

In this section, we give an overview of the algorithm from~\cite{huang2021information}.
We utilize this in Appendix~\ref{app:pauli-upper} to prove the sample complexity upper bound for predicting expectation values of adaptively chosen Pauli observables.
For a more detailed presentation, we direct the reader to~\cite{huang2021information}.

Given copies of an unknown $n$-qubit quantum state, the goal of the algorithm from~\cite{huang2021information} is to predict expectation values $\tr(P_i\rho)$, of $M$ different (non-adaptively chosen) $n$-qubit Pauli observables $P_1,\dots, P_M$, i.e., $P_i \in \{I, X, Y, Z\}^{\otimes n}$.
For this task, \cite{huang2021information} achieves the following rigorous guarantee.

\begin{theorem}[Theorem 4 in~\cite{huang2021information}]
    \label{thm:nonadapt-pauli}
    Let $\epsilon, \delta > 0$.
    Let $P_1,\dots, P_M \in \{I,X,Y,Z\}^{\otimes n}$ be $M$ $n$-qubit Pauli observables.
    Then,
    \begin{equation}
        N = \mathcal{O}\left(\frac{\log(M/\delta)}{\epsilon^4}\right)
    \end{equation}
    copies of an unknown quantum state $\rho$ suffice to predict $\hat{p}_i$ such that
    \begin{equation}
        |\hat{p}_i - \tr(P_i \rho)| \leq \epsilon
    \end{equation}
    for all $1 \leq i \leq M$, with probability at least $1-\delta$.
\end{theorem}

The algorithm that achieves this sample complexity upper bound estimates $\tr(P\rho)$ in two steps.
First, the algorithm estimates the magnitude $\lvert\tr(P\rho)\rvert^2$.
Second, the algorithm determines $\mathrm{sign}(\tr(P\rho))$.

In the first step, one performs $N_1$ rounds of Bell measurements on $\rho \otimes \rho$.
In this way, one can obtain classical measurement data $S^{(t)}=\{S_k^{(t)}\}_{k\in [n]}$, where $t \in [N_1]$ corresponds to the $t$'th round of measurement and $k \in [n]$ corresponds to the $k$th qubit.
Then, the algorithm estimates the magnitude $\lvert\tr(P\rho)\rvert^2$ via the following classical processing.
Let $P = \sigma_1 \otimes \cdots \otimes \sigma_n$.
We can compute
\begin{equation}
    \label{eq:pauli-query}
    q_P(S^{(t)})=\prod_{k=1}^n \trace\left((\sigma_k\otimes\sigma_k)S^{(t)}_k\right) \in \{\pm 1\}
\end{equation}
for each $t\in [N_1]$, where we suggestively denote this quantity by $q_P$, similar to our query notation from Appendix~\ref{sec:int-mod}.
\cite{huang2021information} shows that the expectation of $q_P$ is the value we wish to estimate: 
\begin{equation}
    \mathbb{E}\left[q_P(S^{(t)})\right] = \lvert\trace(P\rho)\rvert^2.
\end{equation}
Thus, one can estimate this expectation value with the empirical mean
\begin{equation}
    q_P(\{S^{(t)}\}_{t\in [N_1]})=\frac{1}{N_1}\sum_{t=1}^{N_1} q_P(S^{(t)}).
\end{equation}
Then,~\cite{huang2021information} shows that if the Pauli observables $P_1,\dots,P_M$ are selected non-adaptively, then 
\begin{equation}
    2N_1=\mathcal{O}\left(\frac{\log(M/\delta)}{\epsilon^4}\right)
\end{equation}
copies of the unknown quantum state $\rho$ suffice to estimate $\lvert\trace(P_i\rho)\rvert^2$ with error at most $\epsilon^2$ for all $1\leq i \leq M$.

In the second step to determine $\mathrm{sign}(\tr(P\rho))$, \cite{huang2021information} performs a coherent measurement on $N_2$ copies of $\rho$.
This measurement measures the eigenvalues of the Pauli operator $P$ on each copy of $\rho$ and takes the majority vote, yielding an estimate for the sign of $\tr(P\rho)$.
For predicting $M$ Pauli observables $P_1,\dots, P_M$, this requires
\begin{equation}
    N_2 = \mathcal{O}\left(\frac{\log(M/\delta)}{\epsilon^2}\right)
\end{equation}
copies of $\rho$ in quantum memory.
Moreover, the argument in~\cite{huang2021information} allows one to reuse these $N_2$ copies of $\rho$ stored in quantum memory for each queried Pauli observable $P_i$.
This is due to the quantum union bound~\cite{aaronson2006qma} (a generalization of the gentle measurement lemma~\cite{aaronson2004limitations,wildebook}), since performing this coherent measurement does not disturb the copies of the quantum state much so that one can perform the coherent measurement repeatedly, for each observable $P_i$ queried.
Notice that this in fact makes this step of the algorithm resistant to adaptivity.

\subsubsection{Classical shadows from uniform POVM measurements}
\label{sec:uniform-povm}
In this section, we discuss a different version of the classical shadow formalism introduced in Appendix~\ref{sec:classical-shadow} in which the shadows are generated from independent measurements in the uniform POVM.
This version was first introduced in~\cite{kueng2019classical}.
These classical shadows from uniform POVM measurements exhibit stronger concentration properties than the canonical classical shadow formalism~\cite{huang2020predicting} using Clifford and Pauli measurements.
This property will be important for our proof of Theorem~\ref{thm:nts} in Appendix~\ref{sec:thresh}.

The classical shadow is obtained by performing the following measurement procedure: apply a Haar random unitary to the state $\rho \mapsto U\rho U^\dag$ and then perform a computational-basis measurement.
Upon receiving the $n$-bit measurement outcome $\hat{\ket{b}} \in \{0,1\}^n$, we can store an (inefficient) classical description of $U^\dag \hat{\ketbra{b}}U$ in classical memory.
Applying an inverse quantum channel, as in the canonical classical shadow formalism~\cite{huang2020predicting}, gives us an approximation of $\rho$
\begin{equation} \label{eq: classical shadow}
    \hat{\rho}=\mathcal{M}^{-1}\left(U^\dag \ketbra{\hat{b}}U\right),
\end{equation}
where the inverted quantum channel is $\mathcal{M}^{-1}(X)=(2^n+1)X-\mathbb{I}$.
Note that $\mathbb{E}[\hat{\rho}] = \rho$, where the expectation is over the random choice of unitary and the randomness in the measurement results.
The classical shadow representation is then a collection of such $\hat{\rho}$ after repeating this measurement procedure several times:
\begin{equation}
    \mathsf{S}(\rho; N) = \{\hat{\rho}_1,\dots, \hat{\rho}_N\}.
\end{equation}

The procedure above is equivalent to repeatedly measuring the unknown state $\rho$ with the uniform POVM $\{d \ketbra{v}\, \mathrm{d}v\}$, where $\mathrm{d}v$ is the unique unitarily invariant measure on the complex unit sphere induced by the Haar measure, and applying the same inverted quantum channel to the measurement results.
This is similar to the tomography methods in~\cite{guctua2020fast}.
This is why we refer to this version of classical shadow as using uniform POVM measurements.

We can then predict properties $\tr(O \rho)$ using the classical shadow by simply taking the empirical mean.
Specifically, we can produce an estimate $\hat{o}$ of $\tr(O\rho)$ as follows.
\begin{equation}
    \label{eq:ohat}
    \hat{o} = \frac{1}{N} \sum_{i=1}^N \tr(O \hat{\rho}_i).
\end{equation}
Because $\mathbb{E}[\hat{\rho}_i] = \rho$, this estimate clearly reproduces $\tr(O\rho)$ in expectation.
In what follows, we consider the case where $N =1$, i.e., we take $\hat{\rho}$ in Equation~\eqref{eq: classical shadow} to be our classical shadow.
The analysis easily generalizes to arbitrary $N$, but we will need the $N = 1$ case in Appendix~\ref{sec:thresh}.

We now prove the key property that this classical shadows estimate for the expectation value of an observable exhibits exponential concentration.

\begin{prop} \label{prop: concentration}
    Fix an observable $O$ and let $\hat{o} = \tr(O\hat{\rho})$, where $\hat{\rho}$ is an estimate of $\rho$ as in Equation~\eqref{eq: classical shadow}. Then for any $0 \leq \tau \leq 1$,
    \begin{equation}
        \Pr[|\hat{o}-\mathbb{E}[\hat{o}]|\geq \tau] \leq 2\exp\left(-\frac{\tau^2}{16B+4\sqrt{B}\tau}\right).
    \end{equation}
\end{prop}


In order to prove the proposition, we require the following two well-known results.
First, recall a subexponential formulation of the scalar Bernstein Inequality.
This is a well-known result in probability (see, e.g., Lemma 5.4 in~\cite{rinaldo2019advanced}).

\begin{theorem}[Scalar Bernstein Inequality]
    \label{thm: bernstein}
    Let $X_1,...,X_m$ be independent random variables with zero mean and variance $\sigma_i^2$ such that for all integers $k\geq 2$, there exists an $R > 0$ such that
    \begin{equation}
        \mathbb{E}|X_i|^k \leq k!R^{k-2}\frac{\sigma_i^2}{2} \quad \textrm{for all} \quad 1 \leq i \leq m.
    \end{equation}
    Set $\sigma^2 = \sum_{i=1}^m \sigma_i^2$. Then, for all $\tau > 0$
    \begin{equation}
        \Pr[\left|\sum_{i=1}^m X_i\right| \geq \tau] \leq 2\exp\left(-\frac{\tau^2/2}{\sigma^2+R\tau}\right).
    \end{equation}
\end{theorem}

Second, we need the following result about the moments of the uniform POVM, which follows from arguments in representation theory (see, e.g.,~\cite{scott2006tight,gross2015partial,guctua2020fast}).

\begin{theorem} [Moments of the uniform POVM] \label{thm: moment}
    Fix $k\in \mathbb{N}$ and let $P_{\mathrm{Sym}^{(k)}}$ denote the projector onto the totally symmetric subspace $\mathrm{Sym}^{(k)} \subset (\mathbb{C}^d)^{\otimes k}$. Then, 
    \begin{equation*}
        \int_{\mathbb{S}^{d-1}} (\ketbra{v})^{\otimes k} \mathrm{d}v = {d+k-1 \choose k}^{-1} P_{\mathrm{Sym}^{(k)}}.
    \end{equation*}
\end{theorem}

One can show this by using the unitary invariance of $\mathrm{d}v$ to show that the left hand side commutes with $U^{\otimes k}$ for all $d\times d$ unitaries combined with Schur's lemma.

We can use Theorem~\ref{thm: moment} to prove bounds on the moments of a random variable based on the classical shadows estimate, which will then allow us to apply Theorem~\ref{thm: bernstein} to obtain exponential concentration.

\begin{lemma} \label{lem: moment bound}
    Fix density matrix $\rho$ of dimension $d$ and observable $O$ such that $tr(O^2)\leq B$. For $\hat{\rho}$ given in \ref{eq: classical shadow}, define the random variable $\Bar{X}=\tr(O\hat{\rho})-\tr(O\rho)$. Then, for $k\geq 2$,
    \begin{equation*}
        \mathbb{E}[|\Bar{X}|^k] \leq k! (4B)^{k/2}.
    \end{equation*}
\end{lemma}
\begin{proof}
    Set
    \begin{equation}
        E = O - \frac{\tr(O)+\tr(O\rho)}{d+1}\mathbb{I}.
    \end{equation}
    Then, we can reformulate $\Bar{X}$ as
    \begin{equation}
        \Bar{X} = (d+1)\bra{v}E\ket{v} \quad \textrm{with prob} \quad d\bra{v}\rho\ket{v}.
    \end{equation}
    We first focus on bounding even moments, i.e., $k$ even. Using Theorem~\ref{thm: moment} to compute the expectation value,
    \begin{align}
        \mathbb{E}[|\Bar{X}|^k] &= \mathbb{E}[\bar{X}^k] = \int_{\mathbb{S}^{d-1}} d\bra{v}\rho\ket{v} ((d+1)\bra{v}E\ket{v})^k \diff v\\
        &= d(d+1)^k \tr\left(\left(\int_{\mathbb{S}^{d-1}} (\ketbra{v})^{\otimes k+1}\right)\rho \otimes E^{\otimes k}\right)\\
        &= d(d+1)^k {d+k \choose k+1}^{-1} \tr(P_{\mathrm{Sym}^{(k+1)}}\rho \otimes E^{\otimes k}).
    \end{align}
    We bound the prefactor and the trace term separately. The prefactor can be bounded as
    \begin{equation*}
        d(d+1)^k {d+k \choose k+1}^{-1} = \frac{(k+1)!d(d+1)^k}{(d+k)...(d+1)d} \leq (k+1)!
    \end{equation*}
    The trace evaluates all possible contractions of the tensor $\rho \otimes E \otimes ... \otimes E$, yielding the bound
    \begin{equation}
    \label{eq:sym}
    \tr(P_{\mathrm{Sym}^{(k+1)}}\rho \otimes E^{\otimes k}) \leq \max\{\tr(E),\|E\|_F\}^k.
    \end{equation}
    This inequality holds by using the definition of the projector onto the symmetric subspace, submultiplicativity of the trace, and $\rho$ having unit trace.
    We compute this explicitly for a few examples to convince the reader.
    Recall that, following the presentation in~\cite{mele2023introduction}, the projection onto the symmetric subspace can be written as
    \begin{equation}
        P_{\mathrm{Sym}^{(k)}} = \frac{1}{k!} \sum_{\pi \in S_k} V_d(\pi),
    \end{equation}
    where
    \begin{equation}
        V_d(\pi) = \sum_{i_1,\dots, i_k \in [d]^k} \ket{i_{\pi^{-1}(1)},\dots, i_{\pi^{-1}(k}}\bra{i_1,\dots, i_k}
    \end{equation}
    for $\pi \in S_k$ with $S_k$ the symmetric group over $k$ elements.
    Then, the inequality in Equation~\eqref{eq:sym} is clear for $k = 1$ since $S_2 = \{(1)(2), (1\;2)\}$:
    \begin{equation}
        \tr(P_{\mathrm{Sym}^{(2)}} \rho \otimes E) = \frac{1}{2}(\tr(\rho) \tr(E) + \tr(E\rho)) \leq \tr(E).
    \end{equation}
    Similarly, we can also show this for $k = 2$, where $S_3 = \{(1)(2)(3),(1)(2\;3),(3)(1\;2), (1\;2\;3),(1\;3\;2), (2)(1\;3)\}$:
    \begin{align}
        &\tr(P_{\mathrm{Sym}^{(3)}}\rho \otimes E \otimes E) \\
        &= \frac{1}{6}(\tr(\rho)\tr(E)\tr(E) + \tr(\rho) \tr(E^2) + \tr(E)\tr(E\rho) + \tr(E\rho E) + \tr(E^2 \rho) + \tr(E)\tr(E\rho))\\
        &= \frac{1}{6}(\tr(E)^2 + \tr(E^2) + 2\tr(E)\tr(E\rho) + 2\tr(E^2\rho)\\
        &\leq \frac{1}{6}(\tr(E)^2 + \tr(E^2) + 2\tr(E)^2 + 2\tr(E^2).
    \end{align}
    Here, if $\tr(E)^2 \leq \tr(E^2) = \norm{E}_2^2$, then the above inequality is bounded by $\tr(E^2) = \norm{E}_2^2$.
    If $\tr(E^2) \leq \tr(E)^2$, then the above inequality is bounded by $\tr(E)^2$.
    Thus, we have that
    \begin{equation}
        \tr(P_{\mathrm{Sym}^{(3)}} \rho \otimes E \otimes E) \leq \max(\tr(E), \norm{E}_2^2)^2.
    \end{equation}
    Similar calculations hold for general $k$, so it is clear that Equation~\eqref{eq:sym} holds.

    We can further bound this quantity:
    \begin{equation}
        \tr(P_{\mathrm{Sym}^{(k+1)}}\rho \otimes E^{\otimes k}) \leq \max\{\tr(E),\|E\|_F\}^k \leq B^{k/2}.
    \end{equation}
    This follows from the following computations:
    \begin{equation}
        \tr(E) = \tr(O) - \frac{d}{d+1}(\tr(O)+\tr(O\rho)) \leq \frac{\tr(O)}{d+1} -\frac{d}{d+1} \tr(O\rho) \leq \frac{d}{d+1}+1 \leq 2.
    \end{equation}
    In the first inequality, we use the assumption that $O\leq \mathbb{I}$ in PSD ordering.
    In the second inequality, we use that the spectral norm of $O$ is bounded by one.
    \begin{align}
        \norm{E}_F^2 &= \tr(E^2)\\
        &= \tr\left(\left(O - \frac{\tr(O)+\tr(O\rho)}{d+1}\mathbb{I}\right)^2\right)\\
        &= \tr(O^2) - 2 \frac{\tr(O) + \tr(O\rho)}{d + 1}\tr(O) + \left(\frac{\tr(O) + \tr(O\rho)}{d + 1}\right)^2 \tr(\mathbb{I}^2)\\
        &= \tr(O^2) - \frac{d+2}{(d+1)^2} \tr(O)^2 - \frac{2}{(d+1)^2}\tr(O)\tr(O\rho) + \frac{d}{(d+1)^2} \tr(O\rho)^2\\
        &\leq \tr(O^2) - \frac{d+2}{(d+1)^2} \tr(O)^2 - \frac{2}{(d+1)^2}\tr(O)\tr(O\rho) + \frac{d}{(d+1)^2} \tr(O)^2\\
        &= \tr(O^2) - \frac{2}{(d+1)^2} (\tr(O)^2 + \tr(O)\tr(O\rho))\\
        &\leq \tr(O^2) - \frac{2}{(d+1)^2}\tr(O)\tr(O\rho)\\
        &\leq \tr(O^2) - \frac{2}{(d+1)^2}\tr(O\rho)^2\\
        &\leq B,
    \end{align}
    where we use the submultiplicative property of the trace several times.
    In summary,
    \begin{equation} \label{eq: even moments}
        \mathbb{E}[|\Bar{X}|^k] \leq (k+1)! B^{k/2} \quad \textrm{for} \quad k=2,4,6,...
    \end{equation}
    In order to bound the odd moments, we use the following trick, which allows us to convert odd moments into even moments. The function $x \mapsto x^{\frac{k}{k+1}}$ is concave on the positive reals. Jensen's inequality then implies
    \begin{equation}
        \mathbb{E}[|\bar{X}|^k] \leq (\mathbb{E}[|\bar{X}|^{k+1}])^{\frac{k}{k+1}}.
    \end{equation}
    Using Equation~\eqref{eq: even moments},
    \begin{equation}
        \mathbb{E}[|\bar{X}|^k] \leq \left((k+2)! B^{(k+1)/2}\right)^{\frac{k}{k+1}} \leq (k+2)!B^{k/2}
    \end{equation}
    for $k\geq 3$ odd. The final claim follows from the observation:
    \begin{equation}
        (k+1)! \leq (k+2)! \leq 2^k k!
    \end{equation}
\end{proof}

With this result, we are now ready to prove Proposition~\ref{prop: concentration}, which follows by an application of Theorem~\ref{thm: bernstein}.

\begin{proof}[Proof of Proposition~\ref{prop: concentration}]
    Let $\bar{X} = \tr(O\hat{\rho}) - \tr(O\rho)$, where $\hat{\rho}$ is a classical shadow as in Equation~\eqref{eq: classical shadow}.
    Then, since $\mathbb{E}[\hat{\rho}] = \rho$, then we clearly have that $\bar{X}$ has zero mean and $\bar{X} = \hat{o} - \mathbb{E}[\hat{o}]$ when $\hat{o} = \tr(O\hat{\rho})$.
    We want to bound $\Pr[|\bar{X}| \geq \tau]$.
    By Lemma~\ref{lem: moment bound}, we found that
    \begin{equation}
        \mathbb{E}[|\bar{X}|^k] \leq k! (4B)^{k/2}.
    \end{equation}
    Note that the variance of $\bar{X}$ is $8B$ by taking the second moment bound in the above inequality so that choosing $R = 2\sqrt{B}$, $\bar{X}$ satisfies the conditions for Theorem~\ref{thm: bernstein}.
    Applying this, we have
    \begin{equation}
        \Pr[|\hat{o} - \mathbb{E}[\hat{o}]| \geq \tau] = \Pr[|\bar{X}| \geq \tau] \leq 2\exp\left(-\frac{\tau^2/2}{8B + 2\sqrt{B} \tau}\right) = 2\exp\left(-\frac{\tau^2}{16B + 4\sqrt{B} \tau}\right),
    \end{equation}
    as claimed.
\end{proof}

\subsection{Fingerprinting Codes}
\label{sec:fingerprint}

Our sample complexity lower bounds for the local and Pauli observable cases rely heavily on fingerprinting codes and their guarantees.

Fingerprinting codes were originally introduced by \cite{705568} within the context of watermarking digital content to prevent piracy.
We first describe the intuition behind these constructions and then give a formal definition.
Suppose we are distributing content to a group of $d$ users, but there is a size $N$ colluding group $\mathcal{S}$ outputting illegal copies. 
Note here that we are suggestively using the same notation as previously introduced for the adaptive data analysis setting (i.e., where $N$ was the number of points in the mechanism's dataset $\mathcal{S}$).
\cite{705568} shows that by assigning unique ``watermarks'' to the content, one can trace illegal copies back to the user in $\mathcal{S}$ that produced it.
Importantly, the illegal copy must be consistent with one of the watermarks given to members in $\mathcal{S}$.
Fingerprinting codes are a generalization of these watermarks that can identify some user in $S$ responsible even when the users in $\mathcal{S}$ collude to construct an illegal copy.
\cite{tardos2008optimal} gave optimal constructions for fingerprinting codes.
Moreover, fingerprinting codes are central to sample complexity lower bounds in the differential privacy literature~\cite{bun2014fingerprinting,ullman2013answering}.
Interactive fingerprinting codes~\cite{fiat2001dynamic} extend this idea even further, being able to adaptively identify each of the users in $\mathcal{S}$ one by one, where constructions are given in~\cite{tassa2005low,laarhoven2013dynamic}.
Interactive fingerprinting codes were used in \cite{steinke2015interactive} to prove lower bounds for general adaptive data analysis problems, where the idea is to adaptively choose queries (watermarks) such that the analyst can identify every point in the dataset $\mathcal{S}$.
Note that if the analyst knows $\mathcal{S}$, he then can easily overfit, e.g., by choosing a query $q$ such that $q(x)=1$ if $x\in \mathcal{S}$ and $0$ otherwise.\\

With this intuition, we may now describe interactive fingerprinting codes more formally.
Interactive fingerprinting codes operate according to a game between an adversary $\mathcal{P}$ and fingerprinting code $\mathcal{F}$. This game has $M$ rounds and is presented in Algorithm~\ref{alg:adaptive}.
Again, here we are suggestively using the same notation as before (i.e., where $M$ was the number of queries made by the analyst).
During round $j$, the fingerprinting code $\mathcal{F}$ broadcasts to every user $i\in[d]$ a code $c^j_i\in\{0,1\}$. The colluding group $\mathcal{S}\subset [d]$ of size $N$ produces a code $a^j$ (for their illegal content), with the requirement that it must satisfy \textit{consistency}: there exists a user $i$ such that $a^j=c^j_i$.
At the end of each round, the fingerprinting code $\mathcal{F}$ accuses a set of users and prevents them from seeing any future broadcasts from $\mathcal{F}$.
Thus, in each round, the fingerprinting code $\mathcal{F}$ iteratively identifies users in the colluding group $\mathcal{S}$.
\begin{algorithm}
\caption{$\mathsf{IFPC}_{N,d,M}(\mathcal{P},\mathcal{F})$ \cite{steinke2015interactive}}\label{alg:adaptive}
$\mathcal{P}$ selects colluding users $\mathcal{S}^1\subseteq [d]$ (unknown to $\mathcal{F}$)\\
\For{$j=1$ to $M(N)$}{
$\mathcal{F}$ outputs column vector $c^j\in\{0,1\}^d$.\\
Let $c^j_{S^j}\in\{0,1\}^{|S^j|}$ be the restriction of $c^j$ to the coordinates of $S^j$. This is given to $\mathcal{P}$.\\
$\mathcal{P}$ generates $a^j\in\{0,1\}$. This is given to $\mathcal{F}$.\\
$\mathcal{F}$ accuses a set of users $I^j\subseteq [d]$. Let $S^{j+1}=S^j \setminus I^j$.
}
\end{algorithm}

We give a security definition for an interactive fingerprinting code based on this game $\mathsf{IFPC}_{N,d,M}(\mathcal{P}, \mathcal{F})$.
Informally, we want both completeness and soundness properties: the fingerprinting code should both be able to identify the colluding users $\mathcal{S}$ and not make many false accusations.
To formalize this, we define the following quantities.
Let
\begin{equation}
    \theta^j\triangleq|\{1\leq k \leq j|\not\exists i \in [d], a^k=c_i^k\}|
\end{equation}
be the number of rounds $1,\dots,j$ such that the output of the adversary $\mathcal{P}$ is not consistent.
Also, let
\begin{equation}
     \psi^j\triangleq\left|\left(\bigcup_{1\leq k\leq j}I^k\right)\setminus S\right|
\end{equation} 
be the number of users in $I^1,...,I^j$ that are falsely accused by $\mathcal{F}$.
We use these quantities to define a collusion resilient interactive fingerprinting code.
\begin{definition}[Collusion resilient nteractive fingerprinting code]
An algorithm $\mathcal{F}$ is said to be an \emph{$N$-collusion resilient interactive fingerprinting code} of length $M$ for $d$ users with failure probability $\epsilon$ if for every adversary $\mathcal{P}$,\begin{equation}
    \Pr_{\mathsf{IFPC}[\mathcal{P},\mathcal{F}]}[\theta^M=0 \lor \psi^M > d/2000] \leq \epsilon
\end{equation}
\end{definition}

This says that a good (collusion resilient) fingerprinting code should be able to force the adversary $\mathcal{P}$ to be inconsistent (and hence identify the colluding users $\mathcal{S}$) while not making many false accusations.
Here, we allow the fingerprinting code to falsely accuse a small constant fraction $1/2000$ of the total number of users $d$.

Our lower bounds (and the classical adaptive data analysis lower bounds) require the existence of interactive fingerprinting codes.
We recall a construction of interactive fingerprinting codes from~\cite{steinke2015interactive}.

\begin{theorem} [Theorem 2.2. in~\cite{steinke2015interactive}]
    For every $1\leq N \leq d$, there is an $N$-collusion resilient interactive fingerprinting code of length $\mathcal{O}(N^2)$ for $d$ users with failure probability $\epsilon\leq \mathcal{O}(1/d)$.
\end{theorem}

Note that in~\cite{steinke2015interactive}, they state and prove a more general theorem, but this is all that we require.
We consider the special case of their theorem in which we tolerate falsely accusing a small constant fraction of users.

\subsection{Cryptography}
\label{app:crypto}

In this section, we review some basic cryptography, which we use in the proof of Proposition~\ref{prop:pauli-lower}.
We first define a private-key encryption scheme.

\begin{definition}[Private-key encryption scheme~\cite{katz2007introduction}]
    Let $\lambda$ be the security parameter.
    Let $\mathcal{K}$ be the key space, $\mathcal{M}$ be the message space, and $\mathcal{C}$ be ciphertext space.
    A private-key encryption scheme is composed of three algorithms $(\Gen, \Enc, \Dec)$, defined as follows.
    \begin{itemize}
        \item $\Gen(1^\lambda)$: On input the security parameter $1^\lambda$, the key-generation algorithm outputs a secret key $\sk \in \mathcal{K}$ at random.
        \item $\Enc_\sk(m)$: On input the secret key $\sk \in \mathcal{K}$ and the plaintext $m \in \mathcal{M}$, the encryption algorithm outputs a ciphertext $c \in \mathcal{C}$.
        \item $\Dec_\sk(c)$: On input the secret key $\sk \in \mathcal{K}$ and the ciphertext $c \in \mathcal{C}$, the decryption algorithm outputs a plaintext $m \in \mathcal{M}$.
    \end{itemize}
    A private-key encryption scheme is correct if 
    \begin{equation}
        \Dec_\sk(\Enc_\sk(m)) = m.
    \end{equation}
\end{definition}

In addition to correctness, we want our encryption schemes to be secure.
In other words, from the encrypted ciphertext, it should be difficult to recover the corresponding plaintext.
The definition of security we need in this setting is perfect secrecy.

\begin{definition}[Perfect secrecy~\cite{katz2007introduction}]
    An encryption scheme $(\Gen, \Enc, \Dec)$ is \emph{perfectly secret} if for every pair of messages $m_0, m_1 \in \mathcal{M}$ and every ciphertext $c \in \mathcal{C}$,
    \begin{equation}
        \Pr_{\sk \sim \Gen(1^\lambda)} [\Enc_\sk(m_0) = c] = \Pr_{\sk \sim \Gen(1^\lambda)}[\Enc_\sk(m_1) = c].
    \end{equation}
\end{definition}

In words, this means that for any pair of messages $m_0, m_1$, it is impossible to distinguish the ciphertext of $m_0$ from the ciphertext of $m_1$.
In fact, one can also devise a simple perfectly secret private-key encryption scheme called the one-time pad~\cite{miller1882telegraphic}.

\begin{definition}[One-time pad]
    \label{def:otp}
    Let $\lambda$ be the security parameter.
    Let $\mathcal{K} = \mathcal{M} = \mathcal{C} = \{0,1\}^\lambda$ be the key, message, and ciphertext spaces.
    The one-time pad is defined by $(\Gen, \Enc, \Dec)$ as follows.
    \begin{itemize}
        \item $\Gen(1^\lambda)$: Output a uniformly random bitstring $\sk \in \{0,1\}^\lambda$.
        \item $\Enc_\sk(m)$: Encrypt $m \in \{0,1\}^\lambda$ as $c = m \oplus \sk$.
        \item $\Dec_\sk(c)$: Decrypt $c \in \{0,1\}^\lambda$ as $m = c \oplus \sk$.
    \end{itemize}
\end{definition}

The perfect secrecy of the one-time pad was first proven in~\cite{shannon1949communication}.
We will not repeat the proof here.

\section{Local Observables} \label{app: local}
We consider the task of predicting properties $\tr(O\rho)$, where $\rho$ is the unknown quantum state that we are given copies of and $O$ is a local observable with spectral norm bounded by $1$.
Here, we consider a local observable to be an observable acting on a constant number of qubits.
In this appendix, we provide a proof of Theorem~\ref{thm:local}, which characterizes the sample complexity for this problem.
We restate the theorem below.
\begin{theorem}[Local observables, detailed restatement of Theorem~\ref{thm:local}]
    There exists an analyst $\mathcal{A}$ and a density matrix $\rho$ on $n = \Omega(M2^{2000\sqrt{M}})$ qubits such that any $(0.99, 1-\frac{1}{2000\sqrt{M}})$-accurate quantum mechanism $\mathcal{M}$ estimating expectation values of $M$ adaptively chosen single-qubit observables requires at least 
    \begin{equation}
        N = \Omega(\sqrt{M})
    \end{equation}
    samples of the quantum state $\rho$.
    Meanwhile, for any $\epsilon,\delta \in (0,1/2)$ such that $\epsilon\delta \geq 4e^{-M/\log^2M \log\log^4 M}$, there exists an $(\epsilon,\delta)$-accurate quantum mechanism $\mathcal{M}$ that can estimate expectation values of $M$ adaptively chosen $k$-local observables using
    \begin{equation}
        N=\mathcal{O}\left(\frac{\min\{3^k\sqrt{M\log(1/\epsilon\delta)},k\log (n/\delta)\}}{\epsilon^2}\right)
    \end{equation}
    samples of the unknown $n$-qubit quantum state $\rho$.
    Moreover, the mechanism is computationally efficient and runs in time $\mathrm{poly}(N, n, \log(1/\delta))$ per observable.
\end{theorem}

In Appendix~\ref{sec:attack}, we begin with a warmup showing how adaptivity can cause overfitting for the classical shadows protocol~\cite{huang2020predicting}.
Then, in Appendix~\ref{sec:lb}, we prove the sample complexity lower bound, which states that there does not exist a quantum mechanism that can accurately estimate $M$ adaptively chosen local observables given $\mathrm{poly\,log}(M)$ copies of $\rho$.
Finally, in Appendix~\ref{sec:classical shadow ub}, we give a quantum algorithm achieving a matching sample complexity upper bound for answering adaptively chosen local observables.

In both Appendices~\ref{sec:attack} and \ref{sec:lb}, we focus on the setting where we are given $N$ copies of a diagonal $n$-qubit density matrix $\rho$.
This diagonal density matrix can be viewed as the classical probability distribution $\mathcal{D}$ in the statistical query model. In addition, the analyst is constrained to only querying local observables $\{Z_i : i \in [n]\}$, where $Z_i$ is the Pauli $Z$ observable that acts only on the $i$th qubit. In this notation, we implicitly have the identity matrix on the remaining qubits. For each query $Z_i$, the mechanism aims to estimate $\tr(Z_i \rho)$.
We show that for any mechanism, there always exists a quantum state $\rho$ (of this diagonal form) and an analyst (even with this restricted power of only querying $Z_i$ observables) such that the mechanism fails to estimate the expectation value of the analyst's adversarially chosen observables.

In this setting, one way of viewing the problem is as a classical problem involving $n$ coins, whose joint distribution is dictated by the diagonal elements of $\rho$. The query $q(x)$ is the outcome of flipping one of the coins in sample $x$, and the goal would be to predict the true biases $q(\mathcal{D})$ of the queried coins.

\subsection{Adaptive Attack} \label{sec:attack}
First, we demonstrate the dangers of adaptivity by constructing an attack querying only single-qubit observables that is sufficient to cause the classical shadows protocol to fail after $\mathcal{O}(N)$ queries, where $N$ is the number of samples of the unknown quantum state $\rho$ the mechanism is given. 
We devise the analyst's adversarial strategy and a density matrix $\rho$ such that classical shadows fails.

Consider a density matrix $\rho$ on $n=M+2^M$ qubits, where recall $M$ is the number of queries the analyst makes. Here, the qubits are correlated such that the outcomes of the last $2^M$ qubits are completely determined by the first $M$. Denote the power set as $\mathcal{P}([M])\triangleq\{I_1,...,I_{2^M}\}$, i.e., $I_1 = \emptyset, I_2 = \{1\}, \dots, I_{2^M}=[M]$. We will simulate a classical probability distribution by defining $\rho$ to be the diagonal matrix
\begin{equation}
    \label{eq:rho-q}
    \rho \triangleq \sum_{Q\in \{0,1\}^n} p(Q) \ketbra{Q}{Q}
\end{equation}
where $Q=(Q_1,...,Q_n) \in \{0,1\}^n$ so that $\ket{Q}$ is a computational basis states. Here, we use subscripts to denote labels for the different qubits.
To construct an adversarial density matrix, we effectively encode the classical distribution from the feature selection attack of~\cite{dwork2015generalization,roth2017adaptive} into the density matrix.
Thus, we define the probability distribution $p(Q)$ as
\begin{equation}
    \label{eq:rho-dist}
    p(Q) = p(Q_1,\dots, Q_n) \triangleq
    \begin{cases}
        \frac{1}{2^M} & \text{if } \mathrm{MAJ}(\{Q_i : i\in I_j\})=Q_{M+j}, \forall j \in [2^M]\\
        0 & \text{otherwise}
    \end{cases},
\end{equation}
where $\mathrm{MAJ}(\{Q_i\mid i\in I_j\})$ denotes the majority element of the set $\{Q_i : i \in I_j\}$.
In other words,
\begin{equation} \label{eq: majority-rule}
    \mathrm{MAJ}(\{Q_i : i \in I_j\}) = \begin{cases}
        1 & \text{if } \sum_{i\in I_j}Q_i>0\\
        0 & \text{otherwise}
    \end{cases}.
\end{equation}
As a technical detail, the rule for $Q_{M+j}$ corresponding to $I=\emptyset$ can simply be set ahead arbitrarily, e.g., set to always be 0.
We can view the distribution $p(Q)$ of the qubits labeled by $Q = (Q_1,\dots, Q_n)$ as follows.
The first $M$ qubits labeled by $Q_1,\dots,Q_M$ are uniformly distributed over $\{0,1\}^M$ while the last $2^M$ qubits $Q_{M+1},\dots,Q_{M+2^M}$ are uniquely determined by the first $M$ qubits via a majority voting rule.

For this density matrix, we want to show that there exists an adversarial analyst that can cause the classical shadows protocol with random Pauli measurements to fail after $\mathcal{O}(N)$ queries.
This means that the classical shadows formalism for predicting local observables fails to protect against adaptive queries.
Recall that classical shadows for predicting local observables are produced by (i) repeatedly measuring each qubit of $\rho$ independently in a random Pauli basis and receiving output bitstring $\hat{\ket{b}}\in \{0,1\}^n$ and then (ii) computing $N$ classical snapshots of $\rho$ of the form $\hat{\rho}_k=\bigotimes_{j=1}^n (3U_j^{(k)\dag}\ketbra{\hat{b}_j^{(k)}}U_j^{(k)}-\mathbb{I})$, for $k=1,\dots, N$, where $U_j^{(k)}$ is the Pauli matrix with respect to which the $j$th qubit in the $k$th copy of $\rho$ was measured.
Then, given a queried observable $O$ from the analyst, the classical shadow mechanism estimates $\trace(O\rho)$ by returning $a(O)\triangleq\frac{1}{N}\sum_{k=1}^N \tr(O\hat{\rho_k})$, i.e., the empirical mean of the expectation values on each of the classical snapshots. \footnote{We remark that for local observables, the empirical mean can be used because the classical shadows output has a range bound in this case. However, the attack still works even if median of means is used.}
We refer to Appendix~\ref{sec:classical-shadow} and~\cite{huang2020predicting} for a more in-depth presentation of the classical shadow formalism.
In our restricted setting, the analyst only queries from the set of single-qubit observables $\{Z_i : i\in [n]\}$.

Finally, we describe the analyst's behavior, which causes classical shadows to fail.
The analyst performs the following procedure.
\begin{enumerate}
    \item Query $Z_1,...,Z_M$ and receive $a(Z_i)$ for $i\in [M]$.
    \item Let $I\triangleq\left\{i \in [M]:a(Z_i)\geq \frac{9}{\sqrt{N}}\right\}$. Then query $Z_{M+j^*}$ where $I_{j^*}=I$.
\end{enumerate}
Here, because we indexed the power set $\mathcal{P}([M])$ as $\mathcal{P}([M]) = \{I_1,\dots, I_{2^M}\}$ and $I \subseteq [M]$, there clearly must exist some $j^* \in [2^M]$ such that $I_{j^*} = I$.
We claim that the mechanism's answer $a(Z_{M+j^*})$ to the query $Z_{M+j^*}$ in Step 2 of the above procedure does not estimate the true expectation value $\tr(Z_{M+j^*}\rho)$ well for $\rho$ defined in Eq.~\eqref{eq:rho-dist}.

The intuition for the attack is that the analyst looks for $Z_i$ whose sample outcomes $a(Z_i)$ are especially biased towards the positive direction. We could also have chosen to have the analyst search for bias in the negative direction; as long as the sample outcomes are all biased in some direction, it does not matter which one.
Each such $Z_i$ will skew the empirical mean of every $Z_{M+j}$ where $i \in I_j$ in the positive direction (since each such $Z_{M+j}$ is correlated with $Z_i$ in our construction). Thus, $Z_{M+j}$ with $I_j=I$ will contain contributions from every such positively-biased $Z_i$, which will amplify its empirical bias towards the positive direction.

Formally, we have the following guarantee, which states that the analyst overfits after linearly many queries.
\begin{theorem}[Adaptive attack]
    \label{thm:adapt-attack}
    There is a constant $c$ such that if $M \geq c\max(N,\log(1/\delta))$ and $N \geq c\log(1/\delta)$, with probability $1-\delta$:\begin{equation}
        |a(Z_{M+j})-\trace(Z_{M+j}\rho)| \geq 0.99
    \end{equation}
\end{theorem}
\begin{proof}
    This can essentially be reduced to the feature selection attack of  \cite{dwork2015generalization} and \cite{roth2017adaptive} (Lecture 3).

    We first describe the random variable $\hat{z}_k \triangleq \tr(Z_i \hat{\rho}_k)$, which is the expectation value of $Z_i$ with respect to the $k$th classical snapshot $\hat{\rho}_k$.
    First, let's consider the case when the random Pauli measurement on the $i$th qubit is chosen to be $Z$.
    This occurs with probability $1/3$.
    Measuring the $i$th qubit of the $k$th copy of $\rho$ in the $Z$ basis yields the outcome $\ket{\hat{b}_i^{(k)}}$, which we denote as $\ket{\hat{b}_i^{(k)}} = \ket{Q^k_i}$ to distinguish it from other measurement results, in line with the notation in Equation~\eqref{eq:rho-q}.
    Measuring the other qubits $j \neq i$ of the $k$th copy of $\rho$ in an independent random Pauli bases yield outcomes $\ket{\hat{b}_j^{(k)}}$.
    Computing $\tr(Z_i \hat{\rho}_k)$ explicitly in this case, we have
    \begin{align}
        \hat{z}_k &= \tr(Z_i \hat{\rho}_k)\\
        &= \tr\left(\left(I \otimes \cdots \otimes I \otimes Z \otimes I \otimes \cdots \otimes I\right)\left(\bigotimes_{j=1}^n 3U_j^{(k)\dagger} \ketbra{\hat{b}_j^{(k)}}U_j^{(k)} - \mathbb{I}\right)\right)\\
        &= \tr(3Z (U_i^{(k)\dagger} \ketbra{Q_i^k}U_i^{(k)} - \mathbb{I})) \cdot 
        \prod_{j\neq i}\tr(3U_j^{(k)\dagger} \ketbra{\hat{b}_j^{(k)}}U_j^{(k)} - \mathbb{I}) \\
        &= \tr(3Z Z \ketbra{Q_i^k}Z - \mathbb{I})\\
        &= \begin{cases}
            3 & Q_i^k = 0\\
            -3 & Q_i^k = 1,
        \end{cases}
    \end{align}
    where in the fourth line, we used our assumption that in this case that the random Pauli measurement on the $i$th qubit is chosen to be $Z$.
    Meanwhile, a similar computation shows that if the random Pauli measurement on the $i$th qubit is $X$ or $Y$, the expectation value $\hat{z}_k$ evaluates to $3$ or $-3$ with equal probability.
    
    In summary, with probability $1/3$, the $i$th qubit is measured in the computational basis to be $Q_i^k$ so that $\hat{z}_k$ is $-3$ or $3$ if $Q_i^k$ is $1$ or $0$, respectively.
    With probability $2/3$, $\hat{z}_k$ is $3$ or $-3$ with equal probability.
    
    Thus, notice that for all $i$, $a(Z_i)=\frac{1}{N}\sum_{k=1}^N \tr(Z_i\hat{\rho}_k)$ is a (rescaled) binomial distribution taking values between $[-3,3]$, where the expected value is $0$ and the standard deviation is $3/\sqrt{N}$. This is because $Q^k_i$ takes values $1$ or $0$ with equal probability, so $\tr(Z_i\hat{\rho}_k)$ is $3$ or $-3$ with equal probability by the above analysis.

    We can use this to show that given $M$ large enough, then $|I|=\Omega(M)$. 
    Recall that a binomial random variable deviates from its mean by a constant factor of its standard deviation with constant probability, so
    \begin{equation}
        \Pr[i \in I]=\Pr[a(Z_i)\geq 9/\sqrt{N}]=\Omega(1).
    \end{equation}
    Denote this probability as $p$. This implies that $\mathbb{E}[|I|]=\Omega(M)$. Let $X_i$ be the random variable that is $1$ when $i \in I$ and $0$ otherwise, and let $X=\sum_{i=1}^M X_i=|I|$. By the Chernoff bound,
    \begin{equation}
        \Pr[X \leq pM - t] \leq \exp(-2t^2/M).
    \end{equation}
    By setting $t=pM/2$ and rearranging, we get that given $M \geq c\log(1/\delta)$ for some constant $c$, then $|I|=\Omega(M)$ with probability at least $1-\frac{\delta}{2}$.

    Assuming $|I|=\Omega(M)$, we will next show that with $M$ large enough and for a state $Q$ chosen uniformly at random from the sample, the odds of the adaptively chosen qubit $Q_{M+j}$ being $1$ is large. Denote $Q_{M+j}$ to be the qubit corresponding to $I$, where by our majority voting rule in Equation~\eqref{eq:rho-dist}, $Q_{M+j}=1$ if and only if $\sum_{i \in I} Q_i > 0$. Moreover, recall $i\in I$ if and only if $a(Z_i)=\frac{1}{N}\sum_{k=1}^N \tr(Z_i\hat{\rho_k})\geq \frac{9}{\sqrt{N}}$, where above we analyzed the random variable $\hat{z}_k$ in two cases: (i) with probability $1/3$ the $i$th qubit was measured in the $Z$ basis, in which case $\hat{z}_k=\tr(Z_i\hat{\rho_k})$ is $3$ if $Q^k_i=1$ and else $-3$, and (ii) with probability $2/3$ the $i$th qubit was measured in the $X$ or $Y$ basis, in which case $\hat{z}_k$ is $3$ or $-3$ with equal probability. 
    Let $T=\{\ket{Q^1},\dots, \ket{Q^N}\}$ denote the computational basis state realizations of the $N$ given copies of the quantum state $\rho$.
    Thus, for every $i \in I$ and set $T$
    \begin{align}
        \E_{Q \sim T}[\mathbbm{1}(Q_i=1)] &\geq \frac{1}{3}\E_{Q\sim T}[\mathbbm{1}(\tr(Z_i\hat{\rho})=3)]\\
        &= \frac{1}{3N}\sum_{k=1}^N \mathbbm{1}(\tr(Z_i\hat{\rho}_k)=3)\\
        &\geq \frac{1}{9}a(Z_i)\\
        &\geq \frac{1}{\sqrt{N}},
    \end{align}
    where the first inequality follows by only considering case (i) for the random variable $\hat{z}_k$ above and discarding case (ii). Thus,
    \begin{equation}
        \E_{Q \sim T}\left[\sum_{i\in |I|}\mathbbm{1}(Q_i=1)\right] \geq \frac{|I|}{\sqrt{N}}.
    \end{equation}
    This implies that $Q_{M+j}=1$ unless $\sum_{i\in |I|}\mathbbm{1}(Q_i=1)$ differs from its expectation by more than $\frac{|I|}{\sqrt{N}}$. Using another Chernoff bound,
    \begin{align}
        \Pr_{Q\sim T}[Q_{M+j} \neq 1] &= \Pr[\sum_{i\in I}\mathbbm{1}(Q_i=1)\leq \mathbb{E}\left[\sum_{i\in I}\mathbbm{1}(Q_i=1)\right] - \frac{|I|}{\sqrt{N}}]\\
        &\leq \exp\left(-\frac{2|I|^2}{N\cdot |I|}\right)\\
        &= \exp\left(\frac{-2|I|}{N}\right)
    \end{align}
    If $|I|\geq N\ln(400)/2$, then $\Pr_{Q\sim S}[Q_{M+j} \neq 1] \leq 1/400$. This follows if $M \geq c\cdot N$, for $c$ large enough.
    We will now bound the difference the classical shadows output $a(Q_{M+j})$ and the empirical mean of measuring the $M+j$th qubit of the samples in the $Z$ basis $q(Z_{M+j})=\frac{1}{N}\sum_{Q\in S} [\mathbbm{1}(Q_{M+j}=1)] - \frac{1}{N}\sum_{Q\in S} [\mathbbm{1}(Q_{M+j}\neq 1)]$. Recall that with probability $1/3$ the expectation value $\tr(Z_{M+j} \hat{\rho}_k)$ falls under case (i), where it matches the output of measuring the $(M+j)$th qubit in the $Z$ basis, and with probability $2/3$ it falls under case (ii), where it randomly outputs $3$ or $-3$ with equal probability.
    This is captured in the random variable $X_k = \tr(Z_{M+j}\hat{\rho}_k)-\bra{Q^k_{M+j}}Z_{M+j}\ket{Q^k_{M+j}}$, which is defined as follows.
    If $\bra{Q^k_{M+j}}Z_{M+j}\ket{Q^k_{M+j}}=1$, $X_k$ is $2$ with probability $2/3$ and $-4$ with probability $1/3$. 
    This is because if $\bra{Q^k_{M+j}}Z_{M+j}\ket{Q^k_{M+j}}=1$, then $Q^k_{M+j} = 0$ so that by case (i) with probability $1/3$, $\tr(Z_{M+j} \hat{\rho}_k) = 3$. Hence, $X_k = 2$ with probability $1/3$. By case (ii), with probability $2/3$, $\tr(Z_{M+j} \hat{\rho}_k)$ is $3$ or $-3$ with equal probability so that $X_k = 2$ with probability $1/3$ and $X_k = -4$ with probability $1/3$.
    In total, we have that $X_k = 2$ with probability $2/3$ and $X_k = -4$ with probability $1/3$.
    If $\bra{Q^k_{M+j}}Z_{M+j}\ket{Q^k_{M+j}}=-1$, then by a similar argument, $X_k$ is $-2$ with probability $2/3$ and $4$ with probability $1/3$.
    By applying Hoeffding's inequality,
    \begin{align*}
        \Pr[|a(Q_{M+j})-q(Z_{M+j})| \geq \frac{1}{200}] &= \Pr[|\sum_{k=1}^N \tr(Z_{M+j}\hat{\rho})-\bra{Q^k}Z_{M+j}\ket{Q^k}| \geq \frac{N}{200}]\\
        &\leq 2\exp{-\frac{N}{18(200)^2}}.
    \end{align*}
    Taking $N=c\log(1/\delta)$ with $c$ being a large enough constant yields $\Pr[|a(Q_{M+j})-q(Z_{M+j})| \geq \frac{1}{200}] \leq \delta/2$. Taking union bound over this event and the event that $|I|=\Omega(M)$, we get that with probability at least $1-\delta$,
    \begin{align}
        a(Z_{M+j}) &\geq q(Z_{M+j})-\frac{1}{200}\\
        &= \frac{1}{N}\sum_{Q\in S} [\mathbbm{1}(Q_{M+j}=1)] - \frac{1}{N}\sum_{Q\in S} [\mathbbm{1}(Q_{M+j}\neq 1)]-\frac{1}{200}\\
        &= \Pr_{Q\sim S}[Q_{M+j}=1] - \Pr_{Q\sim S}[Q_{M+j}\neq 1]-\frac{1}{200}\\
        &\geq \frac{199}{200} - \frac{1}{200}\\
        &\geq 0.99.
    \end{align}
    Also, since $Q^1...Q^M$ are uniformly distributed over $\{0,1\}^M$, $\trace(Z_{M+j}\rho)=0$. This yields the final result.
\end{proof}

\subsection{Numerical Experiment}
\label{sec:numerics}
We consider the same setting as the adaptive attack. For our numerical experiments, we fix the number of samples $N=10000$ and vary the number of queries $M$ that are asked. The quantum state is determined by the majority rule in Equation~\eqref{eq: majority-rule} and is uniform among pure states that satisfy the rule, as in Equation~\eqref{eq:rho-dist}. 
Recall that one can view this setting as a classical problem involving $n$ coins, whose joint distribution is dictated by the diagonal elements of $\rho$. Querying $Z_i$ observables gives the outcome of flipping one of the coins.
The goal is then to predict the true biases of the queried coins.
We simulate this equivalent classical problem in our numerical experiments.

We simulate the same adaptive attack, where $Z_1,...,Z_M$ are queried followed by an adaptive query $Z_{M+j^*}$. The adaptive error is taken as the error of the adaptively chosen query. We run this for $100$ random instances for $M=\{100, 200, 400, 800, 1600, 3200, 6400, 10000\}$ and compute the mean error and standard deviation for each $M$.  We also compare this to a non-adaptive querying rule running on the same instances, where the queries are fixed ahead of time. Specifically, we query $Z_1,...,Z_M$ and then $Z_{M+j}$ where $Q_{M+j}$ is conditioned on three-quarters of the first $M$ qubits. The non-adaptive error is computed by taking the maximum error over these queries. Our numerical results in Figure~\ref{fig:adaptive-attack} show that the adaptive error quickly becomes significantly higher than the non-adaptive error as $M$ grows.
This confirms our theoretical result from Theorem~\ref{thm:adapt-attack} and emphasizes how false discovery can occur in the experiments.

\subsection{Lower Bound}
\label{sec:lb}

In this section, we prove a sample complexity lower bound, which demonstrates that in general, we cannot design a mechanism that can accurately estimate $M$ adaptively chosen local observables given $\mathrm{poly\,log}(M)$ copies of the unknown quantum state $\rho$.
In particular, we show that it is impossible to obtain the logarithmic dependence on the number of queries $M$ present in the classical shadow formalism for the nonadaptive setting~\cite{huang2020predicting}.
We prove the following lower bound.

\begin{prop}[Local observables, lower bound]
    \label{prop:local-lower}
    There exists an analyst $\mathcal{A}$ and density matrix $\rho$ with system size $n=\Omega(M2^{2000\sqrt{M}})$ such that any $(0.99,1-\frac{1}{2000\sqrt{M}})$-accurate quantum mechanism $\mathcal{M}$ estimating expectation values of $M$ adaptively chosen single-qubit observables requires at least $N=\Omega(\sqrt{M})$ samples of the quantum state $\rho$.
\end{prop}

Similarly to the adaptive attack in Appendix~\ref{sec:attack}, we consider the setting in which the analyst queries single-qubit observables $\{Z_i\}$ and the mechanism has access to $N$ copies of a diagonal density matrix $\rho$.
The idea behind the proof of this lower bound builds on lower bounds found in the classical adaptive statistical query setting~\cite{hardt2014preventing, dwork2015preserving,steinke2015interactive}.
In particular, \cite{steinke2015interactive} showed an $\Omega(\sqrt{M})$ sample complexity lower bound for the adaptive statistical query setting.
We provide a similar construction for an analyst interacting with a quantum mechanism, which achieves this lower bound even when the analyst only queries single-qubit observables.
The main tool for the proof will be \textit{interactive fingerprinting codes}~\cite{fiat2001dynamic}, which we reviewed in Section~\ref{sec:fingerprint}.
The argument is similar to the proof of a sample complexity lower bound for the adaptive statistical query setting~\cite{steinke2015interactive}.

The idea is to reduce our problem to \textsf{IFPC} by designing an analyst $\mathcal{A}$ such that if the mechanism $\mathcal{M}$ answers all of the analyst's queries accurately, then we contradict the definition of $\mathcal{F}$ being a collusion resilient fingerprinting code.
In particular, we design an attack on the mechanism that simulates the \textsf{IFPC} game (Algorithm~\ref{alg:adaptive}), where the mechanism $\mathcal{M}$ plays the role of the adversary $\mathcal{P}$ in the \textsf{IFPC} game.
We argue that in our attack, the mechanism $\mathcal{M}$ indeed adheres to the same conditions as the adversary $\mathcal{P}$ (i.e., only having access to information from the non-accused colluding users) so that the guarantees of the interactive fingerprinting codes hold.
These guarantees state that the fingerprinting code $\mathcal{F}$ can force the adversary $\mathcal{P}$ to be inconsistent.
Then, because $\mathcal{M}$ is the same as the adversary $\mathcal{P}$, this implies that $\mathcal{M}$ also answers one of the analyst's queries incorrectly.

\begin{proof}[Proof of Proposition~\ref{prop:local-lower}]
We first define the setup of the problem.
Throughout this proof, we use $\mathcal{A}$ for the analyst, $\mathcal{M}$ for the mechanism, $\mathcal{F}$ for the fingerprinting code, and $\mathcal{P}$ for the adversary in the \textsf{IFPC} game.
Recall that the mechanism is given copies of an unknown quantum state $\rho$, with which it is asked to predict $\tr(O_i \rho)$ for an (adaptively chosen) query $O_i$ from the analyst.
We first precisely define an adversarial choice of $\rho$.
Subsequently, we will present an algorithm for $\mathcal{A}$ that thwarts any mechanism, even when only querying single-qubit observables.

Let $d$ be the number of users in the setting of fingerprinting codes.
Let $M$ be the number of observables that the analyst queries.

\textbf{Defining the density matrix:}
We define the density matrix $\rho$ on $n$ qubits, where we specify $n$ shortly.
In fact, we define an ensemble of density matrices $\{\rho_\ell\}$, where the density matrix $\rho$ given to the mechanism is chosen from the ensemble uniformly at random.
Each density matrix $\rho_\ell$ in the ensemble is diagonal so that it corresponds to a classical distribution.
We effectively encode the adversarial choice of distribution from~\cite{steinke2015interactive} into this density matrix with technical changes to ensure the guarantees hold in our setting.
Moreover, note that we use a diagonal density matrix because the guarantees from interactive fingerprinting codes hold for classical adversaries (i.e., classical mechanisms).
We later remark that quantum mechanisms can be reduced to classical mechanisms.

Let $n = \lceil \log d \rceil + M2^d$.
We index the first $\lceil \log d \rceil$ qubits by an integer in $\{1,\dots, \lceil \log d \rceil\}$, and we index the remaining qubits by tuples $(k, j)$ such that $k \in [2^d], j \in [M]$.
This reflects the underlying structure that the last qubits are $M$ groups of $2^d$ qubits.
In this way, $j$ indexes the group of qubits and $k$ indexes the specific qubit within the $j$th group.
We describe the intuition for this structure later.

Similarly to the adaptive attack (Appendix~\ref{sec:attack}), since the density matrices $\rho_\ell$ in the ensemble we construct is diagonal, we can write
\begin{equation}
    \rho_\ell \triangleq \sum_{Q \in \{0,1\}^n} p_\ell(Q) \ketbra{Q},
\end{equation}
where $\ket{Q}$ is a computational basis state so that $Q \in \{0,1\}^n$.
Note that $\ell \in \{1,\dots, (2^d!)^M\}$, so our ensemble consists of $(2^d!)^M$ density matrices.
The reason for this will become clear later.
We can index the entries of the bitstring $Q$ as previously described:
\begin{equation}
    Q = Q_1\cdots Q_{\lceil \log d \rceil} Q_{(1,1)} \cdots Q_{(2^d,1)} \cdots Q_{(1,M)} \cdots Q_{(2^d,M)} \in \{0,1\}^n.
\end{equation}

The state $\ket{Q}$ can be viewed as follows.
The first $\lceil \log d \rceil$ qubits index the user from the setting of fingerprinting codes, where recall there are $d$ users.
Explicitly, each integer in $\{1,\dots, d\}$ can be represented by in binary using $\lceil \log d \rceil$ bits.
The bitstring $Q_1 \cdots Q_{\lceil \log d \rceil} \in \{0,1\}^{\lceil \log d \rceil}$ encodes this binary representation of the user's index $i \in \{1,\dots, d\}$\footnote{These $\lceil \log d \rceil$ qubits can actually be removed, but we include them for simplicity. They also do not affect the overall system size much.}.
There are also $M$ groups of $2^d$ qubits.
Namely, for $j \in [M]$, these are the groups of qubits corresponding to $Q_{(1,j)}\cdots Q_{(2^d, j)}$.
Each group corresponds to a round $j$ of the interaction between the analyst and mechanism, where at round $j$, the analyst will query a qubit in group $j$.

We want to define the distribution $p_\ell(Q)$ in a specific way such that the first $\lceil \log d \rceil$ qubits determine the state of the remaining ones.
Let $\sigma_{j,\ell}: \{0,1\}^d \to [2^d]$ be a random assignment of $q \in \{0,1\}^d$ to some index $k_j \in [2^d]$ at round $j$.
Here, we suggestively denote a bistring in $\{0,1\}^d$ by $q$, where as hinted before, the analyst will query one of the $2^d$ qubits in a given group $j$.
When we consider a randomly chosen $\rho$ from the ensemble $\{\rho_\ell\}$, we will often drop the $\ell$ subscript on $\sigma_{j,\ell}$ and simply write $\sigma_j$.
Each density matrix $\rho_\ell$ corresponds to a random assignment $\sigma_\ell$ defined by each of the assignments $\sigma_{1,\ell},\dots, \sigma_{M, \ell}$.
Namely,
\begin{equation}
    \sigma_\ell(q^1,\dots, q^{2^d}) = (\sigma_{1,\ell}(q^1),\dots,\sigma_{1,\ell}(q^{2^d}),\dots, \sigma_{M, \ell}(q^1),\dots, \sigma_{M,\ell}(q^{2^d})).
\end{equation}
Thus, $\ell \in \{1,\dots, (2^d!)^M\}$ because this is the number of random permutations of $[M2^d]$.
Define a set of binary strings
\begin{equation}
    \label{eq:r-ell}
    R_\ell \triangleq \{Q \in \{0,1\}^n : Q_{(\sigma_{j,\ell}(q), j)} = q_i\quad \forall q \in \{0,1\}^d,\;\forall j\in [M],\text{ where } Q_1\cdots Q_{\lceil \log d \rceil} = i \in [d]\}.
\end{equation}
This defines precisely how the first $\lceil \log d \rceil$ qubits determine the remaining ones.
Namely, if the first $\lceil \log d \rceil$ bits are a binary representation of the user's index $i \in [d]$, then the other bits in the bitstring $Q \in \{0,1\}^n$ must adhere to the rule that the bit indexed by $(\sigma_{j,\ell}(q), j)$ is equal to $q_i$ for all $q \in \{0,1\}^d$ and for all rounds $j \in [M]$.
With this, we can define $p_\ell(Q)$ as follows:
\begin{equation}
    p_\ell(Q) \triangleq \begin{cases}
        \frac{1}{d} & \text{if } Q \in R_\ell\\
        0 & \text{otherwise}.
    \end{cases}
\end{equation}
In this way, we have a uniform distribution over all users, and the remaining $M2^d$ qubits must satisfy the rule specified in the definition of $R_\ell$ (Eq.~\eqref{eq:r-ell}).

We have completed the definition of this ensemble $\{\rho_\ell\}$ of density matrices.
Recall that the density matrix $\rho$ given to the mechanism is chosen from this ensemble uniformly at random.

\begin{figure}[t]
    \centering
    \includegraphics[scale=0.2]{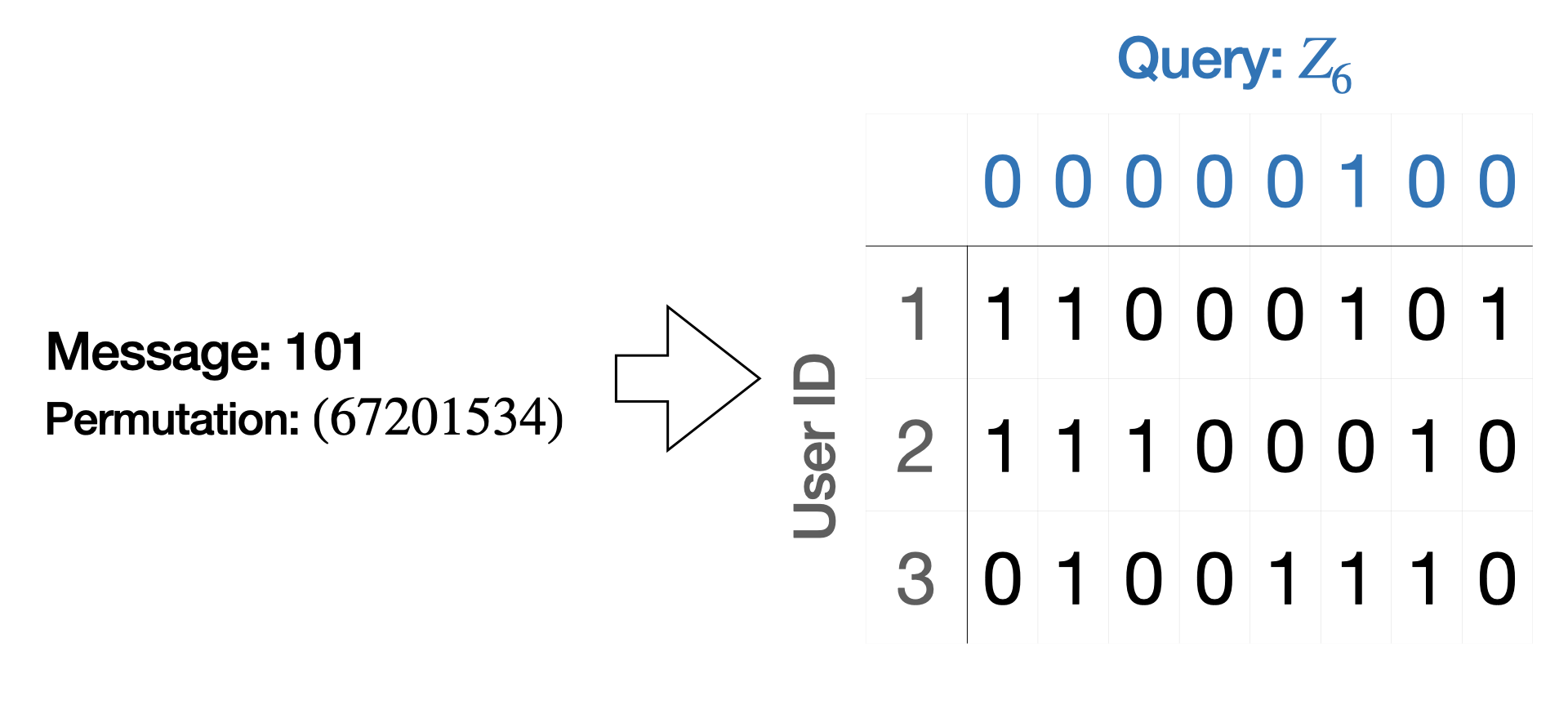}
    \caption{\textbf{Example construction for local observables.} We consider an example where the query (message) is $q^j=101$ and the randomly generated permutation is $\sigma_j = (6,7,2,0,1,5,3,4)$, where we denote each bitstring $q$ by its decimal representation. The table shows the quantum state, with each row being a computational basis state $\ket{Q}$ corresponding to a user (which is encoded in the first $\lceil \log d \rceil$ qubits of $\ket{Q}$). The permutation is encoded into the quantum state as shown in the table, e.g. the first column is $q=110$ (corresponding to $6$). The message is $q^j=101$ (corresponding to $5$) is on the $6$th entry of the permutation, so the analyst will query the single qubit observable $Z_6$.}
    \label{fig:local-observable}
\end{figure}

\textbf{Defining the analyst's behavior:}
With this, we can define the analyst's adversarial behavior.
Similarly to the adaptive attack (Appendix~\ref{sec:attack}), the analyst can cause the mechanism to answer queries inaccurately by querying only single-qubit observables.
In particular, the analyst will query single-qubit $Z$ observables.

Recall that in our setup of the density matrix, there are $M$ groups of $2^d$ qubits.
Each group corresponds to a round $j$ of the interaction between the analyst and mechanism, where at round $j$, the analyst will query the observable $Z_{(k, j)}$ for some $k \in [2^d]$.
Intuitively, at round $j$, the analyst queries a qubit chosen from $(1,j),\dots, (2^d, j)$ uniformly at random and independently for each round.
We describe this explicitly as follows.
Let $q^j \in \{0,1\}^d$ be the query specified by the analyst in round $j$.
Let $\sigma_{j}: \{0,1\}^d \to [2^d]$ be a random assignment of the query $q^j$ to some index $k_j \in [2^d]$ for round $j$.
Then, the analyst will query the observable $Z_{(\sigma_j(q^j),j)}$ at round $j$.
Note that this randomization via $\sigma_j$ is necessary to ensure that the guarantees of the fingerprinting code hold, and this plays the same role as encryption does in~\cite{steinke2015interactive}.
We justify this more formally later in the proof.

Now, we are ready to state the analyst's full attack.
This is given in Algorithm~\ref{alg:real}.

\begin{algorithm}
\caption{Real Attack (Local Version)}
\label{alg:real}
Let the number of users be $d=2000N$.\\
Sample some $n$-qubit quantum state $\rho$ from the ensemble $\{\rho_\ell\}$.\\
Give $\rho^{\otimes N}$ to the mechanism $\mathcal{M}$.\\
Initialize an $N$-collusion resilient fingerprinting code $\mathcal{F}$ for $d$ users and length $M=\mathcal{O}(N^2)$.\\
Let $T^0=\emptyset$.\\
\For{$j=1$ to $M$}{
  Let $c^j\in\{0,1\}^d$ be chosen by $\mathcal{F}$.\\
  Analyst $\mathcal{A}$ chooses query $q^j\in \{0,1\}^d$ such that $q^j_i=c^j_i$ if $i\not\in T^{j-1}$ and $0$ otherwise.\\
  $\mathcal{A}$ queries $Z_{(\sigma_j(q^j),j)}$ to $\mathcal{M}$, where $\sigma_j(q^j)$ denotes the index corresponding to $q^j$.\\
  $\mathcal{M}$ outputs response $a^j$, rounds $a^j$ to $\Bar{a}^j\in\{0,1\}$, and gives $\Bar{a}^j$ to $\mathcal{F}$.\\
  $\mathcal{F}$ accuses $I^j \subseteq [d]$.\\
  Set $T^j \gets T^{j-1} \cup I^j$.
}
\end{algorithm}

Here, we use $T^j$ to denote the set of users accused by the fingerprinting code $\mathcal{F}$ up to and including round $j$.
We also denote $S^j = \mathcal{S} \setminus T^j$, where $\mathcal{S}$ is the set of colluding users.
Thus, $S^j$ is the set of colluding users that have not yet been accused.

\textbf{Reduction to \textsf{IFPC}:}
We first show that our proposed attack follows the same interaction model as the \textsf{IFPC} game.
Crucially, we claim that from the mechanism's perspective, querying the observable $Z_{(\sigma_j(q^j),j)}$ that corresponds to $q^j$ (as in Algorithm~\ref{alg:real}) is effectively the same as querying a random observable $Z_{(i,j)}$ for which the observed $Q_{(i,j)}$ is the same as $c^j_{S^j}$ (as required in Algorithm~\ref{alg:adaptive}).
We also claim that this modification is enough to cause the game to operate in the \textsf{IFPC} model.

To show this, we construct a simulated attack, given in Algorithm \ref{alg:simulated}.
The simulated attack is merely fictitious, as it relies on the analyst knowing the set of colluding users $\mathcal{S}$.
The analyst first simulates sampling from the density matrix $\rho$ that would be given to the mechanism in the real attack (Algorithm~\ref{alg:real}) and gives this simulated data to the mechanism.
In this way, because the analyst knows $\mathcal{S}$, the dataset that is given to the mechanism in the simulated attack (Algorithm~\ref{alg:simulated}) only contains information about the colluding users.
Explicitly, the analyst creates computational basis states $\ket{Q}$, where $Q_1 \cdots Q_{\lceil \log d \rceil}$ is the binary representation of some colluding user's index $i \in \mathcal{S}$ and $Q_{\sigma_j(q), j} = q_i$ for all $q \in \{0,1\}^d, j \in [M]$.
The random assignments $\sigma_j$ are generated in the same manner as described previously (namely, the same as in the real attack (Algorithm~\ref{alg:real}).
Moreover, the simulated interaction with the sampled states were constructed with only knowledge of $c^j_{S^j}$ and contain no information about $c^j$ outside of $S^j$.
Thus, the simulated attack indeed operates in the correct \textsf{IFPC} interaction model required for fingerprinting code guarantees to hold.

\begin{algorithm}
\caption{Simulated Attack (Local Version)}\label{alg:simulated}
Let the number of users be $d=2000N$ and let $\mathcal{U}$ be the uniform distribution over $[d]$.\\
Choose colluding users $\mathcal{S} \sim \mathcal{U}^N$ and give $\mathcal{S}$ to the analyst $\mathcal{A}$.\\
$\mathcal{A}$ simulates sampling from the density matrix $\rho$ by generating $\mathcal{S}'=\{\ket{Q}\}$ corresponding to $\mathcal{S}$.\\
The analyst's procedure for generating $\mathcal{S}'$ is described further in the text.\\
Give $\mathcal{S}'$ to $\mathcal{M}$.\\
Initialize $N$-collusion resilient fingerprinting code $\mathcal{F}$ for $d$ users and length $M=\mathcal{O}(N^2)$.\\
Let $S^1=\mathcal{S}$.\\
\For{$j=1$ to $M$}{
  Let $c^j\in\{0,1\}^d$ be chosen by $\mathcal{F}$. Give $c^j_{S^j}$ to $\mathcal{A}$.\\
  $\mathcal{A}$ chooses a query at random from among $\{Z_{(\sigma_j(q),j)}:q_{S^j}=c^j_{S^j}\}$ and gives it to $\mathcal{M}$.\\
  $\mathcal{M}$ outputs response $a^j$. Round $a^j$ to $\Bar{a}^j\in\{0,1\}$, and give $\Bar{a}^j$ to $\mathcal{F}$.\\
  $\mathcal{F}$ accuses $I^j \subseteq [d]$. Set $S^{j+1} \gets S^j \setminus I^j$.
}
\end{algorithm}

To show that our proposed attack (Algorithm~\ref{alg:real}) follows the same interaction model as the \textsf{IFPC} game, we claim that the mechanism cannot distinguish between our proposed attack and the simulated attack (Algorithm~\ref{alg:simulated}), which as described must be in the \textsf{IFPC} interaction model.
We can show that any mechanism cannot distinguish between the simulated and real attacks by noting the following:
\begin{enumerate}
    \item The distribution over the assignments $\sigma_j$ are the same. 
    
    \item For any given round $j$, the distribution of the queried qubit is the same.
\end{enumerate}

The first point is true because the assignments were randomized uniformly before the first round in both the real and simulated attacks.

The second point is true because of the following argument.
For the real attack (Algorithm~\ref{alg:real}), the assignments $\sigma_j$ are randomized before the analyst selects $q^j$ (only the analyst knows this random seed).
Thus, from the mechanism's perspective, the queried qubit is effectively random, just as in the simulated attack.
Note also that we could not argue this if we only had one copy of $\{Q_{(\cdot,1)}\}$ (or if we did not randomize each of the $M$ groups independently).
If this were the case, the mechanism could learn whether they were receiving the same queries between rounds, so the indices would not truly be random.

Thus, the mechanism's views are identical in both attacks.
Since the distribution over the dataset and indices presented to the mechanism $\mathcal{M}$ are the same, the distribution over answers from $\mathcal{M}$ are also the same.
This implies that the probabilities of inconsistency are the same, i.e., \begin{equation}
    \Pr_{\mathrm{Real}}[\exists j \in [M]: \forall i \in [d],\;\Bar{a}^j\neq c^j_i] = \Pr_{\mathrm{Sim}}[\exists j \in [M]: \forall i \in [d],\; \Bar{a}^j\neq c^j_i]
\end{equation}
By the argument in~\cite{steinke2015interactive} for the algorithm for interactive fingerprinting codes, the distribution over challenges $c^j$ are also the same for the real and simulated attacks.
Thus, we have
\begin{equation}
    \Pr_{\mathrm{Real}}[|T^M\setminus \mathcal{S}|\leq n/2000] = \Pr_{\mathrm{Sim}}[|T^M\setminus \mathcal{S}|\leq n/2000]
\end{equation}
This completes the first part of the proof.

\textbf{Conclusion:}
To finish the proof, we need to show that the analyst can force the mechanism to be inaccurate after $\mathcal{O}(N^2)$ adaptive queries.
To do so, our argument follows that of the attack sketch (without cryptography) from Lecture 19 of \cite{roth2017adaptive} with a some modification.
Throughout the $M$ rounds of interaction between the analyst and mechanism, the guarantees of the fingerprinting code state that with probability $1-1/n$, the following hold:
\begin{enumerate}
    \item The number of false accusations does not exceed $n/2000$, i.e., $|T^j|\leq N+n/2000=n/1000$.

    \item The mechanism $\mathcal{M}$ was forced to be inconsistent, in which case either $c^j_i=1$ for all $i$ or $c^j=-1$ for all $i$.
\end{enumerate}
During these rounds, either
\begin{equation}
    \mathbb{E}_\mathcal{D}[Z^j] \geq 1-2\frac{|T^j|}{n} \geq 1 - \frac{1}{500} \quad \textrm{or} \quad \mathbb{E}_\mathcal{D}[Z^j] \leq -1,
\end{equation}
where $Z^j$ denotes the query that was asked in round $j$.
Suppose that the mechanism answers queries with error at most $\alpha = 0.99$.
Then, in the first case, the answer $a^j > 0$, and in the second case, $a^j < 0$.
In both cases, the rounded answers $\bar{a}^j$ are consistent.
However, the guarantees of the fingerprinting code state that the probability that $\bar{a}^j$ is consistent on every round is at most $1-1/n$.
Thus, the rest of the time, the mechanism was forced to answer the query with error $\alpha > 0.99$.
\end{proof}

We emphasize some observations we made throughout the proof.

\begin{remark}
    Note that the randomization of the qubits in our proposed attack (Algorithm \ref{alg:real}) plays an equivalent role to the encryption in \cite{steinke2015interactive}.
    Namely it ensures the secrecy requirement needed for the IFPC model for queries outside of $\mathcal{S}$.
\end{remark}
\begin{remark}
    The guarantee in this proof are for classical mechanisms $\mathcal{M}$.
    However, since any quantum mechanism can be simulated classically and the arguments of the proof hold without restriction to the time complexity of the mechanism, the guarantee holds for any quantum mechanism as well.
\end{remark}

We also point out that the same lower bounds hold for the following analogous classical adaptive data analysis problem.
Suppose the distribution $\mathcal{D}$ is over an $n$-dimensional data universe $\mathcal{X}\subseteq \{0,1\}^n$, and the analyst is constrained to using only \textit{$k$-sparse statistical queries} $q:\mathcal{X}\mapsto [0,1]$, which we define as functions depending on at most $k$ bits of the input, $k$ being a constant.
Given a dataset $\mathcal{S}=(x_1,...,x_N) \sim \mathcal{D}^N$, the goal of the classical mechanism is to accurately estimate $q(\mathcal{D})=\mathbb{E}_{x\sim \mathcal{D}}[q(x)]$ for $M$ adaptively-chosen queries $q$.
Since our proof for the local observables lower bound involves only classical constructions, we may recycle the proof to obtain new classical lower bounds for this sparse queries problem: considering an analyst querying $1$-sparse queries of the form $q_b(x)=x_b$, the query $q_b$ corresponds to the observable $Z_b$ and the classical distribution $\mathcal{D}$ can be mapped onto a density matrix.

\begin{corollary} [Sparse queries, lower bound]
\label{cor:local_queries}
    There exists an analyst $\mathcal{A}$ and classical distribution $\mathcal{D}$ acting on $\{0,1\}^n$, where $n=\Omega(M2^{2000\sqrt{M}})$ such that any $(0.99,1-\frac{1}{2000\sqrt{M}}$-accurate classical mechanism $\mathcal{M}$ answering $M$ adaptively chosen $1$-sparse statistical queries requires at least $N=\Omega(\sqrt{M})$ from $\mathcal{D}$.
\end{corollary}

\subsection{Upper Bound}
\label{sec:classical shadow ub}

In this section, we prove a sample complexity upper bound for answering adaptive queries specified by local observables.
Recall that in our setting, the analyst's queries are specified by a local observable $O_i$, and the mechanism wishes to estimate expectation values $\tr(O_i \rho)$ accurately given copies of the unknown quantum state $\rho$.
We design a mechanism that can successfully estimate these expectation values, even when the observable $O_i$ is chosen adaptively.
In the case where $O_i$ are chosen non-adaptively, the median of means~\cite{jerrum1986random,nemirovskij1983problem} protocol used in \cite{huang2020predicting} is sufficient. However, we will need to use a different algorithm for when $O_i$ are chosen adaptively.
Moreover, while the proof of our sample complexity lower bound (Appendix~\ref{sec:lb}) only required the analyst to query single-qubit observables, our algorithm works for an analyst querying any local observables.
This upper bound matches our lower bound from Proposition~\ref{prop:local-lower} exactly in $M$ scaling.
In particular, we prove the following proposition.

\begin{prop}[Local observables, upper bound]
    \label{prop:local-upper}
    For every $\epsilon,\delta \in (0,1/2)$ such that $\epsilon\delta \geq 4e^{-M/\log^2M \log\log^4M}$, there exists an $(\epsilon,\delta)$-accurate mechanism $\mathcal{M}$ that can estimate expectation values of $M$ adaptively chosen $k$-local observables with
    \begin{equation}
        N=\mathcal{O}\left(\frac{\min\{3^k\sqrt{M\log(1/\epsilon\delta)},k\log (n/\delta)\}}{\epsilon^2}\right)
    \end{equation}
    samples of the unknown $n$-qubit quantum state $\rho$. The mechanism runs in time $\mathrm{poly}(N, n, \log(1/\delta))$ per observable.
\end{prop}

The idea of the proof is to find the classical shadow (see Appendix~\ref{sec:classical-shadow} for an overview) of the unknown quantum state $\rho$ with respect to which we want to estimate expectation values.
Then, we can apply algorithms from classical adaptive data analysis~\cite{bassily2015algorithmic} to obtain our guarantees.
Here, the classical shadow formalism allows us to export our quantum problem to the classical setting, where we can leverage algorithms from the classical literature.
Moreover, in the case of local observables, classical shadows allow us to estimate expectation values of queried observables (computationally) efficiently.

We will require the following classical algorithm from adaptive data analysis~\cite{bassily2015algorithmic}.

\begin{theorem} [Corollary 1.3 in~\cite{dagan2021boundednoise}] \label{thm:adapt algo}
    For every $\epsilon,\delta \in (0,1/2)$ such that $\epsilon\delta \geq 4e^{-M/\log^2M \log\log^4M}$, there is an $(\epsilon,\delta)$-accurate mechanism $\mathcal{M}$ for $M$ adaptively chosen (scaled) statistical queries $q:\mathcal{X}\mapsto [-C,C]$ that uses
    \begin{equation}
        N=\mathcal{O}\left(\frac{C\sqrt{M\log(1/\epsilon\delta)}}{\epsilon^2}\right)
    \end{equation}
    samples from $\mathcal{X}$. The mechanism runs in time $\mathrm{poly}(N, n, \log(1/\delta))$ per observable.
\end{theorem}

Here, we use Corollary 1.3 in~\cite{dagan2021boundednoise} applied to $C/N$-sensitive queries.
We use this theorem to prove Proposition~\ref{prop:local-upper}.

\begin{proof}[Proof of Proposition~\ref{prop:local-upper}]
    To derive the upper bound for local observables, we apply the algorithm of Theorem~\ref{thm:adapt algo} on classical shadows data obtained from random Pauli measurements.
    Recall from Appendix~\ref{sec:classical-shadow} that classical shadows obtained via random Pauli measurements allow us to predict expectation values of local observables well, as desired.

    First, we claim that the setting of the classical shadow problem aligns with the classical statistical query model from Appendix~\ref{sec:stat-query}\footnote{Note that in Appendix~\ref{sec:stat-query}, we consider statistical queries $q : \mathcal{X} \to [0,1]$, but this setup is the same for scaled statistical queries $q : \mathcal{X} \to [-C,C]$.}.
    Recall that in the classical statistical query model, we wish to design an algorithm (the mechanism) to answer statistical queries $q: \mathcal{X} \to [-C,C]$, whose true value is given by
    \begin{equation}
        q(\mathcal{D}) = \E_{x \sim \mathcal{D}}[q(x)],
    \end{equation}
    where $\mathcal{D}$ is an unknown distribution over the data universe $\mathcal{X}$. The mechanism has access to a sample $\mathcal{S} = \{x_1,\dots, x_N\} \sim \mathcal{D}^N$, which it can use in order to estimate $q(\mathcal{D})$.
    Moreover, statistical queries must satisfy
    \begin{equation}
        \E_{\mathcal{S} \sim \mathcal{D}^N} [q(\mathcal{S})] = q(\mathcal{D}),
    \end{equation}
    where $q(\mathcal{S}) = \frac{1}{N}\sum_{x_i \in \mathcal{S}} q(x_i)$.

    The setting of the classical shadow formalism for predicting local observables is indeed analogous to the statistical query model as follows
    Let the data universe $\mathcal{X}$ be defined as $\{0,1\}^m$ for $m = \mathcal{O}(n)$, where $n$ is the system size.
    One classical snapshot of the $n$-qubit unknown quantum state $\rho$ from random Pauli measurements can be stored in $m$ bits, as justified in Appendix~\ref{sec:classical-shadow}.
    Let $\mathcal{D}$ be the distribution from which the classical shadow data is sampled from.
    Namely, $\mathcal{D}$ is the distribution over the random choice of Pauli gates and the probabilistic outcomes from measuring $\rho$.
    Recall from Appendix~\ref{sec:classical-shadow} that we can write
    \begin{equation}
        \label{eq:true-exp-rho}
        \rho = \E\left[\mathcal{M}^{-1}\left(U^\dagger \hat{\ket{b}}\hat{\bra{b}} U\right)\right] = \E_{\hat{\rho} \sim \mathcal{D}}[\hat{\rho}],
    \end{equation}
    where $U$ is a random tensor product of $n$ single-qubit Pauli gates, $\hat{b} \in \{0,1\}^n$ is the result of measuring $U\rho U^\dagger$ in the computational basis, and $\mathcal{M}$ is a quantum channel.
    Here, we denote
    \begin{equation}
        \label{eq:rho-hat}
        \hat{\rho} = \mathcal{M}^{-1}\left(U^\dagger \hat{\ket{b}}\hat{\bra{b}} U\right).
    \end{equation}
    Then, our sample $\mathcal{S} \sim \mathcal{D}^N$ is the classical shadow $\mathcal{S} = \mathsf{S}(\rho; N) = \{\hat{\rho}_1,\dots, \hat{\rho}_N\}$ from repeating the random measurement procedure $N$ times to obtain $N$ classical snapshots.
    The queries are functions $o: \mathcal{X} \to [-C,C]$ (we justify later why the outputs falls in $[-C,C]$) corresponding to an observable $O$, which are evaluated on a classical snapshot $\hat{\rho} \in \mathcal{X}$ as $o(\hat{\rho}) = \tr(O \hat{\rho})$.
    The true value of the query is given by
    \begin{equation}
        o(\mathcal{D}) = \E_{\hat{\rho} \sim \mathcal{D}}[o(\hat{\rho})] = \E_{\hat{\rho} \sim \mathcal{D}}[\tr(O \hat{\rho})] = \tr(O \rho),
    \end{equation}
    where the last equality follows by Equation~\eqref{eq:true-exp-rho}.
    This aligns well with the classical statistical query model, since we wish to estimate $o(\mathcal{D}) = \tr(O \rho)$.
    Moreover, these functions are indeed statistical queries, as we have
    \begin{equation}
        \E_{\mathcal{S} \sim \mathcal{D}^N}[o(\mathcal{S})] = \frac{1}{N}\sum_{j=1}^N \E_{\mathcal{S} \sim \mathcal{D}^N}[\tr(O \hat{\rho}_j)] = o(\mathcal{D}),
    \end{equation}
    where $o(\mathcal{S}) = \frac{1}{N}\sum_{j=1}^N o(\hat{\rho}_j)$.
    To complete this analogy with the classical statistical query model, it remains to show that $o(\hat{\rho}) \in [-C,C]$ for some constant $C$ and for $\hat{\rho} \in \mathcal{X}$ when $O$ is a local observable.
    Let $\hat{\rho}$ be a classical snapshot as in Equation~\eqref{eq:rho-hat}.
    In the case of predicting local observables, the random unitary $U$ utilized to construct this classical snapshot is a random tensor product of $n$ single-qubit Pauli gates $U = U_1 \otimes \cdots \otimes U_n$.
    Moreover, from Appendix~\ref{sec:classical-shadow}, the quantum channel $\mathcal{M}$ has a specific form:
    \begin{equation}
        \hat{\rho} = \bigotimes_{j=1}^n \left(3U_j^\dag \ketbra{\hat{b}_j}U_j-\mathbb{I}\right) \quad \mathrm{where} \quad \hat{b}_1,...,\hat{b}_n \in \{0,1\}.
    \end{equation}
    We want to show that $|o(\hat{\rho})| = |\tr(O \hat{\rho})|$ is bounded by a constant for a local observable $O$.
    Let $O = \tilde{O} \otimes \mathbb{I}^{\otimes (n-k)}$ be a $k$-local observable, where we assume without loss of generality that $O$ acts nontrivially only on the first $k$ qubits.
    Then, we can compute the expectation value $\tr(O\hat{\rho})$ by plugging in Equation~\eqref{eq:local-channel}:
    \begin{align}
        \trace(O\hat{\rho}) &= \trace\left(O\bigotimes_{j=1}^n \left(3U_j^\dag \ketbra{\hat{b}_j}U_j-\mathbb{I}\right)\right)\\
        &=\trace\left(\left(\Tilde{O}\bigotimes_{j=1}^k \left(3U_j^\dag \ketbra{\hat{b}_j}U_j-\mathbb{I}\right)\right)\otimes \left(\bigotimes_{j=k+1}^n \left(3U_j^\dag \ketbra{\hat{b}_j}U_j-\mathbb{I}\right)\right)\right)\\
        &= \trace\left(\left(\Tilde{O}\bigotimes_{j=1}^k \left(3U_j^\dag \ketbra{\hat{b}_j}U_j-\mathbb{I}\right)\right)\right)\tr\left( \left(\bigotimes_{j=k+1}^n \left(3U_j^\dag \ketbra{\hat{b}_j}U_j-\mathbb{I}\right)\right)\right)\\
        &=\trace\left(\Tilde{O}\bigotimes_{j=1}^k \left(3U_j^\dag \ketbra{\hat{b}_j}U_j-\mathbb{I}\right)\right).
    \end{align}
    In the first line, we used Equation~\eqref{eq:local-channel}.
    In the second line, we used that $O = \tilde{O} \otimes \mathbb{I}^{\otimes (n-k)}$.
    In the third line, we used the multiplicative properties of the trace.
    Finally, in the last line, we used that $\tr(\hat{\rho}) = 1$, where we can view this the second term in the third line as a classical snapshot over $n-k$ qubits.
    Now, by properties of the trace norm, we also have
    \begin{align}
        |\trace(O\hat{\rho})| &\leq \norm{\Tilde{O}\bigotimes_{j=1}^k \left(3U_j^\dag \ketbra{\hat{b}_j}U_j-\mathbb{I}\right)}_1\\
        &\leq \norm{\Tilde{O}}_\infty \cdot \prod_{j=1}^k \norm{3U_j^\dag \ketbra{\hat{b}_j}U_j-\mathbb{I}}_1 \\
        &\leq 3^k,
    \end{align}
    where in the last line we used the assumption that the spectral norm of $\tilde{O}$ is bounded by $1$.
    Thus, we have that $|o(\hat{\rho})| \leq C = 3^k$, where this is a constant because $k = \mathcal{O}(1)$ by definition of a local observable.
    
    With this, we have now shown that the classical shadow framework fits into the classical statistical query setting.
    Thus, we can apply any classical algorithm for answering adaptive statistical queries to derive our rigorous guarantees.
    Applying Theorem~\ref{thm:adapt algo}, we see that in order to estimate expectation values of $M$ adaptively chosen $k$-local observables, we can use
    \begin{equation}
    N=\mathcal{O}\left(\frac{3^k\sqrt{M\log(1/\epsilon\delta)}}{\epsilon^2}\right)
    \end{equation}
    samples of the unknown $n$-qubit quantum state $\rho$.
    This gives us the first term in the minimum in the statement of Proposition~\ref{prop:local-upper}.

    To obtain the second term of $k \log n$, we note that any observable $O$ can be decomposed in the Pauli basis as $O = \sum_{P \in \{I, X, Y, Z\}^{\otimes n}} \alpha_P P$ for some coefficients $\alpha_P$.
    For a $k$-local observable, there are only $3^k \cdot \binom{n}{k} = \mathcal{O}(n^k)$ possible different Pauli terms $P$ in this decomposition.
    One can predict expectation values of all of these observables with respect to the unknown state $\rho$ using the guarantees of the classical shadow formalism, which requires
    \begin{equation}
        N = \mathcal{O}\left(\frac{k\log (n/\delta)}{\epsilon^2}\right).
    \end{equation}
    This follows from Theorem~\ref{thm:classical-shadow}, where the shadow norm is constant for local observables (Proposition S3 in \cite{huang2020predicting}).
    Because these $3^k \cdot \binom{n}{k}$ observables are fixed beforehand, we can use these guarantees.
    Expectation values of any (adaptively chosen) local observable can be expressed in terms of expectations of these $\mathcal{O}(n^k)$ observables, so we can predict any (adaptively chosen) local observable using this classical shadow data.
    Taking the minimum of these two expressions yields the sample complexity claim.

    Finally, the time complexity follows because, as discussed previously, each classical snapshot $\hat{\rho}$ can be stored efficiently in classical memory using $\mathcal{O}(n)$ bits.
\end{proof}

\section{Pauli Observables}
\label{app: pauli}

In this section, we shift our attention to focus on predicting properties $\tr(O\rho)$, where $\rho$ is the unknown $n$-qubit quantum state that we are given copies of and $O \in \{I,X,Y,Z\}^{\otimes n}$ is a Pauli observable.
This type of observable was not explicitly considered in the classical shadow formalism~\cite{huang2020predicting}.
However,~\cite{huang2021information} designed a quantum machine learning model uses $N = \mathcal{O}(\log(M)/\epsilon^4)$ copies of $\rho$ to predict expectation values of $M$ $n$-qubit Pauli observables.
We analyze how the rigorous guarantees change when the Pauli observables are allowed to be chosen adaptively.
In this appendix, we prove Theorem~\ref{thm:pauli}, which provides sample complexity upper and lower bounds for this task.
We restate the theorem below.
\begin{theorem}[Pauli observables, detailed restatement of Theorem~\ref{thm:pauli}]
    There exists an analyst $\mathcal{A}$ and a density matrix $\rho$ on $n = \Omega(M^{3/2})$ qubits such that any $(0.99, 1-\frac{1}{2000\sqrt{M}})$-accurate quantum mechanism $\mathcal{M}$ estimating expectation values of $M$ adaptively chosen Pauli observables requires at least 
    \begin{equation}
        N = \Omega(\sqrt{M})
    \end{equation}
    samples of the quantum state $\rho$.
    Meanwhile, for any $\epsilon,\delta \in (0,1/2)$ such that $\epsilon\delta \geq 4e^{-M/\log^2M \log\log^4 M}$, there exists an $(\epsilon,\delta)$-accurate quantum mechanism $\mathcal{M}$ that can estimate expectation values of $M$ adaptively chosen Pauli observables using
    \begin{equation}
        N = \mathcal{O}\left(\frac{\min\{\sqrt{M\log(1/\epsilon\delta)}, n\}}{\epsilon^4}\right)
    \end{equation}
    samples of the unknown $n$-qubit quantum state $\rho$.
    Moreover, the mechanism is computationally efficient and runs in time $\mathrm{poly}(N, n, \log(1/\delta))$ per observable.
\end{theorem}

The techniques used in this section are similar to Appendix~\ref{app: local}.
In Appendix~\ref{app:pauli-lower}, we prove the sample complexity lower bound, which states that in general, there does not exist a quantum mechanism that can accurately estimate $M$ adaptively chosen Pauli observables with $\mathrm{poly\,log}(M)$ copies of $\rho$.
In Appendix~\ref{app:pauli-upper}, we present a quantum algorithm achieving a matching sample complexity upper bound for answering adaptively chosen Pauli observables.

\subsection{Lower Bound}
\label{app:pauli-lower}

In this section, we prove the sample complexity lower bound, which states that in general, there does not exist a quantum mechanism that can accurately estimate $M$ adaptively chosen Pauli observables with $\mathrm{poly\,log}(M)$ copies of $\rho$.
This gives a stark contrast to the algorithm in~\cite{huang2021information}, which can accomplish this task using $\mathcal{O}(\log M/\epsilon^4)$ copies of $\rho$.
We prove the following lower bound.

\begin{prop}[Pauli observables, lower bound]
    \label{prop:pauli-lower}
    There exists an analyst $\mathcal{A}$ and a density matrix $\rho$ on $n = \Omega(M^{3/2})$ qubits such that any $(0.99, 1-\frac{1}{2000\sqrt{M}})$-accurate quantum mechanism $\mathcal{M}$ estimating expectation values of $M$ adaptively chosen Pauli observables requires at least 
    \begin{equation}
        N = \Omega(\sqrt{M})
    \end{equation}
    samples of the quantum state $\rho$.
\end{prop}

This proof is similar to the proof of Proposition~\ref{prop:local-lower}, which gave a sample complexity lower bound for the task of predicting local observables.
Again, this proof relies on interactive fingerprinting codes~\cite{fiat2001dynamic} (see Appendix~\ref{sec:fingerprint}) and is similar to lower bounds in classical adaptive data analysis~\cite{steinke2015interactive,dwork2015preserving,hardt2014preventing}.
Moreover, our construction achieves this lower bound even when the analyst can only query Pauli observables $\{I, Z\}^{\otimes n}$ and when the mechanism is given $N$ copies of a diagonal $n$-qubit density matrix $\rho$.

Our proof will require some cryptographic concepts, which we review in Appendix~\ref{app:crypto}.

The idea is again to reduce our problem to \textsf{IFPC} by designing an analyst $\mathcal{A}$ such that if the mechanism $\mathcal{M}$ answers all of the analyst's queries accurately, then we contradict the definition of $\mathcal{F}$ being a collusion resilient fingerprinting code.
We argue that in our attack, the mechanism $\mathcal{M}$ adheres to the same conditions as the adversary $\mathcal{P}$ in the \textsf{IFPC} game (i.e., only having access to information from the non-accused colluding users) so that the guarantees of the interactive fingerprinting codes hold.
These guarantees state that the fingerprinting code $\mathcal{F}$ can force the adversary $\mathcal{P}$ to be inconsistent, which also implies that $\mathcal{M}$ also answers one of the analyst's queries incorrectly.

In this case, to show that $\mathcal{M}$ only has access to information from the non-accused colluding users, we use encryption, similarly to~\cite{steinke2015interactive}.
Namely, we will encrypt all information given to the mechanism via a private-key encryption scheme (reviewed in Appendix~\ref{app:crypto}) and encode the secret keys needed to decrypt in the samples given to the mechanism.
In this way, the mechanism can only decrypt the information corresponding to the colluding users, and all other information is hidden by the security of the private-key encryption scheme.


\begin{proof}[Proof of Proposition~\ref{prop:pauli-lower}]
    We follow the same proof structure as the proof of Proposition~\ref{prop:local-lower}.
    Throughout the proof, we use $\mathcal{A}$ for the analyst, $\mathcal{M}$ for the mechanism, $\mathcal{F}$ for the fingerprinting code, and $\mathcal{P}$ for the adversary in the \textsf{IFPC} game.
    Let $d$ be the number of users in the setting of fingerprinting codes.
    Let $M$ be the number of observables that the analyst queries.
    Also, throughout the proof, we use the notation $x_{i:j}$, where $x = (x_1,\dots, x_m)$ is a vector, to denote $x_{i:j} \triangleq (x_i, \dots, x_j)$, where $j \geq i$. 
    
    Recall that the mechanism is given copies of an unknown quantum state $\rho$, with which it is asked to predict $\tr(O_i \rho)$ for an (adaptively chosen) query $O_i$ from the analyst.
    We first define an encryption scheme that we will use.
    We then precisely define the adversarial choice of $\rho$ based on this encryption scheme.
    Subsequently, we will present an algorithm for $\mathcal{A}$ that thwarts any mechanism, even when only querying Pauli observables in $\{I, Z\}^{\otimes n}$.

    \begin{figure}[t]
    \centering
    \includegraphics[scale=0.2]{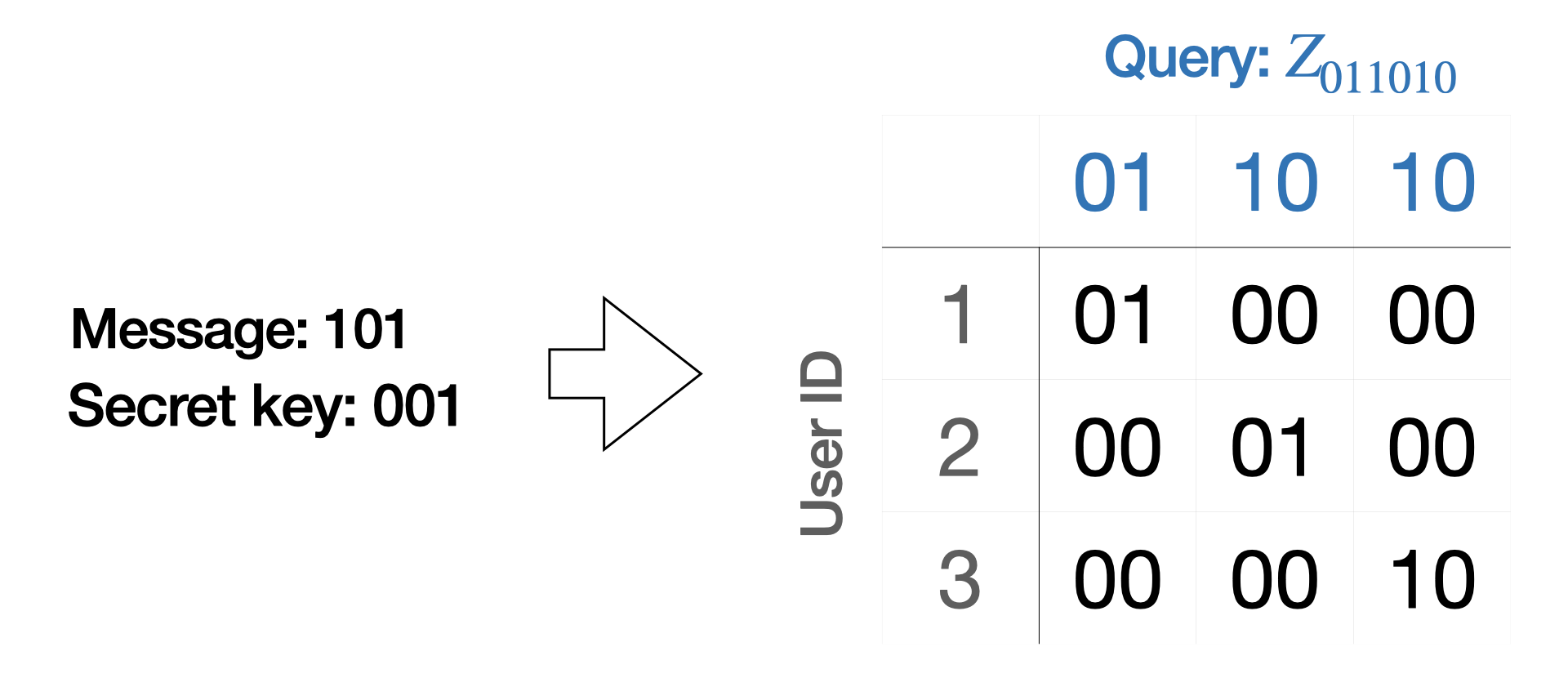}
    \caption{\textbf{Example construction for Pauli observables.} We consider an example where the query (message) is $q^j=101$ and the randomly generated secret key is $001$. The table shows the quantum state, with each row being a computational basis state $\ket{Q}$ corresponding to a user (which is encoded in the first $\lceil \log d \rceil$ qubits of $\ket{Q}$). The secret key is encoded into the quantum state as shown in the table. We group qubits into pairs in each for clarity, which are defined by the encoding rule for the density matrix. By the encryption rule, the analyst will query the Pauli observable $Z_{b^j}$ where $b^j=011010$. One can verify that evaluating $Z_{b^j}$ on each state recovers the original message.}
    \label{fig:pauli-observable}
    \end{figure}
    
    \textbf{Encryption scheme:} The encryption scheme is a version of the one-time pad (Definition~\ref{def:otp}).
    Namely, we use a variation of the one-time pad with $(\Gen, \Enc, \Dec)$ defined as folows. Let $\lambda = d$ be the security parameter.
    Let the key and message spaces be $\{0,1\}^d$ and let the ciphertext space be $\{0,1\}^{2d}$.
    \begin{itemize}
        \item $\Gen(1^d)$: Output a uniformly random bitstring $\sk \in \{0,1\}^d$.
        \item $\Enc_\sk(m)$: Encrypt $m \in \{0,1\}^d$ as $c = c_1\cdots c_{2d} \in \{0,1\}^{2d}$, where
        \begin{equation}
            c_i = \begin{cases}
                10 & \text{if } m_i \oplus \sk_i = 1\\
                01 & \text{if } m_i \oplus \sk_i = 0
            \end{cases}.
        \end{equation}
        \item $\Dec_\sk(c)$: On input $c \in \{0,1\}^{2d}$, parse $c$ as $c = c_{1:2}\cdots c_{2d-1:2d}$, i.e., consider consecutive pairs of bits in $c$. 
        Map $c$ to $c' = c_1'\cdots c_d'$, where
        \begin{equation}
            c'_i = \begin{cases}
                1 & \text{if } c_i = 10\\
                0 & \text{if } c_i = 01
            \end{cases}.
        \end{equation}
        Decrypt $c$ as $m = c' \oplus \sk \in \{0,1\}^d$.
    \end{itemize}
    Note that this is indeed the same as the one-time pad, but with an extra step that maps the output of encryption to a bitstring of length $2d$ and correspondingly maps back for decryption.

    \textbf{Defining the density matrix:} 
    We define the density matrix $\rho$ on $n$ qubits, where we specify $n$ shortly.
    This density matrix is diagonal so that it corresponds to a classical distribution.
    As in the proof of Proposition~\ref{prop:local-lower}, we encode an adversarial choice of classical distribution into this density matrix that takes inspiration from~\cite{steinke2015interactive}.
    Moreover, note that we use a diagonal density matrix because the guarantees from interactive fingerprinting codes hold for classical adversaries (i.e., classical mechanisms).
    We later remark that quantum mechanisms can be reduced to classical mechanisms.

    Let $n = \lceil \log d \rceil + 2Md$.
    We index the first $\lceil \log d \rceil$ qubits by an integer in $\{1,\dots, \lceil \log d \rceil\}$, and we index the remaining qubits by tuples $(k,j)$ such that $k \in [2d], j \in [M]$.
    This reflects the underlying structure that the last qubits are $M$ groups of $2d$ qubits.
    In this way, $j$ indexes the group of qubits and $k$ indexes the specific qubit within the $j$th group.
    We describe the intuition for this structure later.

    Since $\rho$ is a diagonal density matrix, we can write
    \begin{equation}
        \rho \triangleq \sum_{Q \in \{0,1\}^n} p(Q) \ketbra{Q},
    \end{equation}
    where $\ket{Q}$ is a computational basis state so that $Q \in \{0,1\}^n$.
    We can index the entries of the bitstring $Q$ as previously described:
    \begin{equation}
        Q = Q_1 \cdots Q_{\lceil \log d \rceil} Q_{(1,1)} \cdots Q_{(2d,1)} \cdots Q_{(1,M)} \cdots Q_{(2d,M)} \in \{0,1\}^n.
    \end{equation}
    The state $\ket{Q}$ can be viewed as follows.
    The first $\lceil \log d \rceil$ qubits index the user from the setting of fingerprinting codes, where recall there are $d$ users.
    Explicitly, each integer in $\{1,\dots, d\}$ can be represented by in binary using $\lceil \log d \rceil$ bits.
    The bitstring $Q_1 \cdots Q_{\lceil \log d \rceil} \in \{0,1\}^{\lceil \log d \rceil}$ encodes this binary representation of the user's index $i \in \{1,\dots, d\}$.
    There are also $M$ groups of $2d$ qubits.
    Namely, for $j \in [M]$, these are the groups of qubits corresponding to $Q_{(1,j)}\cdots Q_{(2d, j)}$.
    Each group corresponds to a round $j$ of the interaction between the analyst and mechanism.
    Moreover, each pair of qubits corresponding to $Q_{(2i-1,j)}$ and $Q_{(2i, j)}$ represent user $i$ in round $j$.

    We want to define the distribution $p(Q)$ in a specific way such that the first $\lceil \log d \rceil$ qubits determine the state of the remaining ones and such that it encodes the secret keys for our encryption scheme.
    Let $\sk^j \leftarrow \Gen(1^d)$ be the secret key for round $j \in [M]$.
    We generate a new secret key for each round of interaction between the analyst and mechanism.
    We use $\sk^j_i$ to denote the $i$th bit of the secret key $\sk^j$ for round $j$, i.e., the one-bit secret key for the $i$th user in round $j$.
    We define the following sets of binary strings:
    \begin{equation}
        R_0 \triangleq \{Q \in \{0,1\}^n : Q_{(2i-1,j)}Q_{(2i,j)} = 10,\; \forall j \in [M], \text{ where } Q_1 \cdots Q_{\lceil \log d \rceil} = i \in [d] \text{ and } \sk_i^j = 0\}
    \end{equation}
    \begin{equation}
        R_1 \triangleq \{Q \in \{0,1\}^n : Q_{(2i-1,j)}Q_{(2i,j)} = 10,\; \forall j \in [M], \text{ where } Q_1 \cdots Q_{\lceil \log d \rceil} = i \in [d] \text{ and } \sk_i^j = 1\}
    \end{equation}
    \begin{equation}
        R_2 \triangleq \{Q \in \{0,1\}^n : Q_{(2i - 1, j)} Q_{(2i, j)} = 00,\quad \forall i \in [d],\; \forall j \in [M]\}
    \end{equation}
    \begin{equation}
        \label{eq:R}
        R \triangleq R_0 \cup R_1 \cup R_2.
    \end{equation}
    With this, we can define $p(Q)$ as follows
    \begin{equation}
        p(Q) \triangleq \begin{cases}
            \frac{1}{d} & \text{if } Q \in R\\
            0 & \text{otherwise}
        \end{cases}.
    \end{equation}
    This defines precisely how the secret key is encoded into the density matrix.
    We have completed the definition of this density matrix $\rho = \sum_{Q \in \{0,1\}^n} p(Q) \ketbra{Q}$.

    \textbf{Defining the analyst's behavior:} With this, we can define the analyst's adversarial behavior.
    We will show that the analyst can cause the mechanism to answer queries inaccurately when querying Pauli observables in $\{I, Z\}^{\otimes n}$.
    We use $Z_b$ for $b \in \{0,1\}^n$ to denote the observable $Z_b = O_1 \otimes \cdots \otimes O_n$, where $O_i = Z$ if $b_i = 1$ and $O_i = \mathbb{I}$ otherwise.
    Moreover, we use the notation $b^j \in \{0,1\}^n$ to denote the choice of query in the $j$th round.
    We denote $b^j_{(i,j)}$ to denote the entry of the vector $b^j$ corresponding to the qubit indexed by $(i,j)$.
    We also use the notation $b^j_{(i:i', j)}$ to denote the entries of the vector $b^j$ corresponding to the qubits indexed by $(i, j),(i+1,j),\dots, (i', j)$ for $i' \geq i$ and $i, i' \in [2d]$.

    Again, note that the guarantees from interactive fingerprinting codes will only give us guarantees for classical mechanisms.
    Thus, we argue that these types of queries also have a classical analogue.
    Namely, measuring the $n$-qubit state $\rho$ is equivalent to flipping $n$ biased coins and obtaining a sample outcome $x \in \{0,1\}^n$.
    Then, querying the observable $Z_b$ is the same as querying the parity over a selected group of coins indexed by $b$, i.e., $p(b,x) = b \cdot x \bmod 2$.

    The analyst's full attack is given in Algorithm~\ref{alg:real2}.
    \begin{algorithm}
    \caption{Real Attack (Pauli Version)}\label{alg:real2}
    Let the number of users be $d=2000N$.\\
    Let $(\Gen, \Enc, \Dec)$ be the encryption scheme defined in the text.\\
    Generate secret keys $\sk^j \leftarrow \Gen(1^d)$ for each round $j \in [M]$.\\
    Give $\rho^{\otimes N}$ to the mechanism $\mathcal{M}$ for $\rho$ defined as in the text.\\
    Initialize $N$-collusion resilient fingerprinting code $\mathcal{F}$ for $n$ users and length $M=\mathcal{O}(N^2)$.\\
    Let $T^0=\emptyset$.\\
    \For{$j=1$ to $M$}{
      Let $c^j\in\{0,1\}^d$ be chosen by $\mathcal{F}$.\\
      Define $q^j\in \{0,1\}^d$ such that $q^j_i=c^j_i$ if $i\not\in T^{j-1}$ and $0$ otherwise.\\
      Analyst $\mathcal{A}$ chooses query $Z_{b^j}$ where $b^j_{(1:2d,j)}=\Enc(\sk^j,q^j)$ and $b^j$ is $0$ for the other entries.\\
      $\mathcal{A}$ queries $Z_{b^j}$ to $\mathcal{M}$.\\
      $\mathcal{M}$ outputs response $a^j$, rounds $a^j$ to $\Bar{a}^j\in\{0,1\}$, and gives $\Bar{a}^j$ to $\mathcal{F}$.\\
      $\mathcal{F}$ accuses $I^j \subseteq [d]$. Set $T^j \gets T^{j-1} \cup I^j$.
    }
    \end{algorithm}
    Here, we use $T^j$ to denote the set of users accused by the fingerprinting code $\mathcal{F}$ up to and including round $j$.
    We also denote $S^j = \mathcal{S} \setminus T^j$, where $\mathcal{S}$ is the set of colluding users.
    Thus, $S^j$ is the set of colluding users that have not yet been accused.
    
    Note that this attack is almost the same as Algorithm~\ref{alg:real} from the local observables sample complexity lower bound.
    The only difference is that the queries that the analyst $\mathcal{A}$ gives to the mechanism $\mathcal{M}$ are encrypted.
    In particular, in round $j$, given a query $q^j \in \{0,1\}^d$, the analyst instead gives the query $Z_{b^j}$ to the mechanism, where $b^j_{(1:2d, j)} = \Enc_{\sk^j}(q^j)$ and $b^j$ is $0$ for the remaining entries.
    This is so that the mechanism only has access to information from the set of colluding users $S^j$ that have not yet been accused, as we will justify shortly.
    This encryption replaces the random permutations from the proof of Proposition~\ref{prop:local-lower}.

    We also justify that for an input query $q^j$, encrypting it to obtain $b^j$, and querying $Z_{b^j}$ as in Algorithm~\ref{alg:real2} still results in the same query $q^j$.
    Let $Q^{(i)} \in \{0,1\}^n$ be the bitstring from measuring $\rho$ in the computational basis, where $Q^{(i)}_1 \cdots Q^{(i)}_{\lceil \log d \rceil} = i$ and let $z^j_{(2i-1:2i, j)} = Q^{(i)}_{(2i-1,j)}Q_{(2i,j)}$.
    Let $(\Gen^1, \Enc^1, \Dec^1)$ be the encryption scheme $(\Gen, \Enc, \Dec)$ defined above but for security parameter $\lambda = 1$.
    In other words, we use $(\Gen^1, \Enc^1, \Dec^1)$ to denote the special case of $(\Gen, \Enc, \Dec)$ when the key space and message space are $\{0,1\}$ and the ciphertext space is $\{0,1\}^2$.
    Let $\sk^j \leftarrow \Gen(1^d)$.
    Note that because $\sk^j$ is generated uniformly at random in $\Gen$, then each bit $\sk_i^j$ of $\sk^j$ is also a key that could be generated from $\Gen^1$.
    Then, notice that by definition of the encryption scheme and the set $R$ (defined in Equation~\eqref{eq:R}), we have $z^j_{(2i-1:2i, j)} = \Enc^1_{\sk_i^j}(1)$.
    Moreover, we can also write
    \begin{equation}
        \Dec^1_{\sk_i^j}(c) = \Enc^1_{\sk_i^j}(1) \cdot c \bmod 2,
    \end{equation}
    which can be easily checked by definition of the encryption scheme.
    Recall that querying the observable $Z_b$ is the same as querying the parity $p(b,x) = b \cdot x \bmod 2$, where $x$ is a sample outcome from measuring $\rho$.
    In this case, then by definition of $b^j$ in Algorithm~\ref{alg:real2}, we have
    \begin{equation}
        b^j_{(2i-1:2i, j)} \cdot z^j_{(2i-1:2i, j)} = \Dec_{\sk_i^j}(\Enc_{\sk_i^j}(q_i^j)) = q_i^j.
    \end{equation}
    In this way, we see that querying $Z_{b^j}$ is the same as querying the unencrypted $q^j$.
    Thus, the encryption here is only hiding information but not obscuring what we wish to query.

    \textbf{Reduction to \textsf{IFPC}:}
    We now show that our proposed attack follows the same interaction model as the \textsf{IFPC} game.
    Crucially, we claim that the mechanism can only access information from the set of colluding users $S^j$ that have not been accused, i.e., the information $c^j_{S^j}$ (as required in Algorithm~\ref{alg:interact}).

    To show this, we construct a simulated attack, given in Algorithm~\ref{alg:simulated2}.
    The simulated attack is merely fictitious, as it relies on the analyst knowing the set of colluding users $\mathcal{S}$.
    In particular, for $i \notin \mathcal{S}$, the analyst's simulated attack will set $q_i^j = 0$, which is not possible in the real attack because the analyst does not know $\mathcal{S}$.
    The analyst first simulates sampling from the density matrix $\rho$ that would be given to the mechanism in the real attack (Algorithm~\ref{alg:real2}) and gives this simulated data to the mechanism.
    In this way, because the analyst knows $\mathcal{S}$, the dataset that is given to the mechanism in the simulated attack (Algorithm~\ref{alg:simulated}) only contains information about the colluding users.
    Explicitly, the analyst creates computational basis states $\ket{Q}$, where $Q \in R$ (defined in Equation~\eqref{eq:R}), where $i$ is restricted to be in $\mathcal{S}$ instead.
    Moreover, the simulated interaction with the sampled states were constructed with only knowledge of $c^j_{S^j}$ and contain no information about $c^j$ outside of $S^j$.
    Thus, the simulated attack indeed operates in the correct \textsf{IFPC} interaction model required for fingerprinting code guarantees to hold.

    \begin{algorithm}
    \caption{Simulated Attack (Pauli Version)}\label{alg:simulated2}
    Let the number of users be $d=2000N$ and let $\mathcal{U}$ be the uniform distribution over $[d]$.\\
    Let $(\Gen, \Enc, \Dec)$ be the encryption scheme defined in the text.\\
    Choose colluding users $\mathcal{S} \sim \mathcal{U}^N$ and give $\mathcal{S}$ to the analyst $\mathcal{A}$.\\
    $\mathcal{A}$ generates secret keys $\sk^j \leftarrow \Gen(1^d)$ for each round $j \in [M]$.\\
    $\mathcal{A}$ simulates sampling from the density matrix $\rho$ by generating $\mathcal{S}'=\{\ket{Q}\}$ corresponding to $\mathcal{S}$.\\
    Give $\mathcal{S}'$ to $\mathcal{M}$.\\
    Initialize $N$-collusion resilient fingerprinting code $\mathcal{F}$ for $n$ users and length $M=\mathcal{O}(N^2)$.\\
    Let $S^1=\mathcal{S}$.\\
    \For{$j=1$ to $M$}{
      Let $c^j\in\{0,1\}^d$ be chosen by $\mathcal{F}$. Give $c^j_{S^j}$ to $\mathcal{A}$.\\
      Define $q^j \in \{0,1\}^d$ such that $q^j_i=c^j_i$ if $i \in S^j$ and $0$ otherwise.\\
      Analyst $\mathcal{A}$ chooses query $Z_{b^j}$ where $b^j_{(1:2d,j)}=\Enc(\sk^j,q^j)$ and $b^j$ is $0$ for the other entries.\\
      $\mathcal{A}$ queries $Z_{b^j}$ to $\mathcal{M}$.\\
      $\mathcal{M}$ outputs response $a^j$. Round $a^j$ to $\Bar{a}^j\in\{0,1\}$, and give $\Bar{a}^j$ to $\mathcal{F}$.\\
      $\mathcal{F}$ accuses $I^j \subseteq [d]$. Set $S^{j+1} \gets S^j \setminus I^j$.
    }
    \end{algorithm}

    To show that our proposed attack (Algorithm~\ref{alg:real2}) follows the same interaction model as the \textsf{IFPC} game, we claim that the mechanism cannot distinguish between our proposed attack and the simulated attack (Algorithm~\ref{alg:simulated2}), which as described must be in the \textsf{IFPC} interaction model.
    We can show that any mechanism cannot distinguish between the simulated and real attacks by noting the following:
    \begin{enumerate}
        \item The distribution over the measurement outcomes for qubits indexed by $(k, j)$ for $k \in [2d], j \in [M]$ is the same.
        \item For a given round $j$, the distribution of the queries $b^j$ is the same.
    \end{enumerate}

    The first point is true because both the users (by definition of $\rho$) and secret keys are chosen uniformly at random in both the simulated and real attacks.
    Together, these completely determine the rest of the $2dM$ qubits, so their distirbutions are the same.

    The second point is true because of the following argument.
    Notice that the mechanism only sees the secret keys for $i\in \mathcal{S}$, i.e., it observes $z^j_{(2i-1:2i, j)}=\Enc^1(\sk^j_i,1)$ only for $i\in \mathcal{S}$.
    This is clear by construction.
    Moreover, notice that the only difference between the real and simulated attacks with respect to the query $q^j$ is that $q^j$ (and thus $b^j$) is only different outside of $\mathcal{S}$.
    Thus, although the mechanism can know the secret keys for $i \in \mathcal{S}$, $b^j$ is in fact the same for $i \in \mathcal{S}$, so the distribution is the same in this case.
    On the other hand, for $i \notin \mathcal{S}$, the queries $q^j$ may be different.
    However, recall that $b^j$ is the encryption of $q^j$, where our encryption scheme is defined to be a variant of the one-time pad.
    Thus, by the perfect secrecy of the one-time pad, then the distribution over $b^j$ are identical in this case as well.

    Thus, the mechanism's views are identical in both attacks.
    This implies that the probabilities of inconsistency are the same, i.e., \begin{equation}
        \Pr_{\mathrm{Real}}[\exists j \in [M]: \forall i \in [d],\;\Bar{a}^j\neq c^j_i] = \Pr_{\mathrm{Sim}}[\exists j \in [M]: \forall i \in [d],\; \Bar{a}^j\neq c^j_i].
    \end{equation}
    The rest of the proof is the same as the proof of Proposition~\ref{prop:local-lower}.
\end{proof}

Just as in the local observables case, there is also an analogous classical adaptive data analysis problem for which we may recycle our proof in order to obtain a new lower bound.
Here, the distribution $\mathcal{D}$ is over an $n$-dimensional data universe $\mathcal{X}\subseteq \{0,1\}^n$, and the analyst queries from the family of \textit{parity queries}, which are defined the same way as noted previously in this section: $q_b(x)=p(b,x)=b\cdot x$, where $b\in \{0,1\}^n$ and $x\in \mathcal{X}$.
Our proof for the Pauli observables lower bound relies only on classical constructions, so we may recycle our proof: the parity query $q_b$ corresponds to the observable $Z_b$ and the classical distribution $\mathcal{D}$ can be mapped onto a density matrix.

\begin{corollary} [parity queries, lower bound]
\label{cor:parity_queries}
    There exists an analyst $\mathcal{A}$ and classical distribution $\mathcal{D}$ acting on $\{0,1\}^n$, where $n=\Omega(M^{3/2})$ such that any $(0.99,1-\frac{1}{2000\sqrt{M}}$-accurate classical mechanism $\mathcal{M}$ answering $M$ adaptively chosen parity queries requires at least $N=\Omega(\sqrt{M})$ from $\mathcal{D}$.
\end{corollary}

\subsection{Upper Bound}
\label{app:pauli-upper}

In this section, we prove a sample complexity upper bound for predicting expectation values $\tr(P_i\rho)$ for $i \in [M]$ of adaptively chosen $n$-qubit Pauli observables $P_i \in \{I,X,Y,Z\}^{\otimes n}$ with respect to an unknown quantum state $\rho$.
We design a mechanism that can successfully estimate these expectation values, even when the observable $P_i$ is chosen adaptively.
In the case where $P_i$ are chosen non-adaptively, one can use the algorithm in~\cite{huang2021information}.
However, we will need to use a different algorithm for when $P_i$ are chosen adaptively.
Moreover, while the proof of our sample complexity lower bound (Appendix~\ref{app:pauli-lower}) only required the analyst to query observables in $\{I, Z\}^{\otimes n}$, our algorithm works for an analyst querying any Pauli observables.
This upper bound matches our lower bound from Proposition~\ref{prop:pauli-lower} exactly in $M$ scaling.
In particular, we prove the following proposition.

\begin{prop}[Pauli observables, upper bound]
    \label{prop:pauli-upper}
    For every $\epsilon,\delta \in (0,1/2)$ such that $\epsilon\delta \geq 4e^{-M/\log^2M \log\log^4M}$, there is exists an $(\epsilon,\delta)$-accurate mechanism $\mathcal{M}$ that can estimate expectation values of $M$ adaptively chosen Pauli observables using
    \begin{equation}
        N=\mathcal{O}\left(\frac{\min\{\sqrt{M\log(1/\epsilon\delta)},n\log(1/\delta)\}}{\epsilon^4}\right)
    \end{equation}
    samples of the unknown quantum state $\rho$. The algorithm runs in time $\mathrm{poly}(N,n,\log(1/\delta))$ per query.
\end{prop}

The idea of the proof is similar to the proof of Proposition~\ref{prop:local-upper}.
Namely, we can replace a classical subroutine in the algorithm in~\cite{huang2021information} with an algorithm from classical adaptive data analysis~\cite{bassily2015algorithmic} (see Theorem~\ref{thm:adapt algo}).
This subroutine is the only part of the algorithm from~\cite{huang2021information} that is susceptible to adaptivity, so by utilizing the algorithm in~\cite{bassily2015algorithmic} instead, we can protect against this.
We refer to Appendix~\ref{sec:nonadapt-pauli} for a review of the algorithm from~\cite{huang2021information}.

\begin{proof}[Proof of Proposition~\ref{prop:pauli-upper}]
    Recall from Appendix~\ref{sec:nonadapt-pauli} that the algorithm from~\cite{huang2021information} estimates $\tr(P_i\rho)$ in two steps: (1) estimate the magnitude $\lvert\trace(P_i\rho)\rvert^2$ and (2) determine $\mathrm{sign}(\trace(P_i\rho))$.

    As discussed in Appendix~\ref{sec:nonadapt-pauli}, the first step is completed by estimating expectations of
    \begin{equation}
        q_P(S^{(t)})=\prod_{k=1}^n \trace\left((\sigma_k\otimes\sigma_k)S^{(t)}_k\right) \in \{\pm 1\},
    \end{equation}
    where $S^{(t)} = \{S_k^{(t)}\}_{k \in [n]}$ are classical measurement results with $t \in [N_1]$ the $t$'th round of measurement and $P = \sigma_1 \otimes \cdots \sigma_n$ (Equation~\eqref{eq:pauli-query} in Appendix~\ref{sec:nonadapt-pauli}).
    \cite{huang2021information} shows that the expectation value of $q_P$ is exactly the magnitude $\lvert\trace(P\rho)\rvert^2$ and thus estimates this expectation to within $\epsilon^2$ error using the empirical mean of $q_P$.
    However, taking the empirical mean is not sufficient in the adaptive case, so we will need to use a different algorithm, similarly to the proof of Proposition~\ref{prop:local-upper}.

    Our key observation is to notice that $q_P$ is a statistical query with values in $\{-1,1\}$.
    Thus, we can apply the algorithm from Theorem~\ref{thm:adapt algo} with $C = 1$ to estimate the expectation of each $q_{P_i}$, i.e., the magnitudes $\lvert \tr(P_i \rho)\rvert^2$.
    Applying Theorem~\ref{thm:adapt algo}, we can estimate the magnitudes $\lvert \tr(P_i \rho)\rvert^2$ for $M$ adaptively chosen Pauli observables $P_i$ to within $\epsilon^2$ error using
    \begin{equation}
        2N_1=\mathcal{O}\left(\frac{\sqrt{M\log(1/\epsilon\delta)}}{\epsilon^4}\right)
    \end{equation}
    copies of the unknown quantum state $\rho$.

    Note that this still holds despite the subsequent post-processing of the estimation, i.e. taking the square root and estimating $\mathrm{sign}(\trace(P\rho))$, since none of them reveal to the analyst any new information about the dataset of measurement outcomes $\{S^{(t)}_k\}$ (where $\mathrm{sign}(\trace(P\rho))$ is estimated using a separate set of samples).
    Thus, effectively, the analyst only receives information about the dataset $\{S^{(t)}_k\}$ from the output of the Theorem \ref{thm:adapt algo} algorithm.
    Hence, $N_1$ samples as defined above are indeed sufficient to predict magnitudes even for adaptively chosen observables.
    
    As remarked in Appendix~\ref{sec:nonadapt-pauli}, the second step of determining $\mathrm{sign}(\tr(P\rho))$ is already resistant to adaptivity since it reuses copies of $\rho$ via an argument involving the quantum union bound~\cite{aaronson2006qma}.
    Thus, we can use the same
    \begin{equation}
        N_2 = \mathcal{O}\left(\frac{\log(M/\delta)}{\epsilon^2}\right)
    \end{equation}
    as in the original algorithm~\cite{huang2021information} for determining $\mathrm{sign}(\tr(P_i\rho))$ for $1\leq i \leq M$.
    Putting everything together, the total sample complexity is
    \begin{equation}
        \label{eq:pauli1}
        N = 2N_1 + N_2 = \mathcal{O}\left(\frac{\sqrt{M\log(1/\epsilon\delta)}}{\epsilon^4}\right).
    \end{equation}
    To obtain the second term in the minimum in the statement of Proposition~\ref{prop:pauli-upper}, we can also protect against adaptivity by simply using the original non-adaptive algorithm~\cite{huang2021information} to predict expectation values of all $4^n$ Pauli observables.
    By Theorem~\ref{thm:nonadapt-pauli}, this requires
    \begin{equation}
        \label{eq:pauli2}
        N = \mathcal{O}\left(\frac{n\log(1/\delta)}{\epsilon^4}\right).
    \end{equation}
    Because this guarantees that the algorithm will accurately estimate expectation values of all possible Pauli observables, it does not matter if an observable is chosen adaptively.
    Taking the minimum of the two expressions Equations~\eqref{eq:pauli1} and~\eqref{eq:pauli2} gives the claim.

    Finally, the computational complexity follows because the dataset of measurement outcomes $S^{(t)}$ can be stored efficiently in classical memory using $2n$ bits.
    This is clear because, as described in Appendix~\ref{sec:nonadapt-pauli}, $S^{(t)}$ is obtained by measuring $\rho \otimes \rho$ in the Bell basis.
\end{proof}

\section{Adaptive algorithms using classical shadows} \label{app: acs}
In this section, we introduce a series of algorithms for answering adaptive queries using classical shadows~\cite{huang2020predicting}.
Namely, when observables are chosen adaptively, we provide new algorithms for shadow tomography~\cite{aaronson2018shadow} in Appendix~\ref{sec:adapt cs} and quantum threshold search (originally called secret acceptor in~\cite{aaronson2016complexity}) in Appendix~\ref{sec:thresh}.
We will use these algorithms as components of the our algorithm for predicting adaptively chosen bounded-Frobenius-norm properties in Appendix~\ref{app: bfo}.
However, we also believe them to be sufficiently interesting algorithms in their own right, as they apply to arbitrary observables and their non-adaptive versions are often used as subroutines in common quantum algorithms.
Throughout this section, we assume that we have access to a classical shadow dataset $\mathcal{S}(\rho;N)=\{\hat{\rho}_1,...,\hat{\rho}_N\}$, where $\rho$ is the unknown quantum state that we have samples of.
We refer the reader to Appendix~\ref{sec:classical-shadow} for a review of the classical shadow formalism.

\subsection{Adaptive classical shadow tomography} \label{sec:adapt cs}
We again consider receiving adaptively chosen observables $O_1,...,O_M$ and are tasked with estimating $o_i=\trace(O_i\rho)$ for all $1\leq i \leq M$ given samples of the unknown quantum state $\rho$.
Previously, in Appendix~\ref{sec:classical shadow ub}, we saw that if $o_i(\hat{\rho})=\trace(O_i\hat{\rho})$ is bounded in some range $[-C,C]$ where $C$ is a constant and $\hat{\rho}$ is a classical snapshot of $\rho$, then we can use classical adaptive data analysis algorithms~\cite{bassily2015algorithmic} that apply to statistical queries $q:\mathcal{X}\mapsto [-C,C]$ to obtain rigorous sample complexity guarantees. However, in general, it is not guaranteed that the estimator $o_i(\hat{\rho})$ will have bounded range. For instance, in the case of bounded-Frobenius-norm observables, $o_i$ has an exponentially large range, and we are only given that its variance is bounded by a constant $B$ (Lemma 1 and Proposition 1 in~\cite{huang2020classical}).
For classical shadows predicting arbitrary observables, we need to consider algorithms that apply to more general estimators. We consider two algorithms for this task, one of which is computationally efficient.
These two algorithms have different sample complexity guarantees, so we present both, which may be useful in different practical scenarios.
Both of our algorithms utilize differentially private algorithms from~\cite{feldman2017generalization}.

\subsubsection{Computationally efficient mechanism}
\label{sec:comp-eff-mech}
First, we present a computationally efficient algorithm for predicting adaptively chosen properties with sample complexity $\tilde{O}(\sqrt{M})$.
The algorithm is simply a modification of the original median of means~\cite{jerrum1986random,nemirovskij1983problem} protocol used in \cite{huang2020predicting}, where the only difference is that it performs a truncation step and then estimates the median in a differentially private manner. The algorithm can be outlined as follows: 
\begin{enumerate}
    \item Break the set of classical shadows into $K$ batches of size $N$, and compute
    \begin{equation}
        \{\hat{o}_i^{(1)},...,\hat{o}_i^{(K)}\} \quad \mathrm{where} \quad \hat{o}_i^{(k)}=\frac{1}{N}\sum_{j=N(k-1)+1}^{Nk}\trace(O_i\hat{\rho}_j).
    \end{equation}
    \item Let $B=\max_{i=1,...,M}\|O_i-\frac{\trace(O_i)}{2^n}\mathbb{I}\|^2_{\mathrm{shadow}}$. Truncate each $\hat{o}_i^{(k)}$ to lie in the range $[-2\sqrt{B/N}-1,2\sqrt{B/N}+1]$, i.e., collect
    \begin{equation}
        \{\Tilde{o}_i^{(1)},...,\Tilde{o}_i^{(K)}\} \quad \mathrm{where} \quad \Tilde{o}_i^{(k)}=\begin{cases}
        -2\sqrt{B/N}-1 & \text{if } \hat{o}_i^{(k)} < -2\sqrt{B/N}-1\\
        2\sqrt{B/N}+1 & \text{if } \hat{o}_i^{(k)} > 2\sqrt{B/N}+1\\
        \hat{o}_i^{(k)} & \text{otherwise}
    \end{cases}
    \end{equation}
    \item Use any differentially private algorithm for computing the approximate median of $\{\Tilde{o}_i^{(1)},...,\Tilde{o}_i^{(K)}\}$. For example, this can be done via an application of the exponential mechanism \cite{4389483} (e.g., \cite{10.1145/1993636.1993743,feldman2017generalization}). Then, return the result $\hat{o}_i(N,K)$. 
\end{enumerate}

Our sample complexity guarantee then follows directly from Theorem 1.3 in~\cite{feldman2017generalization} and the remarks following it.
This sample complexity upper bound is looser than that of Theorem \ref{thm:adapt algo} that we used in Appendix~\ref{sec:classical shadow ub}.
In particular, it has an additional $\log M$ factor.
However, it applies more generally and allows us to bound the sample complexity in terms of the variance of the estimator.
In the setting of classical shadows, the variance is bounded by the shadow norm (Lemma 1 in~\cite{huang2020predicting}).
Thus, we have the following result.

\begin{theorem}[Computationally efficient adaptive classical shadows]\label{thm:dp cs}
    For a fixed measurement primitive $\mathcal{U}$, a sequence of adaptively chosen $2^n\times 2^n$ Hermitian matrices $O_1,...,O_M$, and accuracy parameters $\epsilon,\delta\in [0,1]$, set
    \begin{equation}
        K=\mathcal{O}\left(\sqrt{M\log(1/\delta)}\log(M/\delta \epsilon)\right) \quad \mathrm{and} \quad N=\mathcal{O}\left(\frac{1}{\epsilon^2}\max_{i=1,...,M}\|O_i-\frac{\trace(O_i)}{2^n}\mathbb{I}\|^2_{\mathrm{shadow}}\right),
    \end{equation}
    where $\|\cdot\|_{\mathrm{shadow}}$ is the shadow norm. Then, there exists an efficient algorithm using $NK$ independent classical shadows that can accurately estimate all the properties:
    \begin{equation}
        |\hat{o}_i(N,K)-\trace(O_i\rho)| \leq \epsilon \quad \textrm{for all}\ 1\leq i \leq M
    \end{equation}
    with probability at least $1-\delta$. The running time is $\mathrm{poly}(N,m,\log(1/\epsilon))$, where $m$ is the number of bits required to store a classical shadow.
\end{theorem}

Here, a measurement primitive is an ensemble of unitaries that is tomographically complete~\cite{huang2020predicting}.

\subsubsection{Private multiplicative weights}
A simple modification of the algorithm from Appendix~\ref{sec:comp-eff-mech} can yield the desired $\mathcal{O}(\log(M))$ sample complexity scaling for predicting adaptively-queried observables.
However, we exchange this favorable sample complexity at the cost of computational complexity, as this new algorithm is computationally inefficient.

In particular, to achieve the $\mathcal{O}(\log M)$ sample complexity, one only needs to modify Step 3 in the algorithm from Appendix~\ref{sec:comp-eff-mech}.
Specifically, the problem of finding the approximate median in Step 3 can be reduced to answering a sequence of statistical queries via another algorithm (e.g., \cite{feldman2017dealing,feldman2017generalization}). Then, one can simply apply the private multiplicative weights algorithm \cite{5670948} to answer these statistical queries, which gives an alternative method for solving Step 3. 
This same approach gives a computationally inefficient algorithm in~\cite{feldman2017generalization}, so our results in the classical shadow setting follow directly from Theorem 5.4 in~\cite{feldman2017generalization}.

\begin{theorem}[Computationally inefficient adaptive classical shadows] \label{thm: tomography protocol}
    For a fixed measurement primitive $\mathcal{U}$, a sequence of adaptively chosen $2^n\times 2^n$ Hermitian matrices $O_1,...,O_M$, and accuracy parameters $\epsilon,\delta\in [0,1]$, there exists an algorithm that uses
    \begin{equation}
        N=\mathcal{O}\left(\frac{\sqrt{m\log(1/\delta)}\cdot \log(M/\epsilon\delta)}{\epsilon^3}\cdot \max_{i=1,...,M}\|O_i-\frac{\trace(O_i)}{2^n}\mathbb{I}\|^3_{\mathrm{shadow}}\right)
    \end{equation}
    samples to accurately estimate all the properties:
    \begin{equation}
        |\hat{o}_i-\trace(O_i\rho)| \leq \epsilon \quad \textrm{for all}\ 1\leq i \leq M
    \end{equation}
    with probability at least $1-\delta$. Here, $m$ is the number of bits needed to store a classical shadow.
\end{theorem}

In particular, when the measurement primitive $\mathcal{U}$ consists of random $n$-qubit Clifford circuits and tensor products of random single-qubit Clifford circuits, then $m=\mathcal{O}(n^2)$ and $m=\mathcal{O}(n)$, respectively. Note that if $m=\mathcal{O}(n^2)$, this improves upon the shadow tomography protocol of \cite{badescu2020improved} by a factor of $\log M / \epsilon$ at the cost of the shadow norm term. We will use this algorithm as one of the two main components for our algorithm in Appendix~\ref{app: bfo}.

\subsection{Threshold search with classical shadows}
\label{sec:thresh}
In this section, we consider a problem called \emph{threshold search} and provide a new algorithm solving this problem.
Threshold search was first introduced by~\cite{aaronson2016complexity} where it was known as ``secret acceptor.''
In this problem, we receive a sequence of adaptively chosen queries $(O_1,\theta_1),...,(O_M,\theta_M)$ in an online fashion. Here, $O_i$ are the observables and $\theta_i\in [0,1]$ are \emph{thresholds}. Given a budget $\ell$, the task is to correctly verify whether $o_i=\trace(O_i\rho)$ exceeds the threshold $\theta_i$. Moreover, we only need to do this until the number of times the threshold is exceeded goes past the budget $\ell$ (though if this never happens then we must correctly provide answers for all $M$ queries). We introduce a new threshold search algorithm based on the classical shadows formalism, which will be used as another main component of the algorithm in Appendix~\ref{app: bfo}.

As in the previous sections, we will apply a classical adaptive data analysis algorithm (the sparse vector algorithm~\cite{10.1145/1536414.1536467,10.1561/0400000042}) to estimators on classical shadows, though with an additional preprocessing step.
For a given classical shadow dataset $\mathcal{S}(\rho;N)$, the algorithm first truncates each $o_i(\hat{\rho})$ to the range $[-T,T]$, where $T=O(\log(1/\epsilon))$. Then, it defines a statistical query estimator of $\trace(O_i\rho)$ as the sample mean of these truncated values and runs the sparse vector algorithm \cite{10.1145/1536414.1536467} using this estimator. We do the truncation step because the sparse vector algorithm works only for statistical (or, more generally, low-sensitivity) queries. Explicitly, the steps of the algorithm are as follows:
\begin{enumerate}
    \item \label{item:truncate} Gather classical shadows data $\mathcal{S}(\rho; N)=\{\hat{\rho}_1,...,\hat{\rho}_N\}$.
    Let
    \begin{equation}
        \label{eq:si-rho-hat}
        s_i(\hat{\rho})=\mathrm{sign}(\trace(O_i\hat{\rho}))\cdot \min\{|\trace(O_i\hat{\rho})|,T\}
    \end{equation}
    denote the truncation of $o_i(\hat{\rho})$ to the interval $[-T,T]$.
    \item \label{item:sparse} Run the sparse vector algorithm over the statistical queries
    \begin{equation}
        s_i(\mathcal{S})=\frac{1}{N}\sum_{j=1}^N s_i(\hat{\rho}_{(j)}).
    \end{equation}
\end{enumerate}

It is known that the sparse vector algorithm has the following guarantees, which can be seen by, e.g., applying the transfer theorem in \cite{bassily2015algorithmic} to Theorem 5.2 in \cite{feldman2017generalization}.

\begin{theorem} [\cite{10.1145/1536414.1536467,10.1561/0400000042}] \label{thm:sv}
    Given parameters $\epsilon,\delta\in [0,1]$ and a budget $\ell \in \mathbb{N}$, there is an algorithm using 
    \begin{equation}
        K=\mathcal{O}\left(\frac{T\sqrt{\ell}\log(M)\log(1/\epsilon\delta)^{3/2}}{\epsilon^2}\right)
    \end{equation}
    samples that responds to an adaptively chosen sequence of query / threshold pairs $(q_1,\theta_1),...,(q_M,\theta_M)$, where $q_i:\mathcal{X}\mapsto [-T,T]$. With probability at least $1-\delta$, the algorithm correctly does the following for all $i\in [M]$:
    \begin{enumerate}
        \item If $q_i(\mathcal{D})>\theta_i$, output "No".
        \item If $q_i(\mathcal{D})\leq \theta_i-\epsilon$, output "Yes".
        \item If the number of "No"s exceed $\ell$, terminate the algorithm.
    \end{enumerate}
\end{theorem}

Using this, we can obtain the following guarantees for our threshold search algorithm described above.

\begin{theorem}[Threshold search] \label{thm:nts}
    Given parameters $\epsilon,\delta\in [0,1]$ and a budget parameter $\ell \in \mathbb{N}$, there is an algorithm using 
    \begin{equation}
        N=\mathcal{O}\left(\frac{\sqrt{B\ell}\log(B)\log(M)\log(1/\epsilon\delta)^{5/2}}{\epsilon^2}\right) \quad \left[B=\max_{i=1,...,M}\tr(O_i^2)\right]
    \end{equation}
    samples that responds to an adaptively chosen sequence of observable/threshold pairs $(O_1,\theta_1),\dots,(O_M,\theta_M)$. With probability at least $1-\delta$, the algorithm correctly does the following for all $i\in [M]$:
    \begin{enumerate}
        \item If $\trace(O_i\rho)>\theta_i$, always output "No".
        \item If $\trace(O_i\rho)\leq \theta_i-\epsilon$, always output "Yes".
        \item If the number of "No"s exceed $\ell$, terminate the algorithm.
    \end{enumerate}
\end{theorem}

Previously, the quantum threshold search algorithm~\cite{badescu2020improved} was used to solve this problem, with the best-known sample complexity as $\mathcal{O}(\ell\log^2 M/\epsilon^2)$.
Our algorithm improves the sample complexity of quantum threshold search by a factor of $\sqrt{\ell}\log M$ in exchange for a term that scales with the bound on the Frobenius norm.
In order to obtain this guarantee, we require that the classical shadows data $\mathcal{S}(\rho; N)$ gathered is obtained via the uniform POVM procedure described in Appendix~\ref{sec:uniform-povm} rather than the canonical classical shadow procedure~\cite{huang2020predicting}.
This is due to its stronger concentration properties, which allow us to obtain a better sample complexity guarantee.
If one were to use canonical classical shadows~\cite{huang2020predicting}, a similar proof can obtain a sample complexity guarantee of $N = \tilde{\mathcal{O}}(B\sqrt{\ell}\log(M)/\epsilon^3)$.
Importantly, using classical shadows with uniform POVM allows us to reduce the scaling with $1/\epsilon$ from $1/\epsilon^3$ to $1/\epsilon^2$.  

\begin{proof}[Proof of Theorem~\ref{thm:nts}]
We first show that the bias of each $s(\hat{\rho}_{(j)})$ defined in Equation~\eqref{eq:si-rho-hat} is small (which implies the bias of $s_i(\mathcal{S})$ is small).
We write $o = \tr(O \hat{\rho})$ and $s = s(\hat{\rho})$ for brevity.
Letting $\mu=\mathbb{E}[o] \in [-1,1]$, we have that the bias $|\mathbb{E}[o]-\mathbb{E}[s]|$ satisfies the following.
\begin{align}
    &|\mathbb{E}[o]-\mathbb{E}[s]|\\
    &=|\mathbb{E}[(o-T)\mathbb{I}(o > T)-(o+T)\mathbb{I}(o < -T)]|\\
    &\leq |\mathbb{E}[(o-\mu-(T-1))\mathbb{I}(o-\mu > T-1)-(o-\mu+T-1)\mathbb{I}(o-\mu < -T+1)]|\\
    &\leq |\mathbb{E}[(o-\mu)\mathbb{I}(o-\mu > T-1)]|+|\mathbb{E}[(o-\mu)\mathbb{I}(o-\mu<-T+1)]|+(T-1)\Pr[|o-\mu|>T-1] \label{eq:26}\\
    &\leq \sqrt{\mathbb{E}(o-\mu)^2\Pr[o-\mu>T-1]} + \sqrt{\mathbb{E}(o-\mu^2)\Pr[o-\mu<-T+1]}+(T-1)\Pr[|o-\mu|>T-1] \label{eq:27}\\
    &\leq \sqrt{8B} \cdot \sqrt{2\Pr[|o-\mu|>T-1]}+(T-1)\Pr[|o-\mu|>T-1] \label{eq:28}\\
    &\leq 4\sqrt{B}\cdot \sqrt{2\exp(-\frac{(T-1)^2}{16B+4\sqrt{B}(T-1)})}+(T-1)\cdot 2\exp(-\frac{(T-1)^2}{16B+4\sqrt{B}(T-1)}) \label{eq:29}
\end{align}
where Equation\eqref{eq:26} follows from triangle inequality, Equation~\eqref{eq:27} from the Cauchy-Schwarz inequality, Equation~\eqref{eq:28} from the Cauchy-Schwarz inequality and the second moment bound in Lemma~\ref{lem: moment bound}, and Equation~\eqref{eq:29} from Proposition~\ref{prop: concentration}.
Setting $T=1+\sqrt{CB}$, it can be verified that if $C\geq 80^2(\log B+4)^2$, then 
\begin{equation}
    T-1 \leq \exp(\frac{(T-1)^2}{32B+8\sqrt{B}(T-1)}).
\end{equation}
Moreover, if $C\geq (40^2\log(12/\epsilon))^2$, then
\begin{equation} \label{eq: exp bound}
    2\exp(\frac{(T-1)^2}{32B+8\sqrt{B}(T-1)}) \leq \epsilon/6.
\end{equation}
Combining these two, we obtain that for $\epsilon < 1$, setting $T=1+40\sqrt{B}\log(48/\epsilon)(\log B + 4)$ yields
\begin{equation}
    (T-1)\cdot 2\exp(-\frac{(T-1)^2}{16B+4\sqrt{B}(T-1)}) \leq \epsilon/6.
\end{equation}
A similar computation shows that if $C \geq 40^2\log(48B/\epsilon)^2$,
\begin{equation}
    4\sqrt{B}\cdot \sqrt{2\exp(-\frac{(T-1)^2}{16B+4\sqrt{B}(T-1)})} \leq \epsilon/6.
\end{equation}
Putting everything together, choosing $T=1+40\sqrt{B}\log(48/\epsilon)(\log B + 4)$ yields
\begin{equation}
    |\mathbb{E}[o]-\mathbb{E}[s]| \leq \epsilon/3.
\end{equation}
This tells us that the truncation of the expectation value to lie in the interval $[-T, T]$ as in Step~\ref{item:truncate} of the algorithm is still close to the true expectation value $\tr(O_i \rho)$.

In Step~\ref{item:sparse} of the algorithm, we run the sparse vector algorithm from Theorem~\ref{thm:sv} over the dataset $\mathcal{\Tilde{S}}$ and queries $s_i$.
Notice that Theorem~\ref{thm:sv} can be seen as the classical analogue of quantum threshold search, where in quantum threshold search $q_i(\mathcal{D})=\trace(O_i\rho)$. For our algorithm, we instead define $q_i(\mathcal{D})$ to be the truncated estimator $s_i:\mathcal{X}\mapsto [-T,T]$, where $\mathcal{X}$ is the space of classical shadows. Moreover, for a given query $(O_i,\theta_i)$, we instantiate the sparse vector algorithm with parameters $\epsilon/3,\delta$ and pass in query/threshold pairs $(s_i,\theta_i-\epsilon/3)$. Then, the guarantees of the sparse vector algorithm imply those of Theorem \ref{thm:nts}:
\begin{enumerate}
    \item Suppose $o_i > \theta_i$. Since $o_i = o_i - s_i + s_i \leq \epsilon/3 + s_i$, then $s_i > \theta_i - \epsilon/3$ $\Longrightarrow$ output "No".
    \item Suppose $o_i \leq \theta_i-\epsilon$. Since $o_i \geq s_i - \epsilon/3$, then $s_i \leq \theta_i - 2\epsilon/3$ $\Longrightarrow$ output "Yes".
\end{enumerate}
The final sample complexity follows from Theorem \ref{thm:sv}.
\end{proof}

\subsubsection{Extension to verifying epsilon closeness}
We include an extension of the threshold search algorithm to the problem of verifying whether $\theta_i$ is close to $\trace(O_i\rho)$. This variant will be used for our problem of predicting expectation values of observables in Appendix~\ref{app: bfo}.

\begin{corollary} \label{lem:nts}
    Consider the same inputs as Theorem \ref{thm:nts}. There is an algorithm using 
    \begin{equation}
        N=\mathcal{O}\left(\frac{B\sqrt{\ell}\log(B)\log(M)\log(1/\epsilon\delta)^{5/2}}{\epsilon^2}\right) \quad \left[B=\max_{i=1,...,M}\tr(O_i^2)\right]
    \end{equation}
    samples and with probability at least $1-\delta$ correctly implements the following for up to $\ell$ mistakes:
    \begin{enumerate}
    \item If $|\trace(O_i\rho)-\theta_i|>\epsilon$, always declare a mistake.
    \item If $|\trace(O_i\rho)-\theta_i|\leq \frac{3}{4}\epsilon$, always pass.
    \item If a mistake was declared, supply a value $\mu'_i$ satisfying $|\trace(O_i\rho)-\mu'_i|\leq \epsilon/4$.
\end{enumerate}
\end{corollary}

\begin{proof}
    The proof for implementing the verification step (i.e., conditions 1 and 2) are the same as in \cite{badescu2020improved}. For each query $(O_i,\theta_i)$, we can simply give      ``simulated inputs'' $(O_i,\theta_i+\epsilon), (\mathbb{I}-O_i,1-\theta_i+\epsilon)$ to the algorithm from Theorem \ref{thm:nts} with parameters $\epsilon/4$, $\delta/2$, $2M$, and $\ell$. With probability at least $1-\delta/2$, we get the correct outputs for the threshold search algorithm. We let the algorithm operate as follows.
\begin{enumerate}
    \item If threshold search declares a mistake on either one of the queries, i.e. 
    \begin{equation}
        \trace(O_i\rho) > \theta_i + \epsilon - \frac{1}{4}\epsilon \quad \mathrm{or} \quad \trace(\mathbb{I}-O_i\rho) > 1-\theta_i + \epsilon - \frac{1}{4}\epsilon,
    \end{equation}
    then the algorithm will correctly declare a mistake for query $(O_i,\theta_i)$.
    \item If threshold search passed on both queries, i.e. 
    \begin{equation}
        \trace(O_i\rho) \leq \theta_i + \epsilon \quad \mathrm{or} \quad \trace(\mathbb{I}-O_i\rho) \leq 1-\theta_i + \epsilon,
    \end{equation}
    then the algorithm will correctly output ``pass.''
\end{enumerate}
It can be verified that these give the correct outputs in conditions 1 and 2. The sample complexity for implementing this is given in Theorem \ref{thm:nts}.

In order to provide the correction $\mu'_i$ when mistakes are made, we can simply run the Theorem~\ref{thm:dp cs} algorithm with accuracy parameters $\epsilon/4, \delta$ in parallel whenever a mistake is made (can be run on a separate batch or on all the samples). This requires
\begin{equation}
    N=\mathcal{O}\left(\frac{B\sqrt{\ell\log(1/\delta)}\log(\ell/\delta\epsilon)}{\epsilon^2}\right) \quad [B=\max_{i=1,...,M} \tr(O_i^2)]
\end{equation}
samples. Combining this with the sample complexity in Theorem \ref{thm:nts} yields the final result.
\end{proof}

\section{Bounded-Frobenius-norm observables} \label{app: bfo}
We consider predicting expectation values of $M$ adaptively chosen observables $O_i$ with bounded-Frobenius-norm, i.e., $\max_i \trace(O_i^2) \leq B$. 
We also consider the subcase of predicting low-rank observables.
For this case, we show that we can design an algorithm for the mechanism that achieves a sample complexity upper bound that scales as $\mathcal{O}(\log(M))$. To help simplify the presentation, we first prove a sample-efficient algorithm for predicting adaptively chosen single-rank observables. Then we extend the ideas to answering adaptively chosen bounded-Frobenius-norm observables and low-rank observables.

\subsection{Single Rank Observables} \label{sec: sro}

We can obtain the following rigorous sample complexity guarantee for predicting adaptively chosen single-rank observables.

\begin{theorem}[Single rank observables] \label{thm:single rank}
    There exists an algorithm that can answer $M$ adaptively chosen single rank projectors $\{\ketbra{\psi^i}\}$ to $\epsilon$ error with probability at least $1-\delta$ using 
    \begin{equation}
        N = \mathcal{O}\left(\frac{\log(M)\log(1/\epsilon\delta)^{3/2}\log(1/\epsilon)}{\epsilon^3}\right)
    \end{equation}
    samples of $\rho$.
\end{theorem}

In particular, we designed the algorithm in Algorithm~\ref{alg:sr} for achieving this sample complexity upper bound.
In Algorithm~\ref{alg:sr}, the tomography protocol $\mathcal{P}$ can be any protocol that implements shadow tomography, for example, the original shadow tomography protocol~\cite{aaronson2018shadow} or the algorithms from Appendix~\ref{sec:adapt cs}.
Similarly, the threshold search (TS) protocol can be any protocol which solves the threshold search problem (reviewed in Appendix~\ref{sec:thresh}), e.g., the quantum threshold search algorithm~\cite{badescu2020improved} or the algorithms from Appendix~\ref{sec:thresh}.
Depending on which algorithm is used to implement the tomography protocol or threshold search, we can obtain different sample complexity guarantees.
By using the new algorithms in Appendix~\ref{app: acs}, we obtain the sample complexity in Theorem~\ref{thm:single rank}.

\begin{algorithm}
\caption{Mechanism for Predicting Single-Rank Observables}\label{alg:sr}
Set $S\leftarrow\{\}$ and initialize the tomography protocol $\mathcal{P}$.\\
Start an instance of threshold search (TS) with parameters $\epsilon, \delta$.\\
\For{$i=1$ to $M$}{
  Receive rank one observable $O=\ketbra{\psi^i} = U_i\ket{0^n}$ and construct estimate $\hat{o}_i \approx \bra{\psi^i_S}\rho\ket{\psi^i_S}$ using $\mathcal{P}$, where $\psi^i_S=\mathrm{Proj}_S(\psi^i)$.\\
  Feed $\hat{o}_i$ to TS.\\
  \uIf{\textit{TS} declares a mistake}{
  Set $S \leftarrow S \cup \{\ket{\psi^i}\}$.\\
  Receive new answer $o'_i$ from TS and return $o'_i$.\\
  Update $\mathcal{P}$ and restart TS.\\
  }\Else{
  Return $\hat{o}_i$.
  }
}
\end{algorithm}

Our algorithm follows a similar framework to that in \cite{badescu2020improved}, which is based on Aaronson's shadow tomography protocol \cite{aaronson2018shadow}. Specifically, we implement a student-teacher type interaction: the student responds to a queried observable $O$ by constructing a guess $\hat{o}_i$ for its expectation value $o_i=\trace(O_i\rho)$ and then passing the guess to the teacher. The teacher either "passes", indicating that the guess is close to the true expectation value, or he declares a "mistake" and supplies a different value $o'_i$ that is closer to the true expectation value.

In Algorithm~\ref{alg:sr}, the teacher is implemented using a threshold search algorithm.
Then, we can obtain the following guarantees for the student algorithm.

\begin{theorem} \label{thm:mistake bound}
    Suppose the following teacher properties are satisfied for all $i \in [M]$:
    \begin{itemize}
        \item If $|\trace(O_i\rho)-\hat{o}_i|>\epsilon$, the teacher declares a mistake.
        \item If $|\trace(O_i\rho)-\hat{o}_i|\leq\frac{3}{4}\epsilon$, the teacher always passes.
        \item If the teacher declares a mistake, the supplied value $o'_i$ satisfies $|\trace(O_i\rho)-o'_i|\leq \frac{1}{4}\epsilon$.
    \end{itemize}
    There is an algorithm for the student using $\Tilde{\mathcal{O}}(\log M/\epsilon^3)$ copies of $\rho$ that provides guesses $\hat{o}_i$ for $\trace(O_i\rho)$ while making at most $\mathcal{O}(1/\epsilon^2)$ mistakes, with probability at least $1-\delta$.
\end{theorem}

The algorithm for the student is essentially to project $\ket{\psi}$ to the space spanned by $S$ and guess the expectation value for the projected state. For a given observable $O=\ketbra{\psi}$, let $\ket{\psi}=\ket{\psi_S}+\ket{\psi_\perp}$, where $\ket{\psi_S}$ is the projection of $\ket{\psi}$ onto the space spanned by $S$ and $\ket{\psi_\perp}$ is the component orthogonal to the space spanned by $S$. Given $S=\{\ket{\psi_1},...,\ket{\psi_k}\}$ is spanned by an orthonormal basis $\phi_1,...,\phi_k$ at the current round, the student will estimate \begin{equation}
    s=\bra{\psi_S}\rho\ket{\psi_S}=\sum_{\substack{i\in [k]\\j\in[k]}} \alpha_i\Bar{\alpha}_j \bra{\phi_j}\rho\ket{\phi_i}
\end{equation}
where $\ket{\psi_S}=\sum_{i=1}^k \alpha_i \ket{\phi_i}$. We show that if the student can estimate $s_i$ for all observables $O_1,...,O_M$ to a suitable precision, then he will make at most $\mathcal{O}(1/\epsilon^2)$ mistakes.

\begin{lemma} \label{lem:kbound}
    Suppose $\rho$ is a density matrix, i.e. $\trace(\rho)=1$ and $\rho \succeq 0$. If $\ket{\phi_1},\ket{\phi_2},...,\ket{\phi_k}$ are orthonormal states and $\bra{\phi_i}\rho\ket{\phi_i}>\epsilon$ $\forall i$, then $k\leq \frac{1}{\epsilon}$.
\end{lemma}

\begin{proof}
    For any orthonormal basis $\{\ket{\phi_1},...,\ket{\phi_{2^n}}\}$, we have $\sum_{i=1}^{2^n} \trace(\ketbra{\phi_i}\rho) = \trace(\rho)=1$. The result follows immediately.
\end{proof}

\begin{lemma} \label{lem:mistake bound}
    Suppose the student has access to an oracle that can always produce an estimate $\hat{o}$ such that $|\bra{\psi_S}\rho\ket{\psi_S}-\hat{o}|<\frac{3}{8}\epsilon$. Then, the oracle makes at most $\frac{256}{9}\frac{1}{\epsilon^2}$ mistakes.
\end{lemma}
\begin{proof}
    Suppose the teacher declares a mistake on observable $\ketbra{\psi}$, where $\ket{\psi}$ has unit norm, i.e. $|\bra{\psi}\rho\ket{\psi}-\hat{o}|\geq \frac{3}{4}\epsilon$. By the oracle guarantees, 
    \begin{equation}
        |\bra{\psi}\rho\ket{\psi}-\hat{o}|\leq |\bra{\psi}\rho\ket{\psi}-\bra{\psi_S}\rho\ket{\psi_S}|+|\bra{\psi_S}\rho\ket{\psi_S}-\hat{o}|\leq |\bra{\psi}\rho\ket{\psi}-\bra{\psi_S}\rho\ket{\psi_S}|+\frac{3}{8}\epsilon.
    \end{equation}
    Thus, $|\bra{\psi}\rho\ket{\psi}-\bra{\psi_S}\rho\ket{\psi_S}| \geq \frac{3}{8}\epsilon$. Let $\{\ket{\phi_i}\}_{i\in[k]}$ span $S$ and $\ket{\phi_\perp}$ be the normalized state corresponding to $\ket{\psi_\perp}$, so we can write the decompositions $\ket{\psi_S}=\sum_{i=1}^k\alpha_i\ket{\phi_i}$ and $\ket{\psi_\perp}=\alpha_\perp \ket{\phi_\perp}$. We then have
    \begin{align}
        |\bra{\psi}\rho\ket{\psi}-\bra{\psi_S}\rho\ket{\psi_S}| &= |\bra{\psi_\perp}\rho\ket{\psi_\perp}+\bra{\psi_S}\rho\ket{\psi_\perp}+\bra{\psi_\perp}\rho\ket{\psi_S}|\\
        &= |\alpha_\perp \Bar{\alpha}_\perp\bra{\phi_\perp}\rho\ket{\phi_\perp}+\sum_{i=1}^k[\alpha_\perp\Bar{\alpha}_i\bra{\phi_i}\rho\ket{\phi_\perp}+\alpha_i\Bar{\alpha}_\perp\bra{\phi_\perp}\rho\ket{\phi_i}]|\\
        &\leq |\alpha_\perp|^2\bra{\phi_\perp}\rho\ket{\phi_\perp}+2\sum_i|\alpha_\perp||\alpha_i||\bra{\phi_\perp}\rho\ket{\phi_i}|
    \end{align}
    where the first and second inequality follows from triangle inequality, respectively. We can bound $|\bra{\phi_\perp}\rho\ket{\phi_i}|$ as follows:
    \begin{align}
        |\bra{\phi_\perp}\rho\ket{\phi_i}|&=|\bra{\phi_\perp}(\rho^\frac{1}{2})^\dag\rho^\frac{1}{2}\ket{\phi_i}|\\
        &=|\bra{\rho^\frac{1}{2}\phi_\perp}\ket{\rho^\frac{1}{2}\phi_i}|\\
        &\leq \sqrt{\bra{\rho^\frac{1}{2}\phi_\perp}\ket{\rho^\frac{1}{2}\phi_\perp}\bra{\rho^\frac{1}{2}\phi_i}\ket{\rho^\frac{1}{2}\phi_i}}\\
        &=\sqrt{\bra{\phi_\perp}\rho\ket{\phi_\perp}\bra{\phi_i}\rho\ket{\phi_i}},
    \end{align}
    where the inequality follows by Cauchy-Schwarz inequality.
    Thus, plugging into the previous equation, we have
    \begin{align*}
        |\bra{\psi}\rho\ket{\psi}-\bra{\psi_S}\rho\ket{\psi_S}| &\leq |\alpha_\perp|^2\bra{\phi_\perp}\rho\ket{\phi_\perp}+2\sum_i|\alpha_\perp||\alpha_i|\sqrt{\bra{\phi_\perp}\rho\ket{\phi_\perp}\bra{\phi_i}\rho\ket{\phi_i}}\\
        &=|\alpha_\perp|\sqrt{\bra{\phi_\perp}\rho\ket{\phi_\perp}}(|\alpha_\perp|\sqrt{\bra{\phi_\perp}\rho\ket{\phi_\perp}}+2\sum_i|\alpha_i|\sqrt{\bra{\phi_i}\rho\ket{\phi_i}})\\
        &\leq |\alpha_\perp|\sqrt{\bra{\phi_\perp}\rho\ket{\phi_\perp}}\sqrt{(|\alpha_\perp|^2+2\sum_i|\alpha_i|^2)(\bra{\phi_\perp}\rho\ket{\phi_\perp}+2\sum_i\bra{\phi_i}\rho\ket{\phi_i})}\\
        &\leq 2|\alpha_\perp|\sqrt{\bra{\phi_\perp}\rho\ket{\phi_\perp}}
    \end{align*}
    where the second to last inequality follows by Cauchy-Schwarz inequality, and the last inequality follows since $\trace(\rho)=1$ and $\sum_{i=1}^k|\alpha_i|^2+|\alpha_\perp|^2=1$ (i.e. $|\psi|$ has unit norm). Since $2|\alpha_\perp|\sqrt{\bra{\phi_\perp}\rho\ket{\phi_\perp}}>\frac{3}{8}\epsilon$, then $\bra{\phi_\perp}\rho\ket{\phi_\perp}>\frac{9\epsilon^2}{256|\alpha_\perp|^2}$. This implies that every time we make a mistake, we add a new component $\ket{\phi_i}$ orthonormal to the current $S$ with expectation value at least $\frac{9\epsilon^2}{256|\alpha_\perp|^2}$. By Lemma \ref{lem:kbound} and the fact that $|\alpha_\perp|\leq 1$, this can happen at most $\frac{256}{9}\frac{1}{\epsilon^2}$ times.
\end{proof}

\subsubsection{Estimating $\bra{\psi_S}\rho\ket{\psi_S}$} \label{sec:online learning}
We now incorporate the tomography protocol from Theorem~\ref{thm: tomography protocol} and show that it can be used to predict $\bra{\psi_S}\rho\ket{\psi_S}$ efficiently.

\begin{lemma} \label{lem:st}
    \begin{equation}
        N_{ST}=\mathcal{O}\left(\frac{\sqrt{\log(1/\delta)}\cdot \log(M/\epsilon\delta)\log(1/\epsilon)}{\epsilon^3}\right)
    \end{equation}
    samples suffice to estimate $s_i=\trace(\ketbra{\psi^i_S}\rho)$ within $\epsilon$ error for all $1\leq i \leq M$ with probability at least $1-\delta$.
\end{lemma}

\begin{proof}
    Suppose the size of the set $S$ is at most $k$. The problem from the tomography protocol's point of view can be framed as follows: at each round, receive a low-dimensional query $\ket{\psi_S}$, which is $\ket{\psi}$ projected onto the low-dimensional space spanned by $S$. The goal is to estimate $\bra{\psi_S}\rho\ket{\psi_S}$. We show that the full tomography problem involving the $2^n \times 2^n$ quantum state $\rho$ is equivalent to a separate tomography problem in this low-dimensional space. Specifically, we claim that there is a quantum state $\hat{\rho}$ of dimension $2^{\lceil \log (k+1) \rceil}$ such that $\bra{\psi_S}\hat{\rho}\ket{\psi_S}=\bra{\psi_S}\rho\ket{\psi_S}$. If this is the case, then as long as we express $\ket{\psi_S}$ in a basis of $S$, then we can simply run our tomography algorithm on the separate problem with the low-dimensional representation of $\ket{\psi_S}$ as input.

    Recall that the set $S$ is constructed iteratively in Algorithm~\ref{alg:sr}, where each time the teacher declares a mistake, a state is added to $S$.
    Let $S_i$ be $S$ after the $i$th mistake, and suppose that, in hindsight, there exists a fixed orthonormal basis $\Phi_i=\{\ket{\phi_1},...,\ket{\phi_i}\}=\Phi_{i-1} \cup \{\ket{\phi_i}\}$ of $S_i$ for all $1\leq i \leq k$. Such a basis can clearly be constructed via Gram-Schmidt orthonormalization.

    We can use this basis to show the claim that there exists some $\hat{\rho}$ such that $\expval{\hat{\rho}}{\psi_S} = \expval{\rho}{\psi_S}$.
    To see that the claim is true, consider the $k \times k$ sub-matrix representing $\rho$ in the basis $\Phi_k$, i.e. $\hat{\rho}_S=\{\bra{\phi_i}\rho\ket{\phi_j}\}_{1\leq i,j \leq k}$. Note that $\bra{\psi_S}\hat{\rho}_S\ket{\psi_S}=\bra{\psi_S}\rho\ket{\psi_S}$ since $\ket{\psi_S} \in \mathrm{span}(S)$. Finally, we only need to "pad" $\hat{\rho}_S$ with additional dimensions to make it a valid quantum state (i.e., unit trace, positive semi-definite, and Hermitian). Specifically, we can obtain $\hat{\rho}$ of dimension $2^{\lceil \log (k+1) \rceil}$ by adding additional dimensions where the additional diagonal elements are non-negative to make $\hat{\rho}$ trace $1$. The off-diagonal elements can be set to $0$, which would make $\hat{\rho}$ positive semi-definite and Hermitian since $\hat{\rho}_S$ is positive semi-definite and Hermitian. Moreover, $\bra{\psi_S}\hat{\rho}\ket{\psi_S}=\bra{\psi_S}\hat{\rho}_S\ket{\psi_S}$.

    Thus, we can now consider the problem of predicting $\bra{\psi_S}\hat{\rho}\ket{\psi_S}$, with $\hat{\rho}$ being a $2^{\lceil \log (k+1) \rceil}$ dimensional quantum state. This can be solved using the adaptive tomography algorithm from Theorem~\ref{thm: tomography protocol}. In this case, the number of bits needed to represent a classical shadow obtained with Clifford measurements for this hypothetical state is $m=\mathcal{O}(\lceil \log (k+1) \rceil^2)$.
    We can outline the algorithm as follows:
    \begin{enumerate}
        \item Obtain a classical shadow dataset of size
        \begin{equation}
            N_{ST}=\mathcal{O}\left(\frac{\sqrt{\log(1/\delta)}\cdot \log(M/\epsilon\delta)\log(1/\epsilon)}{\epsilon^3}\right)
        \end{equation}
        using Clifford measurements.
        \item Prepare an instance of the tomography algorithm from Theorem~\ref{thm: tomography protocol} initialized with dimension $m=\mathcal{O}(\lceil \log (k+1) \rceil^2)$. This also requires $N_{ST}$ samples.
        \item At any round $i$, express $\ket{\psi_{S_i}}$ in $i$ dimensions and feed it as input to the tomography algorithm. This ensures that in hindsight every query $\ket{\psi_{S_i}}$ was expressed in the subspace $\mathrm{span}(S_k)$ so that we have a proper reduction to the low-dimensional problem specified previously. Then return the output of the tomography algorithm.
    \end{enumerate}

    By the mistake bound of Lemma~\ref{lem:mistake bound}, $k=\mathcal{O}(1/\epsilon^2)$. Combining this with Theorem~\ref{thm: tomography protocol} yields the result of Lemma~\ref{lem:st}.
\end{proof}

With this, we can provide a simple proof of Theorem~\ref{thm:mistake bound}

\begin{proof}[Proof of Theorem~\ref{thm:mistake bound}]
   Combining Lemma \ref{lem:st} with Lemma \ref{lem:mistake bound} yields Theorem \ref{thm:mistake bound}. 
\end{proof}

Finally, we can prove our guarantee for single rank observables by combining all of our results.

\begin{proof}
    The teacher can be implemented directly by our threshold search algorithm in Theorem~\ref{thm:nts}. Combined with the mistake bound in Theorem~\ref{thm:mistake bound} yields a sample complexity of 
    \begin{equation}
        N_{TS}=\mathcal{O}\left(\frac{\log(M)\log(1/\epsilon\delta)^{5/2}}{\epsilon^3}\right)
    \end{equation}
    for implementing threshold search. Summing this with the sample complexity for implementing the tomography algorithm
    \begin{equation}
            N_{ST}=\mathcal{O}\left(\frac{\sqrt{\log(1/\delta)}\cdot \log(M/\epsilon\delta)\log(1/\epsilon)}{\epsilon^3}\right)
        \end{equation}
    yields the desired guarantee:
    \begin{equation}
        N = N_{ST} + N_{TS} = \mathcal{O}\left(\frac{\log(M)\log(1/\epsilon\delta)^{5/2}\log(1/\epsilon)}{\epsilon^3}\right).
    \end{equation}
\end{proof}

\subsection{Extension to bounded-Frobenius-norm observables} \label{sec:ebfo}

Having solved the single rank observable case, we now describe how to extend the algorithm to predicting bounded-Frobenius-norm observables. 
We obtain the following guarantees.

\begin{theorem} \label{thm: bfo ub}
    For a sequence of adaptively chosen observables $O_1,...,O_M$ satisfying $\trace(O_i^2)\leq B$ for all $i$ and accuracy parameters $\epsilon,\delta\in [0,1]$, there exists an algorithm that uses
    \begin{equation}
        N=\mathcal{O}\left(\frac{B^{2}\log(M)\log(1/\epsilon\delta)^{5/2}\log(B/\epsilon)}{\epsilon^4}\right)
    \end{equation}
    samples of $\rho$ to accurately predict every property:
    \begin{equation}
        |\hat{o}_i-\trace(O_i\rho)|\leq \epsilon \quad \textrm{for all}\ 1\leq i \leq M
    \end{equation}
    with probability at least $1-\delta$.
\end{theorem}

Thus, we are able to predict expectation values of adaptively chosen bounded-Frobenius-norm observables with the desired scaling of $\log(M)$, independent of system size.
When compared to protocols for adaptive shadow tomography~\cite{badescu2020improved}, we remove the system size dependence in exchange for an extra factor of $1/\epsilon$.

To achieve this sample complexity upper bound, we modify the algorithm from Appendix~\ref{sec: sro}.
The only modifications are (1) we define a processing step that ``projects'' the queried observable $O$ to $S$ and (2) when a mistake is made, every eigenstate of $O$ with a sufficiently large (in magnitude) eigenvalue is added to $S$. In more detail, the algorithm is given in Algorithm~\ref{alg:bf} and we outline it as follows.

\begin{algorithm}
\caption{Mechanism for bounded-Frobenius-norm observables}\label{alg:bf}
Set $S\leftarrow\{\}$ and initialize tomography protocol $\mathcal{P}$ with accuracy parameter $\frac{3}{16}\epsilon$ and failure probability $\delta/2$.\\
Start an instance of threshold search (TS) with parameters $\epsilon/2, \delta/2$.\\
\For{$i=1$ to $M$}{
  Receive truncated observable $\hat{O}_i=\sum_{j:|w_j|>\epsilon/2} w_j\ketbra{\psi^j}$.\\
  Project every eigenstate of $\hat{O}_i$ to $\mathrm{span}(S)$ and construct $\Tilde{O}_i=\sum_{j:|w_j|>\epsilon/2} w_j\ketbra{\psi^j_S}$, where $\psi^j_S=\mathrm{Proj}_S(\psi^j)$ is the projection of the eigenstate $\psi^j$ to $S$.\\
  Receive prediction $\hat{o}_i \approx \trace(\Tilde{O}_i\rho)$ from $\mathcal{P}$. Feed $\hat{o}_i$ to TS.\\
  \uIf{TS declares a mistake}{
  Set $S \leftarrow S \cup \{\ket{\psi^j}:|w_j|>\epsilon/2\}$.\\
  Receive new answer $o'_i$ from TS and return $\hat{o}_i=o'_i$.\\
  Update $\mathcal{P}$ and restart TS.\\
  }\Else{
  Return $\hat{o}_i$.
  }
}
\end{algorithm}

For a given observable $O$, consider its eigenvalue decomposition $O=\sum_{i=1}^{2^n} w_i\ketbra{\psi^i}$, where $\sum_{i=1}^{2^n} |w_i|^2 \leq B$. Let $\hat{O}=\sum_{i:|w_i|>\epsilon/2}|w_i|\ketbra{\psi^i}$ be its truncated eigenvalue decomposition, i.e. we include only eigenstates whose eigenvalues have magnitude greater than $\epsilon/2$ and throw out the rest. Note that there are at most $4B/\epsilon^2$ such eigenstates. At every round, the algorithm projects every eigenstate of the truncated observable $\hat{O}$ to $\mathrm{span}(S)$ and constructs $\Tilde{O}=P_S\hat{O}P_S=\sum_{i:|w_i|>\epsilon/2}w_i\ketbra{\psi^i_S}$, where $P_S$ is the projection operator to $S$ and $\psi^i_S$ is the projection of eigenstate $\psi^i$ to $S$.
Then it uses the tomography protocol from Theorem~\ref{thm: tomography protocol} with $m=4B/\epsilon^6$, accuracy parameter $\epsilon/2$, and failure probability $\delta/2$ to generate prediction $\Tilde{o}\approx\trace(\Tilde{O}\rho)$. If a mistake is made, every eigenstate in $\hat{O}$ is added to $S$. The teacher is implemented by the threshold search algorithm in Theorem~\ref{thm:nts} with accuracy parameter $\epsilon$ and failure probability $\delta/2$.

With this sketch of the algorithm, we can prove Theorem~\ref{thm: bfo ub}, relying on some results from the previous section.

\begin{proof}[Proof of Theorem~\ref{thm: bfo ub}]
    The probability that either threshold search or the tomography protocol fails is at most $\delta$ by union bound. Thus, for the rest of the proof, we assume that both protocols correctly give outputs across all $M$ rounds. Hence, every $\hat{o}$ received from the shadow tomography protocol satisfies $|\hat{o}-\trace(\Tilde{O}\rho)| \leq \frac{3}{16}\epsilon$. Now suppose QTS declares a mistake, so we have $|\hat{o}-\trace(\hat{O}\rho)|>3\epsilon/8$. Using triangle inequality:
    \begin{equation}
        \frac{3}{8}\epsilon < |\hat{o}-\trace(\hat{O}\rho)| \leq |\hat{o}-\trace(\Tilde{O}\rho)|+|\trace(\Tilde{O}\rho)-\trace(\hat{O}\rho)| \leq \frac{3}{16}\epsilon + |\trace(\Tilde{O}\rho)-\trace(\hat{O}\rho)|
    \end{equation}
    Thus, $|\trace(\Tilde{O}\rho)-\trace(\hat{O}\rho)|>3\epsilon/16$. We claim that there is at least one eigenstate $\ket{\psi}$ in the truncated observable $\hat{O}$ such that 
    \begin{equation} \label{eq:eig lb}
        |\bra{\psi_S}\rho\ket{\psi_S}-\bra{\psi}\rho\ket{\psi}|>\frac{3\epsilon^2}{32B}.
    \end{equation}
    We can see this by noting that if this is not true for any eigenstate of $\hat{O}$, then this gives us a contradiction:
    \begin{align}
        \left|\trace(\Tilde{O}\rho)-\trace(\hat{O}\rho)\right| &= \left|\sum_{i:|w_i|>\epsilon/2} w_i(\bra{\psi^i_S}\rho\ket{\psi^i_S}-\bra{\psi^i}\rho\ket{\psi^i})\right|\\
        &\leq \sum_{i:|w_i|>\epsilon/2} |w_i|\left|\bra{\psi^i_S}\rho\ket{\psi^i_S}-\bra{\psi^i}\rho\ket{\psi^i}\right|\\
        &\leq \frac{3\epsilon^2}{32B} \sum_{i:|w_i|>\epsilon/2} |w_i|\\
        &\leq \frac{3}{16}\epsilon,
    \end{align}
    where the first inequality is by triangle inequality, the second is by our assumption, and the third by Cauchy-Schwarz inequality and the Frobenius-norm bound.

    Since a mistake was declared, then every eigenstate in $\hat{O}$ gets added to $S$. We added at least one eigenstate such that $|\bra{\psi_S}\rho\ket{\psi_S}-\bra{\psi}\rho\ket{\psi}|>3\epsilon^2/32B$, so by the mistake bound in Lemma \ref{lem:mistake bound}, a mistake can happen at most $\mathcal{O}(B^2/\epsilon^4)$ times. To obtain the $m$ parameter in the tomography protocol from Theorem~\ref{thm: tomography protocol}, note that every time a mistake is made we add at most $4B/\epsilon^2$ states to $S$. Combining this with the mistake bound, $|S|\leq \mathcal{O}(B^3/\epsilon^6)$, so we set $m=\lceil\log(\mathcal{O}(B^3/\epsilon^6)+1)\rceil^2 = \mathcal{O}(\log(B/\epsilon)^2)$. Moreover, we emphasize that the Frobenius bound still holds for each $\Tilde{O}=P_S\hat{O}P_S$: using the fact that $P_S^2=P_S$ and $I-P_S$ is a projector,
\begin{equation}
    \tr(A^2)-\tr((P_SA)^2) = \tr(A^2)-\tr(AP_SA) = \tr(A(I-P_S)A) = \tr(((I-P_S)A)^2) \geq 0.
\end{equation}
Thus, $\|P_SA\|_F \leq \|A\|_F$. It can similarly be shown that $\|AP_S\|_F \leq \|A\|_F$. 
In sum,
\begin{equation}
    \|P_S\hat{O}P_S\|_F \leq \|\hat{O}\|_F \leq \|O\|_F \leq B.
\end{equation}
Putting these together, the sample complexity of the tomography protocol predicting the projected observables $\Tilde{O}_1,...,\Tilde{O}_M$ is
    \begin{equation}
        N_{ST}=\mathcal{O}\left(\frac{B^{3/2}\sqrt{\log(1/\delta)}\cdot \log(M/\epsilon\delta)\log(B/\epsilon)}{\epsilon^3}\right),
    \end{equation}
    by Theorem~\ref{thm: tomography protocol}.
    The sample complexity of the threshold search algorithm verifying the output of the tomography protocol is
    \begin{equation}
        N_{TS}=\mathcal{O}\left(\frac{B^{2}\log(M)\log(1/\epsilon\delta)^{5/2}}{\epsilon^4}\right),
    \end{equation}
    by Theorem~\ref{thm:nts}.
    Combining these two bounds yields the final result, which guarantees that for every $\hat{o}_i$ returned by the mechanism, $|\hat{o}_i-\trace(\hat{O}_i\rho)|\leq \epsilon/2$ for all $1\leq i \leq M$ with probability at least $1-\delta$. Finally, this implies that the total error is bounded by $\epsilon$:
    \begin{align}
        |\hat{o}_i-\trace(O_i\rho)| &\leq |\hat{o}_i-\trace(\hat{O}_i\rho)|+|\trace(\hat{O}_i\rho)-\trace(O_i\rho)|\\
        &\leq \epsilon/2 + \sum_{j:|w_j|<\epsilon/2}|w_j| \bra{\psi^j}\rho\ket{\psi^j}\\
        &\leq \frac{\epsilon}{2}+\frac{\epsilon}{2}\sum_{j:|w_j|<\epsilon/2} \bra{\psi^j}\rho\ket{\psi^j}\\
        &\leq \frac{\epsilon}{2}+\frac{\epsilon}{2}\\
        &=\epsilon.
    \end{align}
    Thus, the total sample complexity is given by
    \begin{equation}
        N = N_{ST} + N_{TS} = \mathcal{O}\left(\frac{B^{2}\log(M)\log(1/\epsilon\delta)^{5/2}\log(B/\epsilon)}{\epsilon^4}\right)
    \end{equation}
\end{proof}

\subsection{Extension to low-rank observables} \label{sec:lr ub}

We also consider the subcase of predicting expectation values of low-rank observables.
In this case, we can again achieve a sample complexity upper bound scaling logarithmically in the number of observables $M$, independent of system size.
Moreover, we obtain a strict improvement over the best-known sample complexity for shadow tomography, given in \cite{badescu2020improved}.
This can be done by directly applying Algorithm \ref{alg:bf}, with the slight modification that it is not necessary to do the truncation step in this case.

\begin{theorem}
    For a sequence of adaptively chosen observables $O_1,...,O_M$ with rank at most $R$ and accuracy parameters $\epsilon,\delta\in [0,1]$, there exists an algorithm that uses
    \begin{equation}
        N=\mathcal{O}\left(\frac{R^{2}\log(M)\log(1/\epsilon\delta)^{5/2}\log(R/\epsilon)}{\epsilon^3}\right)
    \end{equation}
    samples of $\rho$ to accurately predict every property:
    \begin{equation}
        |\hat{o}_i-\trace(O_i\rho)|\leq \epsilon \quad \textrm{for all}\ 1\leq i \leq M
    \end{equation}
    with probability at least $1-\delta$.
\end{theorem}

\begin{proof}
    The analysis is the exact same as that for Theorem \ref{thm: bfo ub}, except that in Equation \eqref{eq:eig lb}, the guarantee is instead there is at least one eigenstate $\ket{\psi}$ in the observable $\hat{O}$ such that
\begin{equation}
    |\bra{\psi_S}\rho\ket{\psi_S}-\bra{\psi}\rho\ket{\psi}|>\frac{3\epsilon}{16R},
\end{equation}
which follows because $\sum_i |w_i| \leq R$. In this case, the new mistake bound is $\mathcal{O}(R^2/\epsilon^2)$.
The new sample complexity of the threshold search algorithm is 
\begin{equation}
    N_{TS}=\mathcal{O}\left(\frac{R^{2}\log(M)\log(1/\epsilon\delta)^{5/2}}{\epsilon^3}\right),
\end{equation}
by Theorem~\ref{thm: tomography protocol}.
Moreover, whenever a mistake is made, we add at most $R$ eigenstates to $S$, so $|S|\leq \mathcal{O}(R^3/\epsilon^2)$.
Moreover, the rank bound is preserved for $\Tilde{O}$, i.e. $\mathrm{rank}(\Tilde{O})=\mathrm{rank}(P_S\hat{O}P_S) \leq R$, due to the following property: for any matrices (with dimensions properly aligned) $A$ and $B$, $\mathrm{rank}(AB) \leq \min\{\mathrm{rank}(A),\mathrm{rank}(B)\}$.
So the sample complexity of the tomography protocol is
\begin{equation}
    N_{ST} = \mathcal{O}\left(\frac{R^{3/2}\sqrt{\log(1/\delta)} \log(M/\epsilon\delta)\log(R/\epsilon)}{\epsilon^3}\right),
\end{equation}
by Theorem~\ref{thm:nts}.
The result follows from combining these bounds:
\begin{equation}
    N = N_{ST} + N_{TS} = \mathcal{O}\left(\frac{R^{2}\log(M)\log(1/\epsilon\delta)^{5/2}\log(R/\epsilon)}{\epsilon^3}\right).
\end{equation}
\end{proof}

\begin{remark} \label{rmk: substitute}
    In order to implement the threshold search and tomography components, one could have also used existing algorithms used in shadow tomography, e.g. the quantum threshold search algorithm \cite{badescu2020improved} and Aaronson's shadow tomography algorithm \cite{aaronson2018shadow,badescu2020improved}. By tracing through the same arguments, we would have also achieved a system-size independent upper bound. However, we found that the scaling was worse, where the sample complexity was $\Tilde{\mathcal{O}}(B^2\log^2(M)/\epsilon^6)$ for bounded-Frobenius-norm observables and $\Tilde{\mathcal{O}}(R^2\log^2(M)/\epsilon^4)$ for low-rank observables.
\end{remark}

\end{document}